\documentclass[10pt,a4paper]{article} 

\usepackage{jheppub} 
\usepackage{natbib}

\usepackage[utf8]{inputenc}
\usepackage[english]{babel}
\usepackage{microtype} 
\usepackage{tikzsymbols}
\usetikzlibrary{external}
\usepackage{svg}
\svgpath{{figures/}}
\usepackage{amsmath,amssymb,amsfonts,amsthm,mathtools,bm}

\newtheorem{remark}{Remark}
\newtheorem{theorem}{Theorem}
\newtheorem*{theorem*}{Theorem}

\newtheorem{lemma}{Lemma}

\allowdisplaybreaks 

\usepackage{graphicx}  
\usepackage{xcolor}    
\usepackage{float}     
\usepackage{booktabs}
\usepackage{longtable}
\usepackage{enumitem}
\usepackage{tikz}
\usepackage{pgfplots}
\usepackage{svg}
\usepackage{subcaption}
\usepackage[dvipsnames]{xcolor}
\usepackage[ruled,vlined]{algorithm2e}
\usepackage[noend]{algpseudocode}

\usepackage{setspace}
\title{A Federated Many-to-One Hopfield model for associative Neural Networks}

\author[a,1,2]{Andrea Alessandrelli}%
\author[a]{, Fabrizio Durante}%
\author[b, 2]{, Andrea Ladiana}%
\author[c, 2]{, Andrea Lepre}%

\affiliation[a]{Dipartimento di Matematica e Fisica, Università del Salento, Italy}
\affiliation[b]{Dipartimento di Scienze di Base e Applicazioni all'Ingegneria, Sapienza Università di Roma, Italy}
\affiliation[c]{Dipartimento di Matematica, Sapienza Università di Roma, Italy}
\affiliation[1]{Istituto Nazionale di Fisica Nucleare (INFN), Sezione di Lecce, Italy}
\affiliation[2]{Istituto Nazionale di Alta Matematica Francesco Severi (INdAM), Roma, Italy}

\abstract{
Federated learning enables collaborative training without sharing raw data, but struggles under client heterogeneity and streaming distribution shifts, where drift and novel data can impair convergence and cause forgetting. We propose a federated associative-memory framework that learns shared archetypes in heterogeneous, continual settings, where client data are independent but not necessarily balanced. Each client encodes its experience as a low-rank Hebbian operator, sent to a central server for aggregation and factorization into global archetypes. This approach preserves privacy, avoids centralized replay buffers, and is robust to small, noisy, or evolving datasets.
We cast aggregation as a low-rank-plus-noise spectral inference problem, deriving theoretical thresholds for detectability and retrieval robustness. An entropy-based controller balances stability and plasticity in streaming regimes. Experiments with heterogeneous clients, drift, and novelty show improved global archetype reconstruction and associative retrieval, supporting the spectral view of federated consolidation.
}


\begin{document}

\maketitle


\section{Introduction}
\label{sec:introduction}

Federated Learning (FL) enables collaborative learning across multiple data owners without centralizing raw data, by iterating local updates and server-side aggregation.
This paradigm is attractive in privacy- and governance-constrained domains, but it is notoriously challenged by \emph{statistical heterogeneity}: client data can be reasonably modeled as independent across clients but, within a client, they are often scarce and unbalanced.
Such heterogeneity may induce biased or high-variance updates, slows convergence, and can yield unstable global models \citep{mcmahan2017fedavg,kairouz2021advances,hsu2019measuring,li2020fedprox}.
These issues are exacerbated in streaming or open-world settings, where new patterns appear over time and forgetting becomes a primary failure mode.

A complementary line of work studies \emph{memory systems} and continual adaptation.
In neuroscience-inspired accounts such as Complementary Learning Systems (CLS), fast episodic storage and slow integrative learning jointly support generalization under distribution shift \citep{o2014complementary,kumaran2016learning}.
In machine learning, continual learning methods seek to mitigate catastrophic forgetting by regulating the stability--plasticity trade-off (e.g.\ via regularization or replay), yet most approaches assume centralized access to data or to curated buffers.
In FL, the stability--plasticity dilemma becomes \emph{distributed}: even under independence, heterogeneous client marginals and local streaming dynamics interact with drift and novelty, while privacy constraints limit global rehearsal.

In this context, \emph{associative memories} provide a principled abstraction for storing and retrieving prototypical patterns from partial or noisy cues.
Classical Hopfield-type models define content-addressable retrieval as energy minimization, with well-studied capacity and robustness properties \citep{hopfield1982Neural,Amit1985SpinGlass,personnaz1985information}.
Beyond their original formulation, modern Hopfield networks have been linked to attention-like updates and large-capacity retrieval rules, renewing interest in associative mechanisms as building blocks for representation learning \citep{ramsauer2020hopfield}.
However, deploying associative memories in federated regimes raises new questions: which information should be stored locally, what should be aggregated, and how can a server infer \emph{global archetypes} from heterogeneous (yet independent) local memories without accessing raw samples?

This work addresses these questions by proposing a \emph{federated many-to-one associative memory} that learns and maintains a shared set of archetypes under heterogeneous, streaming (i.n.i.d.) regimes.
Each client compresses its local experience into a low-rank Hebbian operator, while the server aggregates these operators and factorizes the resulting global memory to infer archetypes.
Our approach is motivated by two observations.
First, storing \emph{operators} instead of data enables privacy-preserving communication of sufficient statistics that remain meaningful even when local datasets are small or drifting.
Second, archetype discovery can be posed as a spectral inference problem: the aggregated operator concentrates around a low-rank signal plus noise under independence assumptions, so random-matrix tools can predict when global archetypes become detectable and stably retrievable \citep{baik2005phase,benaych2011eigenvalues}.

We summarize our main contributions as follows:
\begin{itemize}
    \item \textbf{Model.} We introduce a federated associative-memory framework where clients transmit low-rank Hebbian operators and the server reconstructs a global archetype set via a scalable factorization procedure (Sec.~\ref{sec:setting}).
    \item \textbf{Theoretical findings.} We provide a random-matrix characterization of the aggregation-and-factorization step, highlighting detectability thresholds and retrieval robustness as functions of heterogeneity, client quality, and sample size. Moreover, we propose an entropy-based controller to balance stability and plasticity in streaming conditions, regulating when to incorporate novel patterns versus consolidating existing archetypes (Sec.~\ref{sec:theory}).
    \item \textbf{Numerical findings.} We evaluate the pipeline under client drift and novelty, demonstrating improved global retrieval and archetype reconstruction on structured datasets (Sec.~\ref{sec:numerics}).
\end{itemize}

The remainder of the paper is organized as follows:
we formalize the model, the federated setting and we describe the implementation and reconstruction pipeline in Sec.~\ref{sec:setting};
we develop the theoretical analysis in Sec.~\ref{sec:theory};
and we present experiments and ablations in Sec.~\ref{sec:numerics}.

\section{Problem setting and notation}
\label{sec:setting}

We now introduce the problem setting and notation, and formalize the federated framework studied in this work. Our goal is to exploit federated interactions among multiple associative memories to reconstruct latent \emph{archetypes} from noisy observations, in the realistic regime where data are partitioned across clients according to heterogeneous (non-i.i.d.) distributions. Before presenting the federated pipeline, we recall the basic building blocks---Hopfield-type associative memories and their heteroassociative multilayer extension---and we specify the generative model used throughout.

\subsection{Associative and heteroassociative neural networks}
Central to our method are Hopfield-like associative memories. We briefly recall the classical formulation, which will serve as a local learner at each client, and then introduce the multilayer associative memory (LAM) architecture used by the server.

\paragraph{Hopfield associative memory (unsupervised).}
Consider a system of $N$ binary neurons with configuration $\boldsymbol{\sigma}=(\sigma_1,\ldots,\sigma_N)\in\{-1,+1\}^N$. Neurons interact pairwise through a synaptic matrix $\bm J^{\mathrm{Hebb}}=\{J^{\mathrm{Hebb}}_{ij}\}_{i,j=1}^N$, and the energy (Hamiltonian) of a configuration is
\begin{equation}
\label{eq:hamiltonian}
\mathcal{H}(\boldsymbol{\sigma}\mid \boldsymbol{\xi})
= - \sum_{i,j=1}^N J^{\mathrm{Hebb}}_{ij}\,\sigma_i\sigma_j .
\end{equation}
Low-energy states correspond to internally consistent configurations, and standard Hopfield dynamics tends to drive the system toward local minima of~\eqref{eq:hamiltonian}.

The couplings are designed to store a set of $K\ge 2$ binary patterns $\{\boldsymbol{\xi}^\mu\}_{\mu=1}^K$, with $\boldsymbol{\xi}^\mu=(\xi_1^\mu,\ldots,\xi_N^\mu)\in\{-1,+1\}^N$. Under Hebb's prescription,
\begin{equation}
\label{eq:hebb}
J^{\mathrm{Hebb}}_{ij}
=\frac{1}{N}\sum_{\mu=1}^K \xi_i^\mu \xi_j^\mu.
\end{equation}
Throughout, we assume the stored patterns are i.i.d.\ Rademacher random variables, i.e.
\begin{equation}
\label{eq:rademacher}
\mathbb{P}(\xi_i^\mu=+1)=\mathbb{P}(\xi_i^\mu=-1)=\frac{1}{2},
\qquad i=1,\ldots,N,\ \mu=1,\ldots,K .
\end{equation}
In appropriate regimes of $K$ versus $N$, each $\boldsymbol{\xi}^\mu$ becomes (meta)stable in the energy landscape. Consequently, initializing the system from a sufficiently close corrupted version of $\boldsymbol{\xi}^\mu$ typically leads the dynamics to converge back to the corresponding prototype. This retrieval property is the hallmark of Hopfield networks as canonical models of associative memory (see Fig.~\ref{fig:hopfield-toy-model} for an illustrative sketch).

\begin{figure}[t!]
    \centering
    \includegraphics[width=8cm]{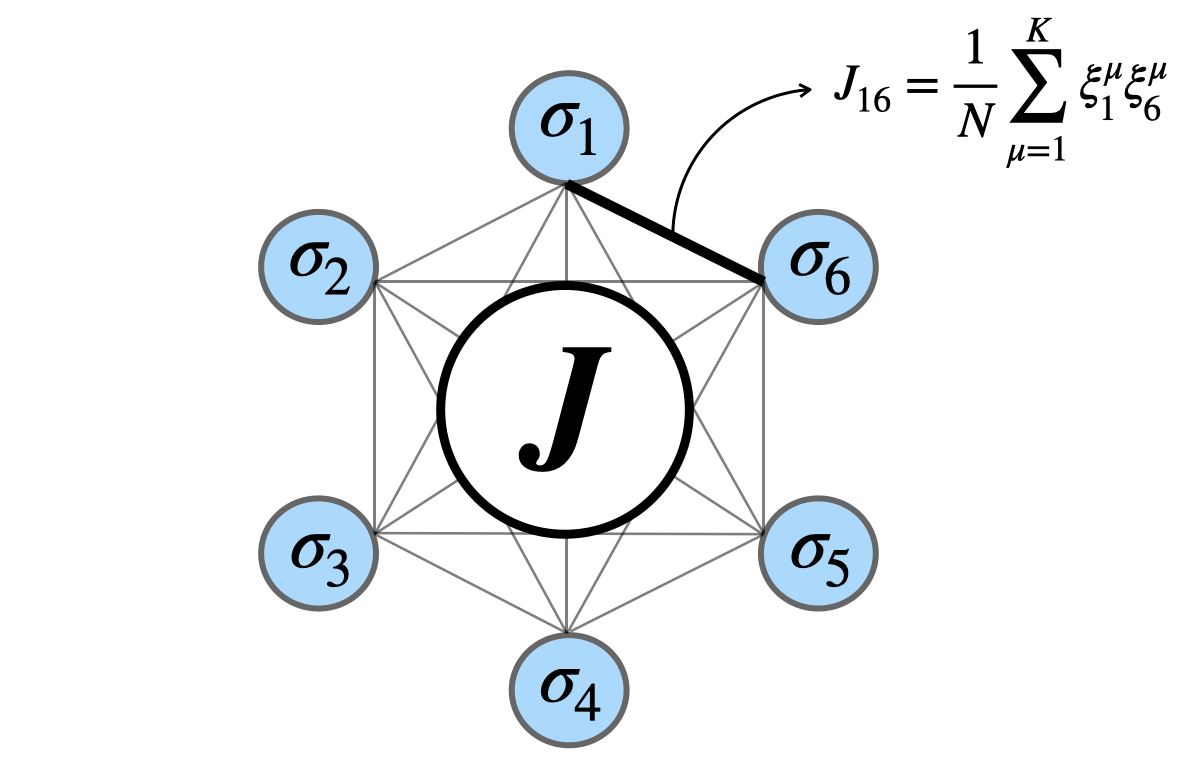}
    \caption{Illustration of a toy Hopfield model with six neurons. It is a fully connected graph where all the edges are determined by the Hebbian rule encoded in the $\bm{J}$ matrix. In particular the figure shows the connection between $\sigma_1$ and $\sigma_6$, which is given, as in Eq.~\eqref{eq:hebb}, by $N^{-1}\sum_{\mu=1}^K\xi^\mu_1\xi^\mu_6$.}
    \label{fig:hopfield-toy-model}
\end{figure}

\paragraph{Archetypes, examples, and client-side learning.}
In the federated setting, we refer to the latent patterns $\{\boldsymbol{\xi}^\mu\}_{\mu=1}^K$ as \emph{archetypes}, namely prototypical data-generating anchors. Each client observes only \emph{examples}, i.e.\ noisy and partial realizations of the archetypes, and aims to learn from its local (unlabeled) dataset.

The examples are supposed to be generated by a signal-plus-noise scenario, that is by selecting an archetype index $\mu$ and, then, by adding a corrupted realization $\boldsymbol{\eta}\in\{-1,+1\}^N$ via independent bit flips:
\begin{equation}
\label{eq:gen}
\mathbb{P}\!\left(\eta_i=\xi^\mu_i \mid \mu\right)=\tfrac{1+r}{2},
\qquad
\mathbb{P}\!\left(\eta_i=-\xi^\mu_i \mid \mu\right)=\tfrac{1-r}{2},
\end{equation}
where $r\in(0,1]$ quantifies the \emph{dataset quality}. As $r\to 1$, examples coincide with their generating archetype; as $r\to 0$, examples become essentially uninformative (nearly unpredictable from the archetypes).

Furthermore,  we consider $L$ clients. Client $c$ stores $M_c$ examples and locally builds a Hopfield-like synaptic matrix using the Hebbian rule
\begin{equation}
\label{eq:hebb_examples}
\big(J_c\big)_{ij}
=\frac{1}{N\,M_c}\sum_{a=1}^{M_c}\big(\eta_{c}\big)^a_i\,\big(\eta_{c}\big)^a_j,
\qquad c=1,\ldots,L,
\end{equation}
where the total dataset size is $M=\sum_{c=1}^L M_c$. Even though the data are assumed to be mutually independent, their random allocation across clients implies that a given client may not have access to all archetypes, but only to a strict subset of them.

Clients do not communicate directly. Coordination is mediated solely by a central server responsible for aggregating synaptic-level information and reconstructing the underlying archetypes. To this end, the server is equipped with a heteroassociative generalization of Hopfield networks, namely the \emph{$L$-layer Associative Memory} (LAM) model introduced in~\cite{agliari2025networks}. A key feature of LAM is that it can operate directly at the synaptic level, enabling reconstruction from interaction matrices alone~\cite{Agliari2025MultiChannel}.

\begin{figure}[t!]
    \centering
    \includegraphics[width=14cm]{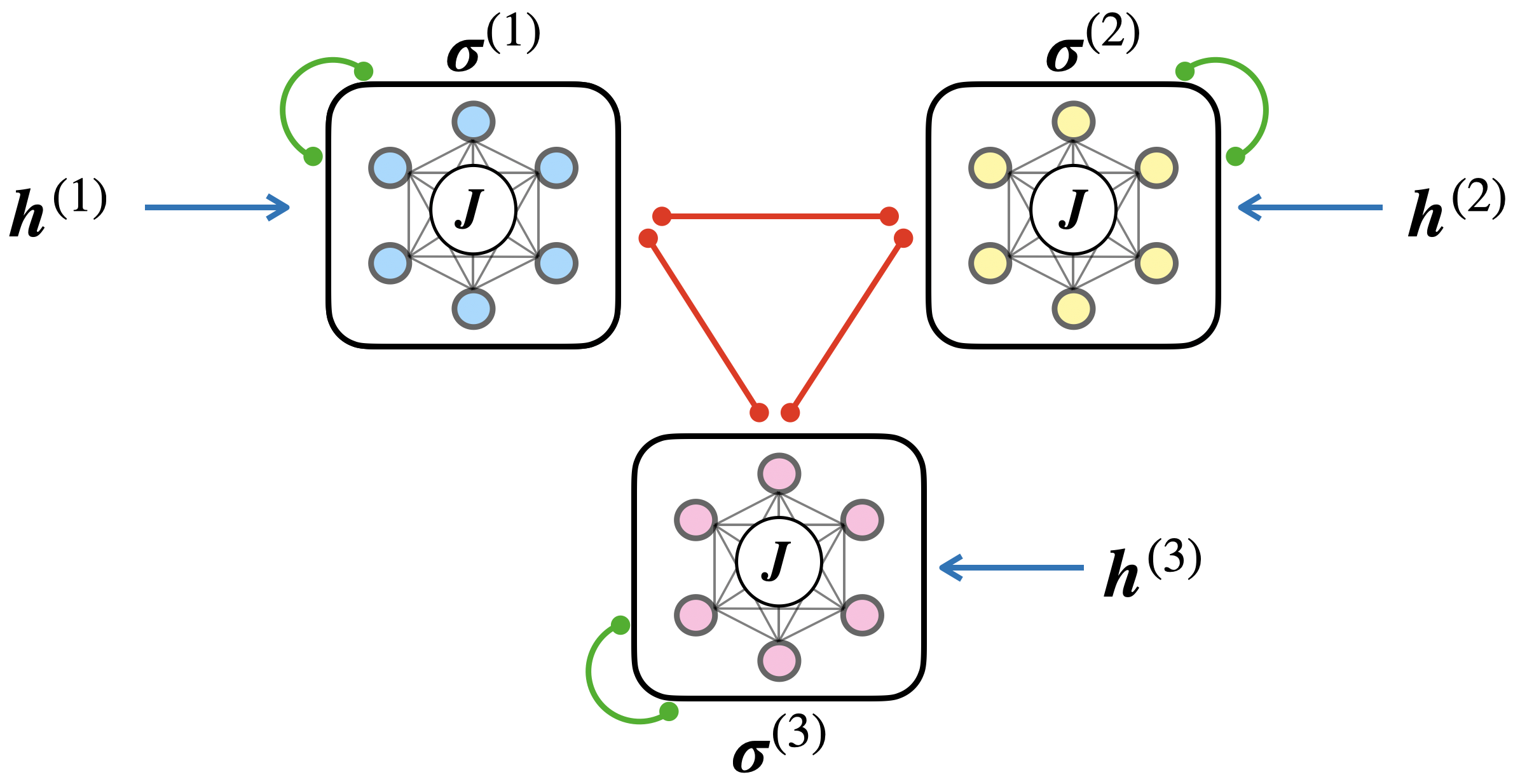}
    \caption{Schematic illustration of the model for $L=3$ layers. Each layer $a\in \{1,2,3\}$ is a Hopfield network (with state $\boldsymbol{\sigma}^{(a)}$) and all layers share the same synaptic coupling matrix. The three contributions to the Hamiltonian in~\eqref{eq:lam_hamiltonian} are highlighted: imitative intra-layer interactions (green self-loop), anti-imitative inter-layer interactions (red links), and the coupling to an external field $\boldsymbol{h}^{(a)}$ (blue arrow).}
    \label{fig:tam-arch}
\end{figure}

\paragraph{LAM: an $L$-layer associative memory.}
We now briefly recall the LAM model, in order to fix notation and emphasize the mechanism that will be exploited by the server.

Consider $L$ layers, each composed of $N$ binary neurons. The configuration of layer $a$ is denoted by $\boldsymbol{\sigma}^{(a)}=(\sigma^{(a)}_1,\ldots,\sigma^{(a)}_N)\in\{-1,+1\}^N$ for $a=1,\ldots,L$. Layers interact through a shared synaptic matrix $\bm J^{\mathrm{Hebb}}$ and through a coupling matrix $\bm g=\{g_{ab}\}_{a,b=1}^L\in\mathbb{R}^{L\times L}$ controlling the intensity and polarity (imitative vs.\ anti-imitative) of intra- and inter-layer interactions. The Hamiltonian takes the form
\begin{equation}
\label{eq:lam_hamiltonian}
\mathcal{H}(\boldsymbol{\sigma}\mid \bm g,\bm J,H,\bm h)
=
- \sum_{a,b=1}^L g_{ab}\sum_{i,j=1}^N \sigma_i^{(a)}\, J^{\mathrm{Hebb}}_{ij}\, \sigma_j^{(b)}
- H \sum_{a=1}^{L}\sum_{i=1}^{N} h_i^{(a)}\,\sigma_i^{(a)}.
\end{equation}
The last term accounts for a layer-dependent external field $\bm h^{(a)}$, modulated by a global strength $H$, which can be used to drive the dynamics.

The sign of $g_{ab}$ determines whether layers tend to align or anti-align: if $g_{ab}>0$ (resp.\ $g_{ab}<0$), configurations that maximize (resp.\ minimize) the overlap $\boldsymbol{\sigma}^{(a)}\!\cdot\!\boldsymbol{\sigma}^{(b)}$ are energetically favored. In the following we focus on the simple yet expressive choice
\begin{equation}
\label{eq:g_struct}
g_{ab}=
\begin{cases}
1, & a=b,\\
-\lambda, & a\neq b,
\end{cases}
\end{equation}
where $\lambda\in\big[0,(L-1)^{-1}\big)$ guarantees that $\bm g$ is positive definite, so that $\bm g$ is an equicorrelation matrix. This structure yields imitative Hopfield-like interactions within each layer, promoting coherent retrieval states, while inducing anti-imitative couplings across distinct layers, discouraging configurations in which all layers retrieve the same pattern (see Fig.~\ref{fig:tam-arch}).

This bias matches our target task: starting from synaptic-level information---namely, an interaction matrix assumed to be Hebbian---the model promotes a \emph{factorization} of the stored content across layers, ultimately enabling reconstruction of the underlying archetypes.

\subsubsection{Factorization procedure}
\label{subsec:factorization_procedure}

We now introduce the factorization procedure embedded in the LAM dynamics, which constitutes the core mechanism enabling the reconstruction of latent patterns from an observed synaptic matrix.

The LAM model~\eqref{eq:lam_hamiltonian}, for a given realization of hidden patterns $\{\bm\xi^\mu\}_{\mu=1}^K$, can disentangle mixed configurations such as classical spurious states. For instance, when $L=3$ one may consider the mixture $\bm x \;=\; \operatorname{sgn}\!\Big(\sum_{\ell=1}^{3}\bm\xi^\ell\Big)$, where $\operatorname{sgn}(\cdot)$ acts componentwise. Feeding $\bm x$ as input to each layer, the dynamics can relax to
\[
(\bm\sigma^{(1)},\bm\sigma^{(2)},\bm\sigma^{(3)})=(\bm\xi^1,\bm\xi^2,\bm\xi^3)
\quad\text{(up to permutations)},
\]
thus recovering the constituents of the mixture~\cite{agliari2025networks}.

In~\cite{Agliari2025MultiChannel}, it is further shown that a closely related architecture---obtained by modifying the heteroassociative component---can tackle more challenging reconstruction tasks. Specifically, consider the Hamiltonian
\begin{equation}
\label{eq:HamHam_quad_only_J}
\begin{aligned}
\widetilde{\mathcal{H}}\!\left(\bm\sigma;\lambda,H,\bm J^{\mathrm{Hebb}},\bm h\right)
&= - \sum_{a=1}^{L}\sum_{i,j=1}^{N} J^{\mathrm{Hebb}}_{ij}\,\sigma^{(a)}_i \sigma^{(a)}_j
\\
&\quad + \frac{\lambda}{N}\sum_{\substack{a,b=1\\ a\neq b}}^{L}\sum_{i,j,k,l=1}^{N}
J^{\mathrm{Hebb}}_{ij}\,J^{\mathrm{Hebb}}_{kl}\,\sigma^{(a)}_{i}\sigma^{(b)}_{j}\sigma^{(a)}_{k}\sigma^{(b)}_{l}
\\
&\quad - H \sum_{a=1}^{L}\sum_{i=1}^{N} h^{(a)}_i \sigma^{(a)}_i \,.
\end{aligned}
\end{equation}

Exploiting the mean-field structure of the model, the Hamiltonian \eqref{eq:HamHam_quad_only_J} can be conveniently recast as $\mathcal E_{N,\bm\Xi}(\bm{\sigma})=
- \sum_{a=1}^{L}\sum_{i=1}^{N} \hat{h}_i^{(a)}\sigma_i^{(a)} $
with the effective field $\hat {\bm h}^{(a)}$ acting on neurons in the layer $a$ being the sum of three contributions:
\begin{equation}
\label{eq:net_field}
\hat{\bm h}^{a}(\bm\sigma)
=
\bm h^{a\to a}(\bm\sigma)
+
\sum_{\substack{b=1 \\ b\neq a}}^L\bm h^{b\to a}(\bm\sigma)
+ 
H\,\bm h^{(a)},
\end{equation}
respectively, the intra-module (auto-associative), inter-module (anti-imitative) and external fields. The definitions of the first two contributions come directly from the expression \eqref{eq:HamHam_quad_only_J}, that is
\begin{eqnarray}
    \bm h^{a\to a}(\bm\sigma) &=& \bm J^{\mathrm{Hebb}}\cdot \bm\sigma^{(a)},\label{eq:intra}\\
    \bm h^{b\to a}(\bm\sigma)&=&
-\frac{\lambda}{N}
\big(\bm J^{\mathrm{Hebb}}\cdot \bm\sigma^{(b)}\big)\big((\bm\sigma^{(b)} )^T\bm J^{\mathrm{Hebb}}\bm\sigma^{(a)}\big).\label{eq:inter}
\end{eqnarray}
\par\medskip
We now set up a dynamics that evolves configurations toward lower-energy states. Allowing for stochastic noise, controlled by the inverse temperature $\beta\in\mathbb{R}_+$, the neuronal configuration $\bm\sigma(t)$ at (discrete) time $t$ is updated synchronously according to
\begin{equation}
\label{eq:update_MC_iclr}
\bm\sigma^{(a)}(t{+}1)
=
\operatorname{sign}\Big[\tanh\Big(\beta\hat{\bm h}^{(a)}(\bm\sigma(t))\Big) + \bm u^{(a)}(t)\Big],
\end{equation}
where $\hat{\bm h}^{(a)}$ is the $a$-th layer effective field, computed at each time-step $t$ according to Eqs. (\ref{eq:net_field}-\ref{eq:inter}), and 
$\bm u^{(a)}(t)\underset{i.i.d}{\sim} \mathcal{U}([-1,1]^N)$ for all $a=1,\dots,L$.

From now on, we assume that the Hebbian kernel $\bm J^{\mathrm{Hebb}}$ is known (or can be reliably recovered from data), while its factorization in terms of the underlying patterns $\{\bm\xi^\mu\}$ is \emph{unknown}. Our goal is to reconstruct the hidden patterns using only $\bm J^{\mathrm{Hebb}}$ (and, when available, a set of mixed inputs).

\paragraph{Candidate generation from mixed inputs.}
In~\cite{Agliari2025MultiChannel}, the model~\eqref{eq:HamHam_quad_only_J} is studied in a setting where one has access to the synaptic matrix and to $m$ mixed inputs of the form
\begin{equation}
\label{eq:mixtures}
\bm x^\alpha \;=\; \operatorname{sgn}\!\left[
f\!\left(\sum_{\mu=1}^{K} c_\mu^\alpha\,\bm\xi^\mu\right)\right],
\qquad
\alpha=1,\dots,m,
\end{equation}
where $f(\cdot)$ is a componentwise nonlinearity (possibly the identity) and the coefficients $\{c_\mu^\alpha\}$ depend on the specific scenario.\footnote{Equivalently, $x_i^\alpha=\operatorname{sgn}\big(f(\sum_{\mu=1}^K c_\mu^\alpha \xi_i^\mu)\big)$ defines a nonlinear random mapping from the $K$-dimensional feature vector $(\xi_i^1,\dots,\xi_i^K)$ to $m$ outputs; one may interpret it as a perceptron-like transformation, with the site index $i$ labeling datapoints.}

For each $\alpha$, we clamp the external fields to the same input across layers, i.e.\ $\bm h^{(a)}=\bm x^\alpha$ for all $a=1,\dots,L$, and let the dynamics \eqref{eq:update_MC_iclr} evolve to stationarity. Collecting the final configurations across all layers and repeating over the $m$ mixtures yields $Lm$ candidates $\{\bar{\bm\sigma}^{\,\ell}\}_{\ell=1}^{Lm}$. Unlike approaches enforcing $L\ge K$ in a single run, here one can trade layers for repetitions: it suffices to choose $m$ such that $Lm\gtrsim K$ (and, in practice, to promote coverage of distinct sources). 
However two main issues may arise: $i)$ \emph{duplicate} candidates, i.e.\ configurations with large mutual overlap; $ii)$  \emph{spurious equilibria}, i.e.\ stable states unrelated to the desired patterns.

Therefore, since not every run necessarily converges to a disentangled solution, we introduce an acceptance criterion that can be evaluated without access to the ground truth.

\paragraph{Acceptance criterion.}
The acceptance test combines two filters addressing (i) duplicates and (ii) spurious equilibria.
\begin{itemize}
\item[(i)] \textbf{De-duplication via mutual overlap.}
We compute the pairwise overlap
\begin{equation}
\label{eq:overlap_crit}
q_{\ell k} \;=\; \frac{1}{N}\sum_{i=1}^{N}\bar\sigma^{\,\ell}_i \bar\sigma^{\,k}_i,
\end{equation}
and discard duplicates whenever $q_{\ell k}>\delta$, a conservative threshold for random patterns.
A value of $\delta = 0.5$ indicates that two candidates are strongly correlated and therefore likely represent the same recovered pattern (up to noise), since for independent random $\pm1$ patterns the overlap concentrates near zero as $N$ grows. Thus, $0.5$ provides a safe margin to merge near-identical reconstructions while avoiding accidental merges between distinct patterns.\footnote{For two configurations that differ on a fraction $p$ of spins, one has $q \approx 1-2p$; hence $q>0.5$ corresponds to $p<0.25$. In contrast, overlaps between independent random patterns are typically $O(N^{-1/2})$, so values above $0.5$ are overwhelmingly unlikely unless the solutions are duplicates.}.

\item[(ii)] \textbf{Spectral filter via a pseudo-inverse kernel.}
Let $\bm C\in\mathbb{R}^{K\times K}$ be the pattern correlation matrix, $C_{\mu\nu}=\frac{1}{N}\sum_{i=1}^{N}\xi_i^\mu\xi_i^\nu$.
We define the pseudo-inverse kernel~\citep{personnaz1985information,kanter1987associative} as
\begin{equation}
\label{eq:KS_matrix}
J^{KS}_{ij}
\;=\;
\frac{1}{N}\sum_{\mu,\nu=1}^{K}\xi_i^\mu\,C^{-1}_{\mu\nu}\,\xi_j^\nu.
\end{equation}
It can be shown \citep{personnaz1985information,kanter1987associative} that the true patterns are eigenvectors of $\bm J^{KS}$ with eigenvalue $1$, namely
\[
\sum_{j=1}^{N}J^{KS}_{ij}\xi_j^\mu=\xi_i^\mu,
\qquad
\frac{1}{N}\sum_{i,j=1}^{N}\xi_i^\mu J^{KS}_{ij}\xi_j^\mu=1.
\]
Consequently, a theoretically ideal, yet infeasible, acceptance criterion would necessitate
\begin{equation}
\label{eq:ideal_test}
\frac{1}{N}\,(\bar{\bm\sigma}^{\,\ell})^\top\bm J^{KS}\bar{\bm\sigma}^{\,\ell}\;=\;1.
\end{equation}
Although $\bm J^{KS}$ depends on the unknown patterns, it can be approximated directly from the observed Hebbian kernel $\bm J^{\mathrm{Hebb}}$ via iterative unlearning~\citep{fachechi2019dreaming}:
\begin{equation}
\label{eq:iter_JKS}
\bm J_{k+1}
\;=\;
\bm J_k + \frac{\epsilon}{1+\epsilon k}\,\big(\bm J_k - \bm J_k^2\big),
\qquad
\bm J_0=\bm J^{\mathrm{Hebb}},
\quad
\epsilon<\left[\left(1+\sqrt{\operatorname{Tr}\bm J^{\text{Hebb}}/N}\right)^2-1\right]^{-1}.
\end{equation}
We denote by $\widehat{\bm J}^{KS}$ the corresponding unique fixed point (or its numerical approximation).\\
Since exact equality in~\eqref{eq:ideal_test} is rarely achieved in simulations---due to finite fractions of flipped bits in $\bar{\bm\sigma}^{\,\ell}$ and to the approximate nature of $\widehat{\bm J}^{KS}$---we accept a candidate whenever
\begin{equation}
\label{eq:JKS_crit}
\frac{1}{N}\,(\bar{\bm\sigma}^{\,\ell})^\top\,\widehat{\bm J}^{KS}\,\bar{\bm\sigma}^{\,\ell} \;>\; \tau
\end{equation}
where $\tau$ is a threshold whose value can be selected using random matrix theory (see App.~\ref{app:threshold} for a deeper discussion).
\end{itemize}

Overall, we retain $\bar{\bm\sigma}^{\,\ell}$ if and only if
\begin{equation}
\label{eq:acc_crit}
\text{(i)}\;\; \max_{k\neq \ell} q_{\ell k} < 0.5
\qquad\text{and}\qquad
\text{(ii)}\;\; \frac{1}{N}\,(\bar{\bm\sigma}^{\,\ell})^\top\,\widehat{\bm J}^{KS}\,\bar{\bm\sigma}^{\,\ell} > \tau.
\end{equation}
The accepted, distinct candidates form the set $\{\bm\xi_R^\ell\}_{\ell=1}^{\hat{K}}$, where $\hat{K}$ denotes the number of reconstructed patterns.

\paragraph{Generating initial configurations from $\bm J^{\mathrm{Hebb}}$ alone.}
In our target setting, the only accessible object is the synaptic coupling matrix $\bm J^{\mathrm{Hebb}}$, whereas the LAM dynamics also requires input configurations resembling spurious mixtures, i.e.\ sign combinations of the original patterns. We therefore need a mechanism to generate suitable initial states directly from $\bm J^{\mathrm{Hebb}}$.

To this end, we exploit the spectral structure of the pseudo-inverse coupling $\bm J^{KS}$ defined in~\eqref{eq:KS_matrix}. The eigenspace associated with eigenvalue $1$ is $K$-dimensional and is spanned by linear combinations of the true patterns. Moreover, we can approximate $\bm J^{KS}$ starting from $\bm J^{\mathrm{Hebb}}$ via the unlearning iteration~\eqref{eq:iter_JKS}. These observations lead to the following fully unsupervised factorization pipeline:

\begin{enumerate}
\item \textbf{Approximate the pseudo-inverse kernel.}
Starting from $\bm J^{\mathrm{Hebb}}$, iterate
\[
\bm J_{k+1} \;=\; \bm J_k + \frac{\epsilon}{1+\epsilon k}\,(\bm J_k - \bm J_k^2),
\qquad \bm J_0=\bm J^{\mathrm{Hebb}},
\]
until convergence to $\widehat{\bm J}^{KS}$.

\item \textbf{Spectral filtering.}
Compute the eigendecomposition of $\widehat{\bm J}^{KS}$ and retain eigenvectors associated with eigenvalues larger than a threshold (here $\tau$, consistently with~\eqref{eq:JKS_crit}). This yields:
\begin{itemize}
\item an estimate of the effective number of stored patterns,
\[
\hat{K} \;=\; \#\{\lambda_i > \tau\},
\]
used as a proxy for the true load $K$;
\item a set of $\hat{K}$ orthogonalized linear combinations of the true patterns,
\begin{equation}
\label{eq:lin_comb_fatt}
\tilde{x}_i^\delta \;=\; \sum_{\mu=1}^{K}\tilde{c}_\mu^\delta\,\xi_i^\mu,
\qquad
\delta = 1,\dots,\hat{K},
\end{equation}
spanning the same subspace as $\{\bm\xi^\mu\}$.
\end{itemize}

\item \textbf{Generate synthetic spurious mixtures.}
Construct nonlinear mixtures to be used as LAM inputs:
\begin{equation}
\label{eq:lin_comb_fatt_spurious}
x_i^\alpha \;=\; \operatorname{sgn}\!\left(\sum_{\delta=1}^{\hat{K}} c_\delta^\alpha\,\tilde{x}_i^\delta\right),
\qquad
\alpha = 1,\dots,m,
\end{equation}
where $c_\delta^\alpha$ are random coefficients (e.g.\ Gaussian or binary).

\item \textbf{Run LAM and accept reconstructions.}
Feed $\{\bm x^\alpha\}_{\alpha=1}^m$ to the LAM, run the dynamics, and apply the acceptance criterion~\eqref{eq:acc_crit} to extract $\{\bm\xi_R^\ell\}_{\ell=1}^{\hat{K}}$.
\end{enumerate}
The above procedure is summarized in Algorithms~\ref{algo:reconstruction_gen}--\ref{algo:reconstruction_pruning} of App.~\ref{app:pseudocode}.

In summary, this procedure enables an unsupervised factorization of a Hebbian kernel into its underlying patterns, without requiring external samples or side information. The key idea is to leverage the spectral geometry of the pseudo-inverse estimate to generate self-consistent synthetic mixtures, which are then disentangled by the LAM dynamics into the fundamental components.
\\
From now on, we drop the superscript ``Hebb'' and denote the Hebbian synaptic matrix simply by $\bm J$.

\subsection{Dataset partition and round definition}
\label{subsec:dataset_partition}

To fully specify the federated learning setting, we consider $L$ clients and a federation lasting for $T$ communication rounds. To model robustness under strong heterogeneity, we define $K_c\subseteq\{1,\dots,K\}$ as the \emph{effective support} of archetypes observed by client $c$ over the considered horizon, assuming that the federation as a whole covers the entire latent population: $\bigcup_{c=1}^{L}K_{c} \;=\; \{1, \dots, K\}$\footnote{This union represents the complete latent space. The actual presence of specific archetypes at round $t$ is governed by the time-varying mixtures $\pi_{t,c}(\mu)$, allowing for scenarios where classes emerge mid-training.}.

\paragraph{Local batches and class mixtures.}
At round $t\in\{1,\dots,T\}$, each client $c\in\{1,\dots,L\}$ receives a local batch of $M_c^{t}$ examples, $\big\{\big(\bm\eta_c^t\big)^{m}\big\}_{m=1}^{M_c^{t}}$, drawn according to the generative model~\eqref{eq:gen}. We denote by $\bm\pi_{t,c}$ the corresponding class-mixing distribution, i.e.\ $\pi_{t,c}(\mu)$ is the probability that an example observed by client $c$ at round $t$ is generated from archetype $\mu$. Consistent with the definition of $K_c$, this mixture satisfies:
\begin{equation}
\pi_{t,c}(\mu)=0
\qquad
\forall\, t=1,\dots,T,\ \forall\, \mu\notin K_c.
\end{equation}

Consequently, different clients observe different subsets $K_c$ of archetypes and operate under distinct noise levels $r_c$, so that local datasets reflect only specific facets of the global archetypal population. This induces a rigorous statistically non-i.i.d.\ federated setting \cite{kairouz2021advances,hsu2019measuring}.

\noindent
It is important to clarify that the condition $\pi_{t,c}(\mu) = 0$ for $\mu \notin K_c$ describes the \emph{statistical realization} of the local datasets rather than a structural limitation. We do not assume that client $c$ is fundamentally incapable of observing archetype $\mu$ (e.g., due to sensor blindness), but rather that the local generative process is characterized by high heterogeneity and sparsity, effectively rendering specific classes absent from the local sampling budget within the training rounds. This perspective aligns with realistic non-i.i.d. federated scenarios, such as small local datasets or regional prevalence variations, and naturally accommodates the emergence of novel archetypes as the mixture distributions evolve over time.

\paragraph{Round-level aggregate mixture and sample budget.}
We define the global round-level mixture as the client-wise average
\begin{equation}
\label{eq:global_mixture_round}
\pi_t(\mu)
\;:=\;
\frac{1}{L}\sum_{c=1}^L \pi_{t,c}(\mu),
\end{equation}
which summarizes the class-mixing distribution at the level of the entire federation during round $t$ (in particular, whether archetype $\mu$ is present or absent at the server at that round).

We also introduce the total federation sample budget, i.e.\ the total volume of data processed across all clients and rounds:
\begin{equation}
\label{eq:Mtot}
M_{\mathrm{tot}}
\;=\;
\sum_{t=1}^{T}\sum_{c=1}^{L} M_c^t.
\end{equation}

\paragraph{Exposure and coverage.} We conclude by introducing two quantities that are not directly used by the algorithm, but strongly affect the dynamics, as we will see in Sec.~\ref{sec:theory}. The first is the \emph{federated exposure} of archetype $\mu$ at round $t$:
\begin{equation}
\label{eq:exposure}
e_\mu(t)
\;=\;
\frac{1}{L}\sum_{c=1}^L\frac{1}{M_c^t}\sum_{m=1}^{M_c^t}
\mathbf{1}\!\left\{\big(\bm\eta_c^t\big)^{m}\ \text{was generated from archetype }\mu\right\}.
\end{equation}
In words, $e_\mu(t)$ is the empirical realization of the global mixture probability $\pi_t(\mu)$; it measures, on average across clients, the fraction of local examples at round $t$ that originate from archetype $\boldsymbol{\xi}^\mu$. By the law of large numbers, $e_\mu(t) \to \pi_t(\mu)$ as the batch sizes $M_c^t \to \infty$.

The second quantity is the \emph{coverage} up to round $t$:
\begin{equation}
\label{eq:coverage}
\mathcal{C}(t)
\;=\;
\Big\{\mu\in\{1,\dots,K\}:\ \sum_{s=1}^{t} e_\mu(s)>0\Big\},
\end{equation}
namely the set of all archetypes that have been observed at least once (globally) up to round $t$.

In the next subsection, we describe how the federated dynamics unfolds across rounds, the components involved, and how we assess the quality of the model's output.

\section{Main model architecture}
\label{subsec:model-pipeline}

We now describe how the previously described models for associative memories can be embedded into a federated setting. The overall workflow, sketched in Fig.~\ref{fig:pipeline}, is organized in four steps and repeats over $T$ communication rounds. Throughout, raw examples remain strictly local to the clients; only aggregated statistics (synaptic operators) are exchanged.

\paragraph{Pipeline overview.}
At each round $t$, clients build (or update) a local synaptic matrix from their current batch, upload it to the server, and receive back a server-side reconstruction encoded as a synaptic operator. Clients then fuse this global information with their new local evidence through a convex combination, and the process iterates.

\begin{itemize}
    \item[$(i)$] \textbf{Initialization (local operator at $t=0$).}
    At the initial round, each client $c$ computes an unsupervised synaptic matrix from its first batch of $M_c^{0}$ examples:
    \begin{equation}
    \label{eq:init_local_hebb}
        \big(J^{(0)}_c\big)_{ij}
        \;=\;
        \frac{1}{N\,M_c^{0}}
        \sum_{a=1}^{M_c^{0}}
        \big(\eta^{(0)}_c\big)_i^{a}\big(\eta^{(0)}_c\big)_j^{a},
        \qquad c=1,\ldots,L.
    \end{equation}
    Each client then transmits $\bm J_c^{(0)}$ to the server.

    \item[$(ii)$] \textbf{Server aggregation and factorization.}
    The server aggregates the received operators (here, by an equally-weighted average) to obtain a global synaptic matrix
    \begin{equation}
    \label{eq:server_agg}
        \big(J^{(t)}_s\big)_{ij}
        \;=\;
        \frac{1}{L}\sum_{c=1}^{L}\big(J^{(t)}_c\big)_{ij}.
    \end{equation}
    The server then runs the unsupervised factorization procedure (Sec.~\ref{subsec:factorization_procedure}), using $\bm J_s^{(t)}$ as input to the LAM-based reconstruction mechanism, and produces a set of $\hat{K}$ reconstructed patterns
    \[
    \bm J_s^{(t)}
    \;\longrightarrow\;
    \mathrm{LAM}
    \;\longrightarrow\;
    \Big\{\hat{\bm\xi}^{(t)}_{s,\mu}\Big\}_{\mu=1}^{\hat{K}},
    \]
    where $\hat{K}$ is the effective number of retrieved archetypes identified by the spectral filtering and acceptance steps. In general, $\hat{K}$ may differ from the true number $K$ (typically $\hat{K}\le K$ in difficult regimes).

    \item[$(iii)$] \textbf{Broadcast and client fusion.}
    The reconstructed patterns are broadcast back to clients in operator form, via the server synaptic matrix
    \begin{equation}
    \label{eq:server_reencoded}
        \big(\hat{J}^{(t)}_s\big)_{ij}
        \;=\;
        \frac{1}{N}\sum_{\mu=1}^{\hat{K}}
        \big(\hat{\xi}^{(t)}_{s,\mu}\big)_i\,
        \big(\hat{\xi}^{(t)}_{s,\mu}\big)_j.
    \end{equation}
    Upon receiving $\hat{\bm J}^{(t)}_s$, each client combines it with the local operator inferred from its \emph{new} batch at round $t+1$ via a convex blend:
\begin{equation}
\label{eq:convex_comb}
\big(J^{(t+1)}_c\big)_{ij}
\;=\;
w_c(t)\,
\underbrace{\frac{1}{N\,M_c^{t+1}}\sum_{a=1}^{M_c^{t+1}}
\big(\eta^{(t+1)}_c\big)_i^{a}\big(\eta^{(t+1)}_c\big)_j^{a}}_{\mathclap{\text{current local evidence}}}
\;+\;
\big(1-w_c(t)\big)
\underbrace{\big(\hat{J}^{(t)}_s\big)_{ij}}_{\mathclap{\text{server reconstruction}}}
\end{equation}
    where $w_c(t)\in[0,1].$

    \item[$(iv)$] \textbf{Iteration.}
    Steps $(ii)$--$(iii)$ are repeated for $t=0,\ldots,T-1$.
\end{itemize}

These steps are illustrated in Figs.~\ref{fig:pipeline}--\ref{fig:LAM_layer}. The pseudo-code is detailed in Algorithm~\ref{algo:Fede_algo} in App.~\ref{app:pseudocode}.

\begin{figure}[t!]
    \centering
    \includegraphics[width=15.5cm]{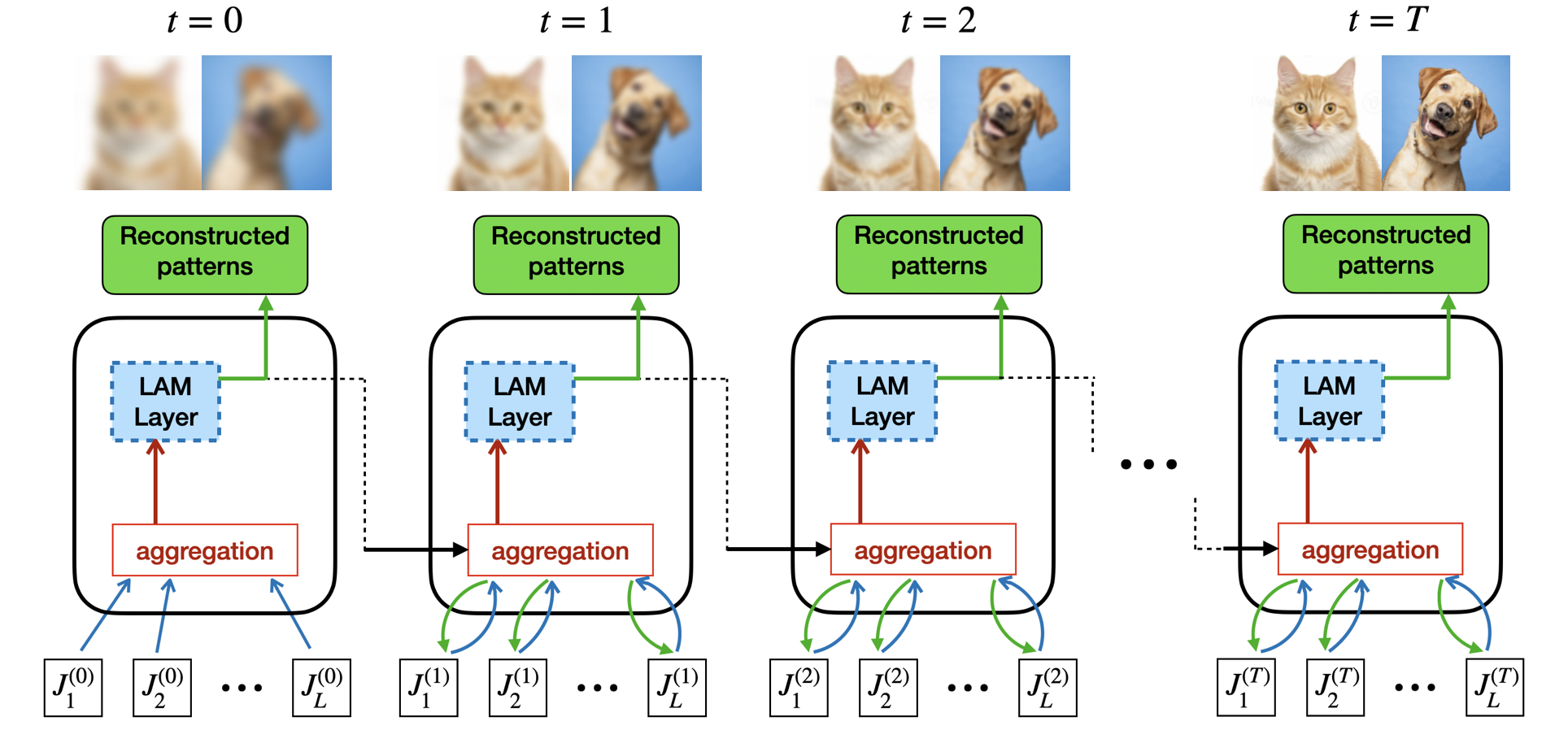}
    \caption{Illustration of the federated pipeline. Each block represents a federation round $t$. At each round, the $L$ client layers provide as input their synaptic matrices, each estimated from the batch of examples available at that round. These matrices are combined in the \textit{aggregation} layer. For $t=0$, the aggregation uses only the information received from the clients. For $t>0$, the aggregation combines both (i) the client information from round $t$ and (ii) the server-side information carried over from the previous round $t-1$. 
    \\
    The aggregated information can either be sent back to the clients as feedback to improve the estimation of the ground-truth synaptic matrix, or forwarded to the \textit{LAM} layer, where pattern reconstruction is performed. At round $t=0$ there is no feedback from the federation to the clients. At each round, we can access both the reconstructed patterns obtained after the \textit{LAM} layer and the updated client synaptic matrices.
    Panel $b)$ shows a zoom-in of the \textit{LAM} layer, where the \textit{pattern reconstruction} is performed at each round. This layer takes as input the output of the \textit{aggregation} layer and first applies an iterative algorithm to estimate $\mathbf{J}^{KS}$. It then generates a sufficiently large set of initial mixing states and feeds them into the LAM model, which collects possible \textit{pattern candidates}. These candidates are passed to the \textit{pruning} layers, which remove duplicates and apply a \textit{spectral criterion} to discard forbidden candidates. The final output of the procedure is the set of $\hat{K}$ reconstructed patterns.}
    \label{fig:pipeline}
\end{figure}

\begin{figure}[t!]
    \centering
    \includegraphics[width=15.5cm]{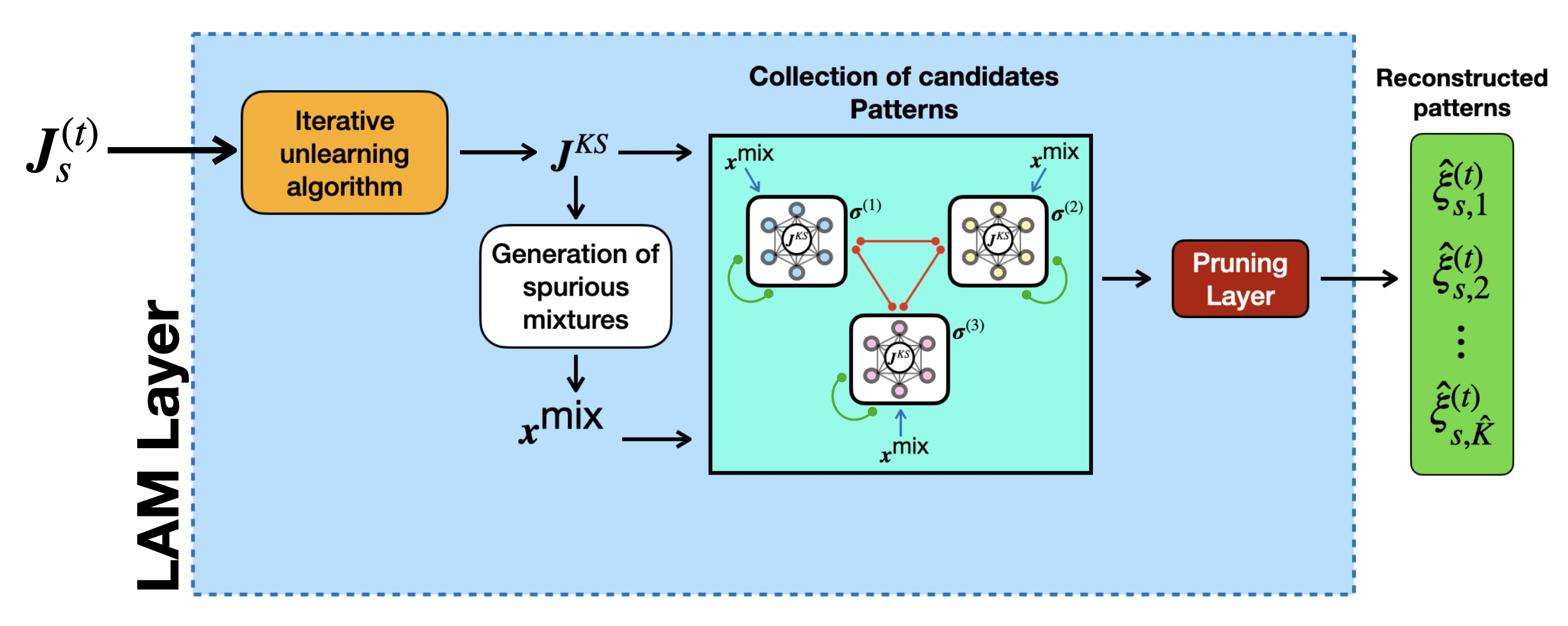}
    \caption{We show a zoom-in of the \textit{LAM} layer, where the \textit{pattern reconstruction} is performed at each round. This layer takes as input the output of the \textit{aggregation} layer and first applies an iterative algorithm to estimate $\mathbf{J}^{KS}$. It then generates a sufficiently large set of initial mixing states and feeds them into the LAM model, which collects possible \textit{pattern candidates}. These candidates are passed to the \textit{pruning} layers, which remove duplicates and apply a \textit{spectral criterion} to discard forbidden candidates. The final output of the procedure is the set of $\hat{K}$ reconstructed patterns.}
    \label{fig:LAM_layer}
\end{figure}

Our pipeline clearly departs from standard federated learning frameworks: clients do not exchange gradients or model parameters, but instead share \emph{synaptic operators} that capture archetypal structural information~\cite{personnaz1986collective}. Consistent with the federated paradigm, raw data remain strictly local and are never transmitted to the server; only aggregated correlation statistics are communicated. This design is not merely a practical convenience. By transmitting operator-level summaries, each client contributes a concise representation of the underlying archetypal geometry, allowing the server to integrate a coherent global structure while fully respecting data locality.

A central ingredient of the procedure is the client-wise convex combination in~\eqref{eq:convex_comb}. The hyper-parameter $w_c(t)$ acts as a \emph{plasticity knob}: small $w_c(t)$ privileges consolidation (lower variance but higher bias toward past information), whereas large $w_c(t)$ privileges plasticity (lower bias to stale information, higher variance due to finite-batch fluctuations). This operator-level mechanism mirrors the complementary learning systems (CLS) hypothesis in cognitive neuroscience---fast traces integrating into a slower long-term store---and provides an interpretable analogue for federated continual learning~\cite{o2014complementary,kumaran2016learning}.

\paragraph{Reconstruction metrics.}
To quantify retrieval performance after each round, we monitor the alignment between the reconstructed server patterns and the ground-truth archetypes. Specifically, we report the \emph{magnetization} of the reconstructed set,
\begin{equation}
\label{eq:mag}
m_t((\hat{\bm\xi}_s^{(t)})\mid\bm\xi^\mu)
\;:=\;
\max_{\nu\in\{1,\dots,\hat{K}\}}
\frac{1}{N}\,
\big|(\hat{\bm\xi}_s^{(t)})^\nu\cdot \bm{\xi}^{\mu}\big|,
\end{equation}
which measures the maximum normalized overlap between any reconstructed archetype and the true pattern set.

In synthetic experiments, where the teacher operator is available, we also consider the normalized Frobenius error between the reconstructed kernel and the ground-truth Hebbian matrix:
\begin{equation}
\label{eq:fro}
\mathrm{FRO}_t(\bm{J})
\;:=\;
\frac{\big\|\bm{J}^{(t)}-\bm{J}^\star\big\|_F}{\|\bm{J}^\star\|_F},
\end{equation}
where $\bm{J}^\star$ denotes the true Hebbian coupling matrix and $\big\|\cdot\big\|_{F}$ is the Frobenius norm.

Now that the main components of our pipeline and the associated evaluation metrics have been introduced, we can move to the presentation of our core results. We will address both the theoretical and algorithmic sides of the proposed framework.

\section{Theoretical findings} \label{sec:theory}

On the theory side, leveraging tools from random matrix theory, we characterize when and why \emph{novelties} become detectable, namely when examples originating from archetypes that have never appeared before a given round start to surface in the federation. In particular, we show that \emph{exposure} is the key quantity governing detectability: as exposure increases, Baik-Ben Arous-Péché(BBP)-type outliers split from the bulk spectrum and the corresponding leading eigenvectors align with archetypal directions, thereby providing reliable seeds for reconstruction.

On the algorithmic side, we introduce a data-driven, entropy-based controller that adapts the network plasticity over rounds. This mechanism allows the system to remain responsive to distributional drift and the arrival of new archetypes, while simultaneously protecting previously consolidated representations. As a by-product, the same controller provides a principled form of \emph{preventive caution} against adverse conditions such as low-quality data batches or the presence of corrupted (infected) clients, by automatically reducing over-reactivity when the incoming information is inconsistent.

Overall, our results substantiate two claims:
\begin{itemize}
    \item[-] \textbf{Exposure governs detectability.} When a class gains sufficient exposure, BBP outliers detach from the bulk and the leading eigenvectors align with archetypal directions, enabling reliable seeding and reconstruction.
    \item[-] \textbf{Entropy-driven stability--plasticity control.} An entropy-based controller yields robust schedule tracking without hand tuning, preserving consolidated attractors while remaining adaptive to drift. The system incorporates new archetypes with low detection latency and minimal disruption, expanding its effective representational subspace precisely when new directions become informative.
\end{itemize}

\subsection{The role of exposure and BBP transitions} 

\label{subsec:bbp}

This section explains when and why archetypal structure becomes \emph{spectrally visible} in the server-side estimator---a prerequisite for the seeding and counting steps in our pipeline. At communication round $t$, the server forms the aggregated Hebbian operator $\bm J_s^{(t)}$ by combining client-side correlators computed on heterogeneous, fragmented batches. Since the server never observes raw data, and since class presence can vary across clients and across rounds, the key question is: \emph{under what conditions does $\bm J_s^{(t)}$ develop eigenvalue/eigenvector signatures that reveal which archetypes are currently present?}

Our answer is spectral. We show that the \emph{population} operator associated with $\bm J_s^{(t)}$ has a spiked--Wishart structure in which each archetype contributes a rank-one spike with strength proportional to its round-level exposure. Moreover, $\bm J_s^{(t)}$ concentrates sharply around this population operator, so the finite-sample spectrum inherits the same bulk/outlier geometry up to controlled fluctuations. This yields explicit detectability thresholds (BBP transitions), predictions for outlier locations, and quantitative eigenvector alignment laws, which together justify the spectral steps used downstream (thresholding to estimate $\hat K(t)$ and using leading eigenvectors as reconstruction seeds).

\vspace{3mm}
\noindent
\paragraph{Spiked decomposition and the meaning of exposure.}
With the data model of Subsec.~\ref{subsec:dataset_partition}, let $\sigma^2$ denote the (per-coordinate) noise variance.
In the Rademacher channel, this is given by the identity $\sigma^2=1-r^2$, where for simplicity we set $r_c = r$ for all $c \in \{1, \dots, L\}$.
To apply standard random matrix results, we analyze the spectrum of the \emph{rescaled} operator $\tilde{\bm{J}}_s^{(t)} := N \bm{J}_s^{(t)}$.
The population counterpart of this rescaled estimator decomposes into an isotropic noise part plus a finite-rank signal aligned with the true archetypes.
Letting $u_\mu := \bm{\xi}^\mu/\sqrt{N}$ denote the normalized archetype vectors, we have:
\begin{equation}\label{eq:spiked-decomp}
\mathbb{E}[\tilde{\bm{J}}_s^{(t)}]\;=\;\sigma^2 \bm{I}\;+\;r^2\sum_{\mu=1}^K \pi_t(\mu)\,u_\mu u_\mu^\top.
\end{equation}
The identity term $\sigma^2\bm I$ captures the isotropic contribution of random bit flips, while each archetype contributes a rank-one projector $u_\mu u_\mu^\top$ weighted by the global mixture mass $\pi_t(\mu)$ and signal intensity $r^2$.
Consequently, strictly identifying the spike strength in the spiked-covariance model yields:
\begin{equation}
\label{eq:spike_strength_def}
\text{Spike strength of }\mu: \quad \theta_\mu^{(t)} = r^2\,\pi_t(\mu) 
\quad\Longrightarrow\quad
\kappa_\mu^{(t)} = \frac{r^2\,\pi_t(\mu)}{\sigma^2} = \frac{r^2\,\pi_t(\mu)}{1-r^2}.
\end{equation}
Intuitively, an archetype is invisible when it has (near-)zero exposure, and it becomes detectable only once its exposure makes the corresponding spike strong enough to separate from the noise bulk.

The remaining task is to turn this intuition into a finite-sample statement: (a) show that $\bm J_s^{(t)}$ stays close to its expectation in operator norm at the relevant scales, and (b) characterize when the spikes of~\eqref{eq:spiked-decomp} produce outliers and aligned eigenvectors in the empirical spectrum.

\vspace{3mm}
\noindent
\paragraph{Non-asymptotic stability of the round operator.}
We first collect concentration bounds ensuring that $\bm J_s^{(t)}$ is a small perturbation of its spiked population counterpart. Throughout, $\|\cdot\|_{\mathrm{op}}$ denotes operator norm, $M_{\mathrm{round}}=L\,M_c$\footnote{While the general framework allows for time-varying and heterogeneous budgets $M_c^t$ per client, for the sake of theoretical clarity we conduct the concentration analysis assuming a representative balanced round where $M_c^t \equiv M_c$. In this simplified setting, $M_{\mathrm{round}} = L M_c$ acts as the effective aggregate sample volume; the results extend naturally to the general case by replacing $M_{\mathrm{round}}$ with the total round-level budget $M_{\mathrm{round}}(t) = \sum_{c=1}^L M_c^t$.}, and $\sigma^2$ is as above.

\begin{theorem}[Concentration]\label{thm:concentration}
\leavevmode\par\smallskip\noindent
Let $(\chi_j)_{j=1}^N$ be independent random variables. For a fixed communication round index $t$, let $\bm{J}_s^{(t)}$ be defined by \eqref{eq:server_agg}. Then, for every deviation level $u>0$:
\begin{enumerate}
\item[\textnormal{(i)}]\label{thm:concentration:subg} \textbf{(Sub-Gaussian case)} 
If $\|\chi_j\|_{\psi_2}\le C$ for all $j$, there exists a constant $c_1>0$ (depending only on $C$) such that
\begin{equation}\label{eq:conc-main}
\mathbb{P}\left(\,\big\|\bm{J}_s^{(t)}-\mathbb{E}[\bm{J}_s^{(t)}]\big\|_{\mathrm{op}}>u\,\right)
\;\le\; 2N\exp\left(-\,c_1\,M_{\mathrm{round}}\cdot \min\left\{\frac{u^2}{\sigma^{4}},\,\frac{u}{\sigma^{2}}\right\}\right),
\end{equation}
where $\sigma^2:=\max_j \mathrm{Var}(\chi_j)$.

\item[\textnormal{(ii)}]\label{thm:concentration:rademacher} \textbf{(Bounded / Rademacher channel)} 
If $\chi_j\in\{\pm1\}$ with $\mathbb{E}[\chi_j]=r$, there exists a universal constant $c_\ast>0$ such that
\begin{equation}\label{eq:conc-rad}
\mathbb{P}\left(\,\big\|\bm{J}_s^{(t)}-\mathbb{E}[\bm{J}_s^{(t)}]\big\|_{\mathrm{op}}>u\,\right)
\;\le\; 2N\exp\left(-\,c_\ast\,M_{\mathrm{round}}\cdot \min\left\{\frac{u^2}{\sup_{\lambda\in[0,1]}\lambda(1-\lambda)},\,u\right\}\right),
\end{equation}
and, in particular,
\begin{equation}
\sup_{\lambda\in[0,1]}\lambda(1-\lambda)= \tfrac14.
\end{equation}

\item[\textnormal{(iii)}]\label{thm:concentration:finiteN} \textbf{(Finite-$N$ refined bound)} 
In the bounded/Rademacher setting of \textnormal{(ii)}, for all $N$ and $M_c$ there exist universal constants $c_3,C>0$ such that
\begin{equation}\label{eq:conc-finiteN}
\mathbb{P}\left(\,\big\|\bm{J}_s^{(t)}-\mathbb{E}[\bm{J}_s^{(t)}]\big\|_{\mathrm{op}}>u\,\right)
\;\le\; 2\exp\left(-\,c_3\,M_{\mathrm{round}}\cdot \min\left\{\frac{u^2}{\sup_{\lambda\in[0,1]}\lambda(1-\lambda)},\,u\right\}\;+\;C\log N\right).
\end{equation}
\end{enumerate}
\end{theorem}

\vspace{3mm}
\noindent
The bounds \eqref{eq:conc-main}--\eqref{eq:conc-finiteN} are non-asymptotic and scale with the effective round size $M_{\mathrm{round}}$. Operationally, they ensure that the empirical spectrum of $\bm J_s^{(t)}$ cannot deviate significantly from the spectrum of its spiked expectation unless the round is severely under-sampled. In particular, once $M_{\mathrm{round}}$ is large enough for the right-hand side to be small at the scale of interest, the bulk edge and any supercritical outliers predicted by the spiked model remain stable under finite-sample noise. (Proofs and auxiliary ingredients are deferred to the appendix.)

\medskip
\noindent
\paragraph{BBP transitions: when spikes separate and align.}
Given the stability above, we can locate the empirical spectrum relative to the Marchenko--Pastur (MP) bulk induced by the isotropic term, and describe precisely when an exposure-weighted spike produces a detectable outlier. The next theorem formalizes this in the standard spiked-covariance setting under near-orthogonality of the spike directions (which is natural for random $\{\pm1\}$ archetypes).

\vspace{3mm}
\noindent
\begin{theorem}[Detection thresholds for the spiked covariance operator]\label{thm:detection-thresholds}
Let $u_\mu:=\xi^\mu/\sqrt{N}$ and consider the rank--$K$ spiked population covariance
\[
C\;:=\;\sigma^2 I_N + \sum_{\mu=1}^K \theta_\mu\,u_\mu u_\mu^\top,
\qquad 
\kappa_\mu:=\frac{\theta_\mu}{\sigma^2},
\]
with aspect ratio $q:=N/M_{\mathrm{round}}$. Denote by
\(
\lambda_\pm(q):=\sigma^2(1\pm\sqrt q)^2
\)
the Marchenko--Pastur edges (cf.~\eqref{eq:mp-edge}). Let $J$ be the round--$t$ operator, which can be written as
\[
J = \Sigma^{1/2} S_0 \Sigma^{1/2},\qquad
S_0:=\frac{1}{M_{\mathrm{round}}}\sum_{k=1}^{M_{\mathrm{round}}} Z_k Z_k^\top,
\]
with $\Sigma=C$ and whitened rows $Z_k$ satisfying the standing assumptions under which the isotropic MP local law holds (Theorem~\ref{thm:mp-global-law}). Assume that the number of spikes $K$ is fixed (independent of $N$), that the spike directions $u_\mu$ have unit norm, and that their Gram matrix $\Gamma:=U^\top U$ with $U:=[u_1,\dots,u_K]$ satisfies
\[
\varepsilon_{\mathrm{orth}}
:=\max_{\mu\neq\nu}|\langle u_\mu,u_\nu\rangle|\to0
\quad\text{as }N\to\infty.
\]
Assume additionally that the nonzero spike strengths $\{\kappa_\mu\}$ are pairwise distinct. Then, for every fixed $\delta\in(0,1)$ there exists a deterministic sequence $\varepsilon_N(\delta)\downarrow0$ such that, with probability at least $1-\delta$ for all $N$ large enough, the following hold:
\begin{enumerate}[label=(D\arabic*)]
\item \textbf{Bulk confinement.} All but at most $K$ eigenvalues of $J$ lie in
\begin{equation}\label{eq:bulk-band}
\big[\lambda_-(q)-\varepsilon_N(\delta),\ \lambda_+(q)+\varepsilon_N(\delta)\big],
\end{equation}
where
\begin{equation}\label{eq:epsilonN}
\varepsilon_N(\delta)\;=\;C\sqrt{\frac{\log(1/\delta)+\log N}{M_{\mathrm{round}}}}\;+\;C'\,\frac{\log^2 N}{N^{2/3}},
\end{equation}
for universal constants $C,C'>0$ (depending only on $q$ and the tail/regularity parameters).

\begin{figure}[t!]
    \centering
    \includegraphics[width=1.0\linewidth]{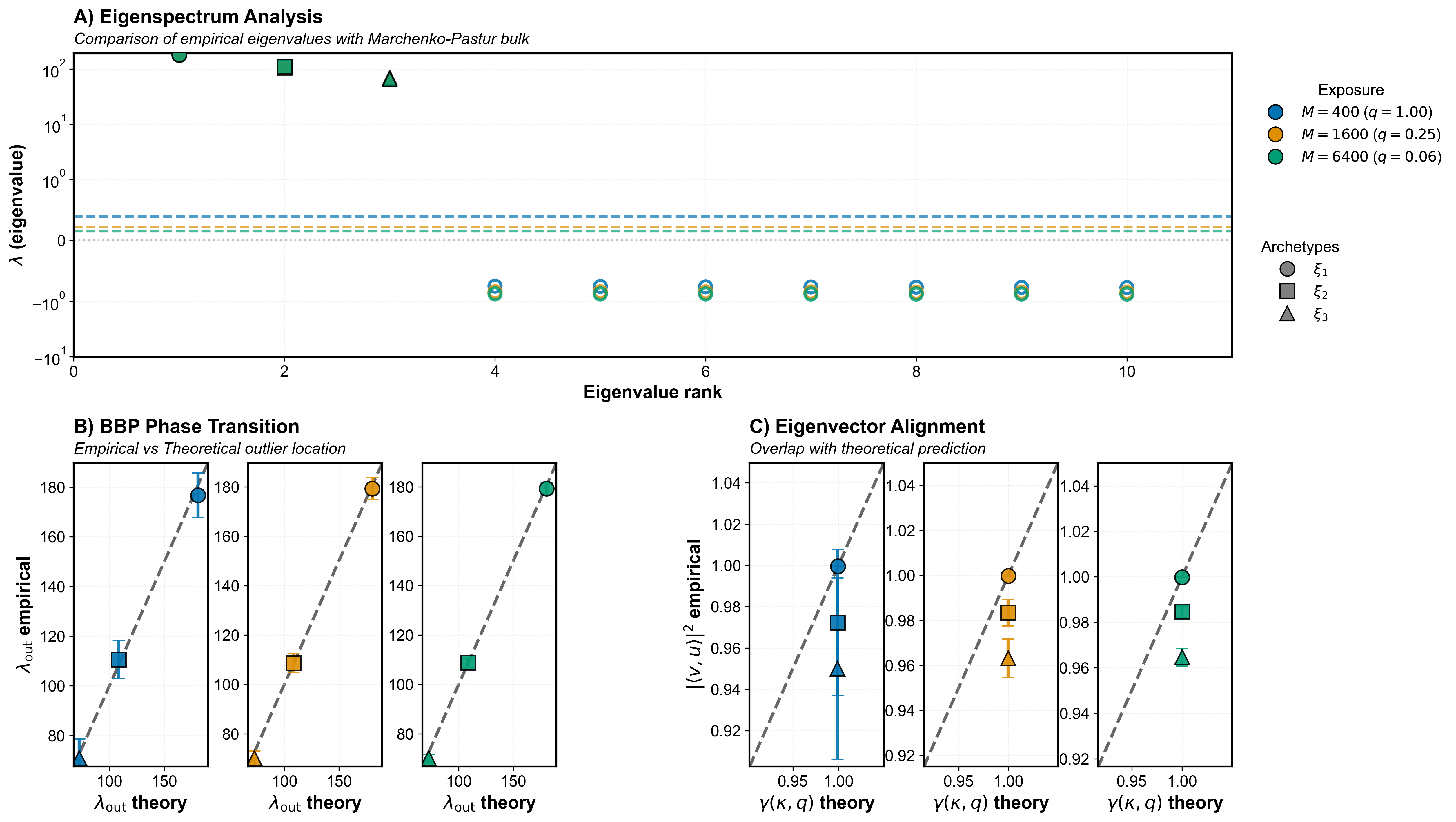}
    \caption{Spectral validation of the BBP framework in federated Hebbian learning. 
    \textit{(Top)} Panel A displays the top 10 eigenvalues of the empirical Hebbian matrix for three representative exposure levels ($M \in \{400, 1600, 6400\}$), with the Marchenko–Pastur upper edge $\lambda_+(q)$ shown as a dashed line for each aspect ratio $q = N/M$. Filled markers denote eigenvalues above the MP threshold (detected spikes corresponding to the three strong archetypes), while empty markers indicate eigenvalues below threshold. The near-absence of empirical spikes for the weak archetypes validates that exposure below the BBP threshold $\kappa < \sqrt{q}$ prevents spectral detection despite their theoretical presence.
    \textit{(Bottom Left)} Panel B compares empirical spike eigenvalues (averaged over 50 trials) against the closed-form BBP prediction $\lambda_{\text{out}}(\kappa, q)$ across all three archetypes and exposure levels. The near-diagonal scatter ($R^2 \in [0.955, 0.977]$) confirms that the spiked-Wishart population model, combined with concentration guarantees, accurately predicts empirical outlier locations at finite sample size ($N=400$).
    \textit{(Bottom Right)} Panel C shows eigenvector overlaps $|\langle v, u \rangle|^2$ compared against the theoretical alignment formula $\gamma(\kappa,q)$ from Theorem D4. The median gap between empirical and theoretical values is $0.0163$, with the maximum gap of $0.0486$ concentrated on the weakest archetype ($\xi_3$). This systematic gap for lower-exposure archetypes is consistent with the finite-size correction $O(N^{-1/2})$ predicted by the theory: the formula $\gamma(\kappa,q)$ is asymptotic ($N \to \infty$), and at $N=400$ one observes deviations that scale inversely with exposure strength. The variance of empirical overlaps visibly decreases as $M$ increases (moving left to right), reflecting the convergence of the sample estimator to its population limit.
    }
    \label{fig:bbp-panel}
\end{figure}

\item \textbf{BBP threshold (spike detection).} For each $\mu$ with $\kappa_\mu>0$ there is at most one eigenvalue of $J$ that can be asymptotically attributed to the spike $\kappa_\mu$. In particular:
\begin{itemize}
\item If $\kappa_\mu>\sqrt q$, then there exists a unique outlier eigenvalue $\lambda_\mu(J)>\lambda_+(q)$ that converges to the population-dependent value in (D3). 
\item If $\kappa_\mu\le\sqrt q$, no eigenvalue separates from the bulk edge $\lambda_+(q)$ due to that spike.
\end{itemize}

\item \textbf{Outlier location (decoupled/orthogonal case).} In the idealized case where the spike directions are exactly orthonormal, $\langle u_\mu,u_\nu\rangle=\delta_{\mu\nu}$, the asymptotic location of the outlier associated with a spike of strength $\kappa>0$ is
\begin{equation}\label{eq:outlier-loc}
\lambda_{\mathrm{out}}(\kappa)\;=\;\sigma^2\,(1+\kappa)\Big(1+\frac{q}{\kappa}\Big).
\end{equation}
Moreover, $\lambda_{\mathrm{out}}(\kappa)>\lambda_+(q)$ if and only if $\kappa>\sqrt q$, and at $\kappa=\sqrt q$ the outlier merges with the upper MP edge.

In the nearly-orthogonal case, the non-orthogonality perturbs the population spikes by at most
$O(\varepsilon_{\mathrm{orth}})$ (by standard eigenvalue perturbation of the Gram matrix $\Gamma=U^\top U$),
hence the corresponding sample outliers concentrate around the same BBP locations up to this bias.
More precisely, for each $\mu$ with $\kappa_\mu>\sqrt q$ and for every fixed $\delta\in(0,1)$, there exists
$C_\delta>0$ such that, for all $N$ large enough,
\[
\mathbb{P}\Big(\big|\lambda_\mu(J)-\lambda_{\mathrm{out}}(\kappa_\mu)\big|
\le C_\delta\big(\varepsilon_{\mathrm{orth}}+N^{-1/2}\big)\Big)\ \ge\ 1-\delta.
\]
Equivalently, $\lambda_\mu(J)=\lambda_{\mathrm{out}}(\kappa_\mu)+O(\varepsilon_{\mathrm{orth}})+O_p(N^{-1/2})$,
and in particular $\lambda_\mu(J)\xrightarrow{P}\lambda_{\mathrm{out}}(\kappa_\mu)$ as
$N\to\infty$ and $\varepsilon_{\mathrm{orth}}\to0$. \\
Moreover, choosing $\delta=\delta_N$ summable (e.g. $\delta_N=N^{-D}$ for $D>2$) and applying
Borel--Cantelli yields the almost sure convergence $\lambda_\mu(J)\to\lambda_{\mathrm{out}}(\kappa_\mu)$.

\item \textbf{Eigenvector alignment.} For each $\mu$ with $\kappa_\mu>\sqrt q$, let $v_\mu$ be the (random) unit eigenvector of $J$ associated with the outlier eigenvalue near $\lambda_{\mathrm{out}}(\kappa_\mu)$, with sign chosen so that $\langle v_\mu,u_\mu\rangle\ge0$. Then
\begin{equation}\label{eq:eig-overlap}
\big|\langle v_\mu,u_\mu\rangle\big|^2 \ \xrightarrow{P}\ 
\gamma(\kappa_\mu,q)
:=\frac{1-\dfrac{q}{\kappa_\mu^2}}{1+\dfrac{q}{\kappa_\mu}}
\in(0,1),
\end{equation}
as $N,M_{\mathrm{round}}\to\infty$ with $N/M_{\mathrm{round}}\to q$. In the nearly-orthogonal case the same limit holds, with finite-$N$ deviations of order $O(\varepsilon_{\mathrm{orth}}+N^{-1/2})$.
\end{enumerate}
\end{theorem}

\medskip
\noindent
\paragraph{From theory to an operational detector.}
Theorem~\ref{thm:detection-thresholds} yields a direct, round-wise detection rule: the bulk is confined within a narrow band around the MP support (D1), while any archetype with supercritical signal-to-noise ratio $\kappa_\mu>\sqrt q$ generates an outlier eigenvalue above the upper edge (D2), with predictable location (D3) and nontrivial eigenvector overlap (D4). Accordingly, given a cut $\tau$ (e.g.\ MP/Shuffle/TW) and a finite-$N$ cushion calibrated to the band width in~\eqref{eq:epsilonN}, we define
\begin{equation}\label{eq:keff-op}
\hat{K}(t)\;=\;\#\left\{i:\ \lambda_i\big(\bm{J}^{(t)}_{\mathrm{KS}}\big)>\tau\right\},
\end{equation}
which counts the number of empirically supercritical directions up to the controlled fluctuations. Let $\{v_i\}_{i=1}^{\hat{K}(t)}$ denote the corresponding leading eigenvectors; the alignment law~\eqref{eq:eig-overlap} explains why these vectors provide meaningful seeds, and why their quality improves monotonically as exposure increases.

\medskip
\noindent
\paragraph{Interpretation for federated dynamics.}
Equation~\eqref{eq:spiked-decomp} makes the organizing principle explicit: exposure determines spike strength, spike strength determines spectral separation, and separation determines whether an archetype can be discovered from $\bm J_s^{(t)}$ at that round. In particular, archetypes with negligible exposure cannot generate outliers and remain invisible until their exposure increases; conversely, once exposure crosses the BBP threshold, their eigenvalues detach and their eigenvectors begin to align, enabling reliable reconstruction. Figure~\ref{fig:bbp-panel} provides a numerical validation of the outlier and alignment predictions at finite $N$.

Taken together, these results close the loop between the federated sampling process and the spectral heuristics used in the pipeline: a principled thresholding step yields $\hat K(t)$, while the associated leading eigenspace provides high-quality seeds for the subsequent heteroassociative refinement. In the next sections we exploit this lens to study non-stationary regimes, where $\pi_t$ drifts and new archetypes enter the federation, and to quantify the resulting detection latency and interference effects.

\begin{remark}[Scope of Theoretical Guarantees]
The concentration bounds and spectral thresholds derived in Theorems \ref{thm:concentration} and \ref{thm:detection-thresholds} are established under the assumption of a fixed effective noise level characterizing the aggregated operator. 
While our numerical simulations explore a more complex regime involving fully adaptive, client-specific weights $w_c(t)$ and the presence of pure-noise attackers, the theoretical analysis provides the fundamental detectability guarantees for the \textit{effective} federation. 
Essentially, the theory predicts \textit{when} archetypes become spectrally detectable provided that the adaptive controller successfully suppresses high-variance clients, thereby validating the consistency of the reconstruction limits in the asymptotic regime.
\end{remark}

\subsection{Plasticity and stability: the role of $w_c(t)$}
\label{subsec:plasticity_w}

In this subsection we provide additional details on the construction of an adaptive client-wise weight $w_c(t)$ in the blending rule~\eqref{eq:convex_comb}. Recall that, at round $t$, client $c$ updates its local operator by combining the server reconstruction from the previous round with the correlator computed from the current local batch:
\begin{equation}
\label{eq:w_update_rule}
\bm{J}_c^{(t)} \;=\; \bigl(1-w_c(t)\bigr)\,\bm{A}^{(t-1)} \;+\; w_c(t)\,\tilde{\bm{J}}_c^{(t)},
\end{equation}
where
\begin{equation}
\label{eq:A_def}
\bm{A}^{(t-1)} \;=\; \frac{1}{N}\, \bm{\hat\xi}^{(t-1)}\bigl(\bm{\hat\xi}^{(t-1)}\bigr)^{\top},
\qquad
\tilde{\bm{J}}_c^{(t)} \;=\; \frac{1}{N\,M_{c}^t}\sum_{a=1}^{M_c^t}(\bm{\eta}^{(t)}_c)^a \bigl[(\bm{\eta}_c^{(t)})^a\bigr]^{\top}.
\end{equation}
Here, $\bm A^{(t-1)}$ encodes the archetypal structure reconstructed at the server up to round $t-1$, while $\tilde{\bm J}_c^{(t)}$ captures the statistics of the new batch observed by client $c$ at round $t$.

An effective adaptive rule for $w_c(t)$ should react to \emph{two} distinct phenomena: (i) the quality of the example--to--archetype channel (controlled by $r$), and (ii) genuine distributional change across rounds (e.g.\ the appearance of previously unseen archetypes). Both effects manifest as discrepancies between the sign structure of the ``consolidated'' operator $\bm A^{(t-1)}$ and that of the ``current'' operator $\tilde{\bm J}_c^{(t)}$. We therefore measure their agreement at the level of signs.

To keep notation light, set $\bm A=\bm A^{(t-1)}$ and $\bm B=\tilde{\bm J}_c^{(t)}$, and define the sign matrices
\begin{align*}
  s^{A}_{ij} &= \operatorname{sign}(A_{ij}),\\
  s^{B}_{ij} &= \operatorname{sign}(B_{ij}).
\end{align*}
We then consider the empirical agreement probability over off-diagonal entries,
\begin{equation}
\label{eq:p_sign_agreement}
p \;:=\; \frac{1}{N(N-1)}\sum_{i\neq j}\mathbf{1}\{s^A_{ij} = s^B_{ij}\},
\end{equation}
and map it to a scalar uncertainty score via the binary entropy
\begin{equation}
\label{eq:H_AB}
H_{AB} \;:=\; h_2(p)
\;=\; -p\log_2 p - (1-p)\log_2(1-p).
\end{equation}
By construction, $H_{AB}\in[0,1]$: it is small when $\bm A$ and $\bm B$ are highly sign-consistent, and it approaches $1$ when their signs are nearly uncorrelated.

A crucial point is that $H_{AB}$ cannot be interpreted in absolute terms without accounting for the intrinsic corruption noise in the data. Indeed, $\bm B$ is a Hebbian correlator built from noisy examples, whereas $\bm A$ is a Hebbian correlator built from reconstructed archetypes. Even in the ideal case where the reconstruction is perfect (so that $\bm A$ matches the true archetypal Hebbian matrix), the example channel introduces a nonzero baseline disagreement.

Under the generative model~\eqref{eq:gen}, a simple calculation gives, for $i\neq j$,
\begin{equation}
\label{eq:pairwise_channel}
\mathbb{P}\bigl(\eta_i\eta_j=\xi^\mu_i\xi^\mu_j\bigr) \;=\; \frac{1+r^2}{2}.
\end{equation}
This motivates a \emph{minimum uncertainty} (noise floor) defined as
\begin{equation}
\label{eq:H_min}
H_{\min}(r) \;=\; h_2\!\left(\frac{1+r^2}{2}\right).
\end{equation}
At the opposite extreme, when examples are completely uninformative ($r=0$), one has maximal uncertainty,
\begin{equation}
\label{eq:H_max}
H_{\max}=h_2\!\left(\tfrac12\right)=1.
\end{equation}

We can now define the adaptive weight by normalizing the observed uncertainty above the noise floor:
\begin{equation}
\label{eq:w_entropy_rule}
w_c(t)
\;=\;
\max\bigg(0, \ \frac{H_{AB}-H_{\min}(r)}{1-H_{\min}(r)}\bigg).
\end{equation}
This mapping ensures $w_c(t)\in[0,1]$\footnote{Since finite-size fluctuations may occasionally yield empirical entropies slightly below the theoretical floor, we clip the value to ensuring non-negativity.} 
and it has the intended qualitative behavior:
\begin{itemize}
\item If $H_{AB}\approx H_{\min}(r)$, then the discrepancy between $\bm A$ and $\bm B$ is largely explained by the intrinsic corruption noise, and the rule drives $w_c(t)\approx 0$, favoring stability/consolidation.
\item If $H_{AB}$ is close to $1$, the current batch correlator becomes nearly random in sign relative to $\bm A$, which we interpret as a strong signal of distributional change; the rule then yields $w_c(t)\approx 1$, increasing plasticity.
\end{itemize}

It is worth noting that the numerator $H_{AB}-H_{\min}(r)$ can become small in two qualitatively different ways. First, for a fixed sign-consistency level $H_{AB}$, increasing $H_{\min}(r)$ corresponds to a degraded example channel (smaller $r$), so that the observed entropy is largely explained by corruption noise and $w_c(t)$ is damped. Second, for fixed $r$, the entropy $H_{AB}$ itself may relax toward $H_{\min}(r)$: the round carries little genuinely new information and the controller again reduces plasticity.

Moreover, to prevent abrupt changes of $w_c(t)$ across rounds, we apply an exponential moving average (EMA) to the instantaneous value $w_{\mathrm{current}}$ computed from~\eqref{eq:w_entropy_rule}:
\begin{equation}
\label{eq:w_ema}
w_c(t) \leftarrow \alpha\, w_{c}(t) \;+\; (1-\alpha)\, w_c(t-1),
\qquad \alpha\in[0,1].
\end{equation}
This adaptive behavior is illustrated in the numerical experiments (see Fig.~\ref{fig:adaptive}).

\begin{figure}[t]
\centering
\includegraphics[width=0.95\textwidth]{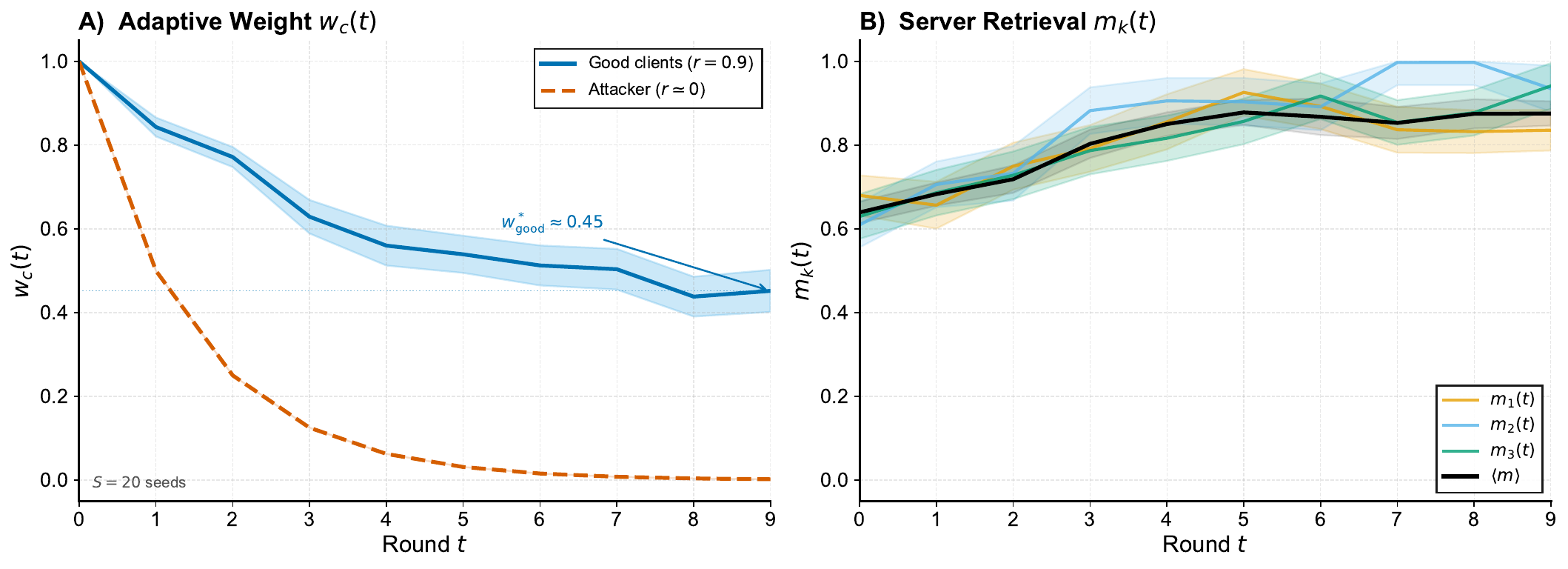}
\caption{%
Per-client adaptive weight dynamics in federated unsupervised learning with adversarial noise.
{(A)}~Temporal evolution of the adaptive weight $w_c(t)$ (Eq.~\ref{eq:convex_comb}) for good clients (blue, $r=0.9$) and one attacker client receiving pure noise (vermillion dashed, $r\simeq0$).
The weight update rule is based on normalized sign-agreement between the client's local Hebbian correlator $J^{(t)}_c$ and the server's reconstruction operator from the previous round (Section~\ref{subsec:plasticity_w}).
Good clients maintain stable non-zero weights ($w_{\mathrm{good}} \approx 0.4$--$0.6$), reflecting high-quality local data.
The attacker's weight rapidly converges to zero ($w_{\mathrm{att}}  \widetilde{\to} 0$ within $\sim$10 rounds), effectively down-weighting noisy contributions in the server aggregation.
Shaded regions represent standard error over $S=20$ independent seeds.
{(B)}~Server-side retrieval quality: per-archetype magnetization $m_k(t)$ (colored lines, $k=1,2,3$) and mean retrieval $\langle m \rangle$ (black line).
The system maintains high reconstruction fidelity ($\langle m \rangle \gtrsim 0.7$) despite the presence of one attacker, demonstrating robustness of the adaptive weighting scheme.
Parameters: $N=1000$ neurons, $K=3$ archetypes, $L=5$ clients (4 good + 1 attacker), $T=10$ rounds, $M=800$ examples/client/round, $\alpha_{\mathrm{EMA}}=0.5$.
}
\label{fig:client_aware_adaptive_weight}
\end{figure}

To validate the per-client adaptive weight mechanism described in Section~\ref{subsec:plasticity_w}, 
we conducted a federated unsupervised learning experiment in which one client (the \emph{attacker}) 
receives pure noise ($r\simeq0$), while the remaining clients observe data drawn from the same latent 
archetypes at high quality ($r=0.9$).
Figure~\ref{fig:client_aware_adaptive_weight}A shows the temporal evolution of the adaptive 
weight $w_c(t)$ (Eq.~\ref{eq:convex_comb}) for both populations.
The weight update rule, based on normalized sign-agreement between the client's local 
Hebbian matrix and the server's reconstruction operator, successfully discriminates signal from noise: 
good clients stabilize at $w_{\mathrm{good}} \approx 0.4$--$0.6$, balancing local evidence with 
server guidance, while the attacker's weight collapses to near-zero within $\sim$10 rounds.

Importantly, this down-weighting of noisy clients preserves server-side reconstruction quality 
(Figure~\ref{fig:client_aware_adaptive_weight}B).
Despite the presence of one attacker among five clients ($20\%$ contamination), 
the mean retrieval $\langle m \rangle$ remains above $0.8$ throughout training, indicating that 
the server's TAM dynamics successfully disentangle the $K=3$ latent 
archetypes from the aggregated, but adaptively weighted—Hebbian correlators.
Per-archetype magnetizations $m_k(t)$ (colored curves) exhibit stable convergence, with minor 
fluctuations attributable to stochastic sampling and the non-convex TAM energy landscape.

This experiment demonstrates that \emph{local} sign-agreement entropy, computed independently 
by each client without access to ground-truth labels, provides a reliable signal for data quality, 
enabling robust federated learning even under adversarial noise injection.
The approach generalizes naturally to heterogeneous client populations with varying $r_c$ 
(partial observations, label noise, distribution shift), offering a principled mechanism for 
\emph{plasticity-stability balance} in decentralized unsupervised settings.

\section{Numerical simulations}
\label{sec:numerics}

We now empirically validate the proposed federated associative-memory pipeline and the theoretical mechanisms developed in Secs.~\ref{subsec:bbp}--\ref{subsec:plasticity_w}. Unless otherwise stated, clients follow the pipeline of Subsec.~\ref{subsec:model-pipeline}: at each round they upload local correlators, the server aggregates them into $\bm J_s^{(t)}$, applies spectral sharpening and LAM-based factorization (Subsec.~\ref{subsec:factorization_procedure}), and broadcasts back the reconstructed operator $\hat{\bm J}^{(t)}_s$ used by clients in the convex fusion update~\eqref{eq:convex_comb}. Reconstruction quality is assessed via the magnetization~\eqref{eq:mag}. 

We consider three progressively more challenging regimes: (i) non-stationary \emph{exposure drift} and the stability--plasticity trade-off controlled by $w_c(t)$; (ii) \emph{novelty emergence}, where new archetypes enter the federation mid-training; and (iii) \emph{structured data environments}, where we evaluate the approach on structured datasets.

\begin{figure}[t!]
    \centering
    \includegraphics[width=1\linewidth]{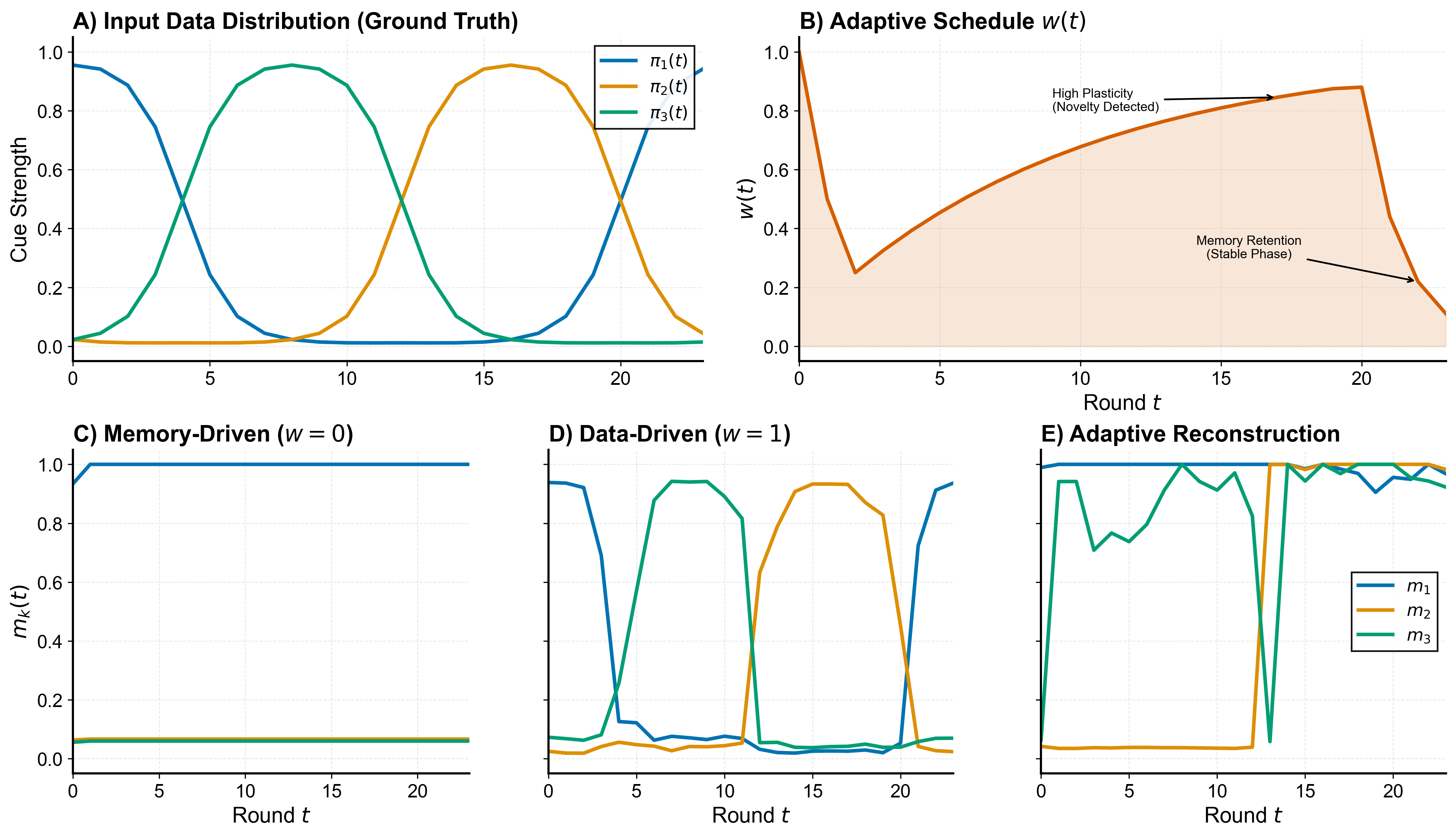}
    \caption{{Adaptive plasticity resolves the stability--plasticity dilemma under exposure drift.}
    $A)$ Ground-truth round-level mixture $\pi_t(\mu)$ (cf.~\eqref{eq:global_mixture_round}) in a sequential schedule where archetypes become dominant one after another. $B)$ Adaptive client weight $w_c(t)$ computed from the entropy controller in~\eqref{eq:w_entropy_rule} and smoothed via~\eqref{eq:w_ema}. Peaks align with mixture transitions (novelty/shift), while plateaus yield small $w_c(t)$ (consolidation).
    $C)$ Memory-driven limit ($w_c(t)\equiv 0$): the operator update ignores new evidence and the system fails to acquire later archetypes ($m_2,m_3\simeq 0$).
    $D)$ Data-driven limit ($w_c(t)\equiv 1$): the update is overly plastic and magnetizations track the instantaneous mixture, leading to rapid forgetting when exposures decrease.
    $E)$ Adaptive rule: newly exposed archetypes are acquired while previously consolidated ones remain retrievable, yielding sustained high magnetization across all modes.}
    \label{fig:adaptive}
\end{figure}


\subsection{Adaptive plasticity under exposure drift}
\label{subsec:drift}

We first study continual learning under \emph{exposure drift}, i.e.\ a time-varying round-level class mixture $\pi_t\in\Delta^{K-1}$ (Subsec.~\ref{subsec:dataset_partition}). In this experiment we use $L=3$ clients and $K=3$ archetypes. At each round, clients upload correlators and the server forms $\bm J_s^{(t)}$; at the population level, $\mathbb{E}[\bm J_s^{(t)}]$ follows the spiked decomposition~\eqref{eq:spiked-decomp} with spike strengths controlled by exposure. Hence, as the mixture $\pi_t$ shifts, the spectral evidence supporting each archetype shifts accordingly: when an archetype gains exposure, its spike becomes supercritical and produces an outlier/eigenvector alignment as described by the BBP theory (Thm.~\ref{thm:detection-thresholds}); when exposure wanes, the same direction becomes weakly supported by fresh data and risks being overwritten by overly plastic updates.

This is precisely the role of the convex fusion update~\eqref{eq:convex_comb}. For clarity, we denote by $\bm A^{(t-1)}$ the operator broadcast by the server at the previous round (constructed from reconstructed archetypes as in~\eqref{eq:server_reencoded}), and by $\tilde{\bm J}_c^{(t)}$ the client correlator from the current batch. The weight $w_c(t)$ trades consolidation versus adaptation (Subsec.~\ref{subsec:plasticity_w}). In Fig.~\ref{fig:adaptive}A, exposures shift in stages, inducing bursts of non-stationarity separated by plateaus.

The two fixed-$w$ extremes illustrate complementary failure modes. In the \emph{memory-driven} regime ($w_c(t)\equiv 0$), early consolidated structure dominates and subsequent exposures fail to reshape the operator, preventing acquisition of later archetypes (Fig.~\ref{fig:adaptive}-C). Conversely, in the \emph{data-driven} regime ($w_c(t)\equiv 1$), the system becomes overly reactive to the current batch: magnetizations quickly decay when an archetype stops being prevalent, producing catastrophic forgetting (Fig.~\ref{fig:adaptive}-D).

The adaptive entropy controller~\eqref{eq:w_entropy_rule} resolves this trade-off. As shown in Fig.~\ref{fig:adaptive}-B, $w_c(t)$ spikes during transitions (large sign inconsistency between consolidated memory and incoming data) and relaxes during stable phases (consistency up to the intrinsic noise floor $H_{\min}(r)$). The resulting dynamics (Fig.~\ref{fig:adaptive}-E) retains previously learned archetypes while still incorporating newly dominant ones. In other words, exposure governs \emph{when} an archetype becomes spectrally learnable (via BBP outliers), while adaptive $w_c(t)$ governs \emph{how strongly} the update should trust current data versus consolidated memory to prevent overwrite.

\medskip
\noindent
We next move from drift in prevalences to genuine \emph{class emergence}, where the representational subspace must expand when new archetypes become exposed.

\subsection{Continual learning with archetype emergence}
\label{subsec:novelty}

We now intensify non-stationarity by letting new archetypes enter the federation mid-training. The system starts with $K_{\mathrm{old}}=3$ archetypes. At round $t_{\mathrm{intro}}=12$, $K_{\mathrm{new}}=3$ additional archetypes are introduced through a four-round ramp in the mixture, after which all $K_{\mathrm{tot}}=K_{\mathrm{old}}+K_{\mathrm{new}}=6$ modes have non-zero exposure. Clients follow the same pipeline without replay; therefore, successful adaptation hinges on engaging plasticity when novel directions become informative, while maintaining stability on previously consolidated attractors.

\begin{figure}[t!]
    \centering
    \includegraphics[width=1\linewidth]{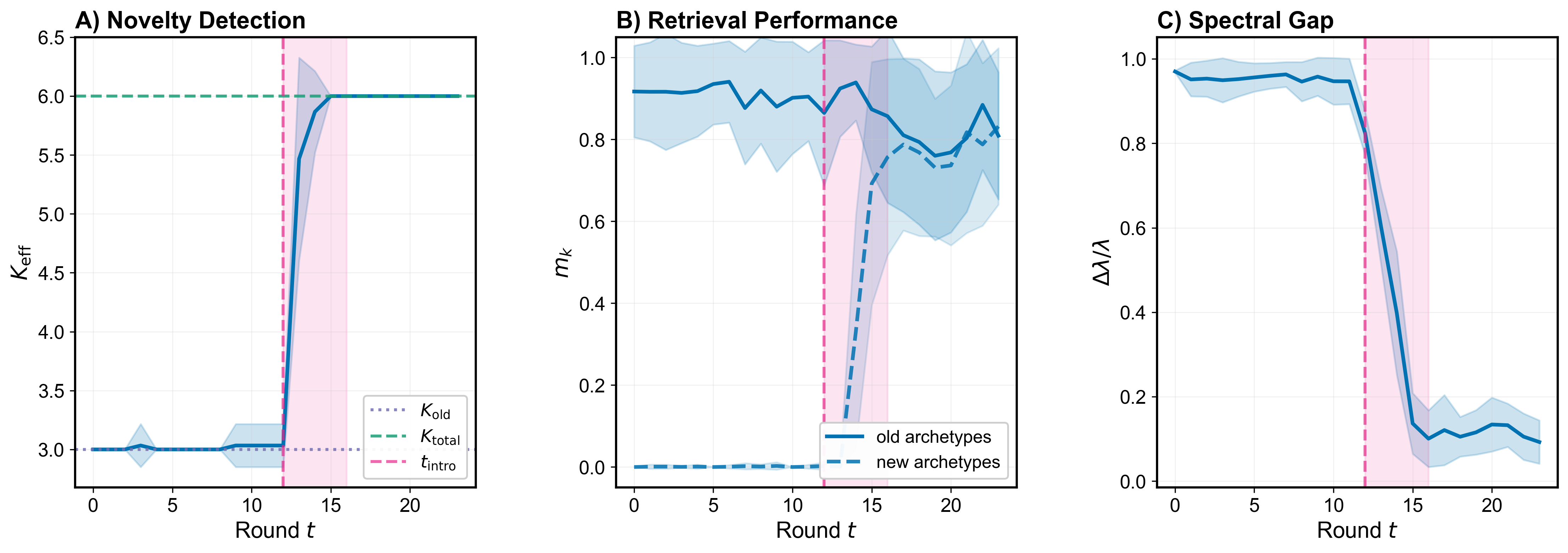}
    \caption{Detecting and integrating novel archetypes in a federated setting.
    Setup: $L=3$ clients, $N=400$, $T=24$, example quality $r=0.8$, and $M_c^t$ fixed across rounds. The federation starts with $K_{\mathrm{old}}=3$ archetypes; $K_{\mathrm{new}}=3$ additional archetypes are introduced at $t_{\mathrm{intro}}=12$ via a four-round ramp (vertical dashed line and shaded interval). Curves are averaged over 30 random seeds (shaded regions: standard deviation).
    \textit{(Left)} Effective dimensionality $\hat K(t)$ estimated by eigenvalue thresholding (cf.~\eqref{eq:keff-op}). The system rapidly transitions from $\hat K\simeq 3$ to $\hat K\simeq 6$ within a few rounds after the introduction, indicating low-latency novelty detection.
    \textit{(Center)} Magnetizations (overlaps~\eqref{eq:mag}) for old archetypes (solid) and newly introduced ones (dashed). New archetypes are acquired during/shortly after the ramp, while old archetypes exhibit only a mild transient dip and recover, consistent with targeted (rather than indiscriminate) plasticity.
    \textit{(Right)} Relative spectral gap at the $K_{\mathrm{old}}$ boundary, $(\lambda_{K_{\mathrm{old}}}-\lambda_{K_{\mathrm{old}}+1})/\lambda_{K_{\mathrm{old}}}$, computed from the (sharpened) server operator. The gap collapses during the ramp, signaling the breakdown of the $K_{\mathrm{old}}$-dimensional hypothesis, and stabilizes once the representation has expanded to $K_{\mathrm{tot}}=6$ directions.}
    \label{fig:novelty_detection}
\end{figure}

Figure~\ref{fig:novelty_detection} summarizes the key dynamics. The effective rank estimate $\hat K(t)$ (left) increases from $3$ to $6$ within a few rounds after $t_{\mathrm{intro}}$, reflecting the exposure-driven BBP mechanism: as the new archetypes gain exposure, their spikes become detectable and contribute additional outliers/eigenvectors (Thm.~\ref{thm:detection-thresholds}). The overlap curves (center) show that newly introduced archetypes are learned quickly, while previously consolidated ones remain retrievable, exhibiting only a transient decrease around the transition. Finally, the spectral-gap diagnostic (right) provides an architecture-agnostic signature of novelty: the gap shrinks when the old subspace ceases to be sufficient and stabilizes again once the representation has expanded.

Taken together, these results support the central claim of the paper: the system expands its representational subspace precisely when new directions become informative (exposure-driven detectability), and it does so with limited interference thanks to the stability--plasticity control implemented by $w_c(t)$.

\subsection{Application to structured datasets}
\label{app:structured}

\begin{figure}[t]
    \centering
    \begin{subfigure}[b]{0.20\textwidth}
        \centering
        \includegraphics[width=\textwidth]{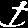}
        \caption{Anchor}
        \label{fig:img1}
    \end{subfigure}
    \hfill 
    \begin{subfigure}[b]{0.20\textwidth}
        \centering
        \includegraphics[width=\textwidth]{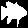}
        \caption{Bass}
        \label{fig:img2}
    \end{subfigure}
    \hfill
    \begin{subfigure}[b]{0.20\textwidth}
        \centering
        \includegraphics[width=\textwidth]{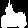}
        \caption{Motorbike}
        \label{fig:img3}
    \end{subfigure}
    
    \caption{The $K=3$ archetype utilized in the structured experiment.}
    \label{fig:tre_immagini}
\end{figure}

The proposed pipeline was validated on \emph{unstructured}, randomly generated archetypes whose entries are independent and uniformly distributed over $\{-1,+1\}$. A natural question is whether the same federated factorization and reconstruction procedure can handle \emph{structured} patterns, where pixels are spatially organized and the resulting archetypes are mutually correlated.

To investigate this question we select $K=3$ binary silhouettes (see Fig. \ref{fig:tre_immagini}) from the Caltech-101 silhouettes dataset, a collection of $28\times 28$ black-and-white object outlines ($N=HW=784$ pixels).

Once the archetypes $\{\bm\xi^\mu\}_{\mu=1}^K$ are obtained, noisy examples are generated through the same channel used for random patterns~\eqref{eq:gen}: each observed sample $\bm\eta$ is produced by independently flipping each bit of a randomly chosen archetype with probability $(1-r)/2$ (see the left column in Fig. \ref{fig:structured_results}). The full federated pipeline of Subsec.~\ref{subsec:model-pipeline} is then executed: $L=3$ clients upload local Hebbian correlators at each round; the server aggregates, applies spectral sharpening and LAM-based factorization, and broadcasts reconstructed operators; clients fuse new local evidence with the broadcast operator via the convex combination rule~\eqref{eq:convex_comb} with fixed weight $w=0.6$.

We run $T=12$ communication rounds for $10$ independent random seeds and measure reconstruction quality through the magnetization~\eqref{eq:mag}. To probe the effect of data quality, we consider two distinct noise regimes: a \emph{moderate}-quality setting with $r=0.6$ (each example has $\approx 20\%$ flipped bits) and a \emph{high}-quality setting with $r=0.8$ ($\approx 10\%$ flipped bits).

Figure~\ref{fig:structured_results} presents the outcome for both noise levels side by side. Each panel displays, from left to right: representative noisy input examples, the evolution of per-archetype magnetizations $m_\mu(t)$ averaged over the $10$ seeds (shaded bands: $\pm 1\sigma$), and the final reconstructed patterns.

In the high-quality regime ($r=0.8$, right panel) all three archetypes are recovered with high fidelity: magnetizations rise quickly above $0.9$ and remain stable throughout the $12$ rounds. The reconstructed images are visually recognizable as the original silhouettes with only minor pixel-level noise.

In the moderate-quality regime ($r=0.6$, left panel) reconstruction is more challenging. The noisier examples reduce the signal-to-noise ratio of each local correlator, slowing convergence and increasing run-to-run variability, as reflected in the wider $\pm\sigma$ bands. Nevertheless, the pipeline still achieves non-trivial magnetizations, confirming that the federated aggregation and spectral sharpening stages provide sufficient denoising to recover structured patterns even when individual examples are substantially corrupted.

\begin{figure}[t!]
    \centering
    \begin{minipage}[t]{0.49\textwidth}
        \centering
        \includegraphics[width=\textwidth]{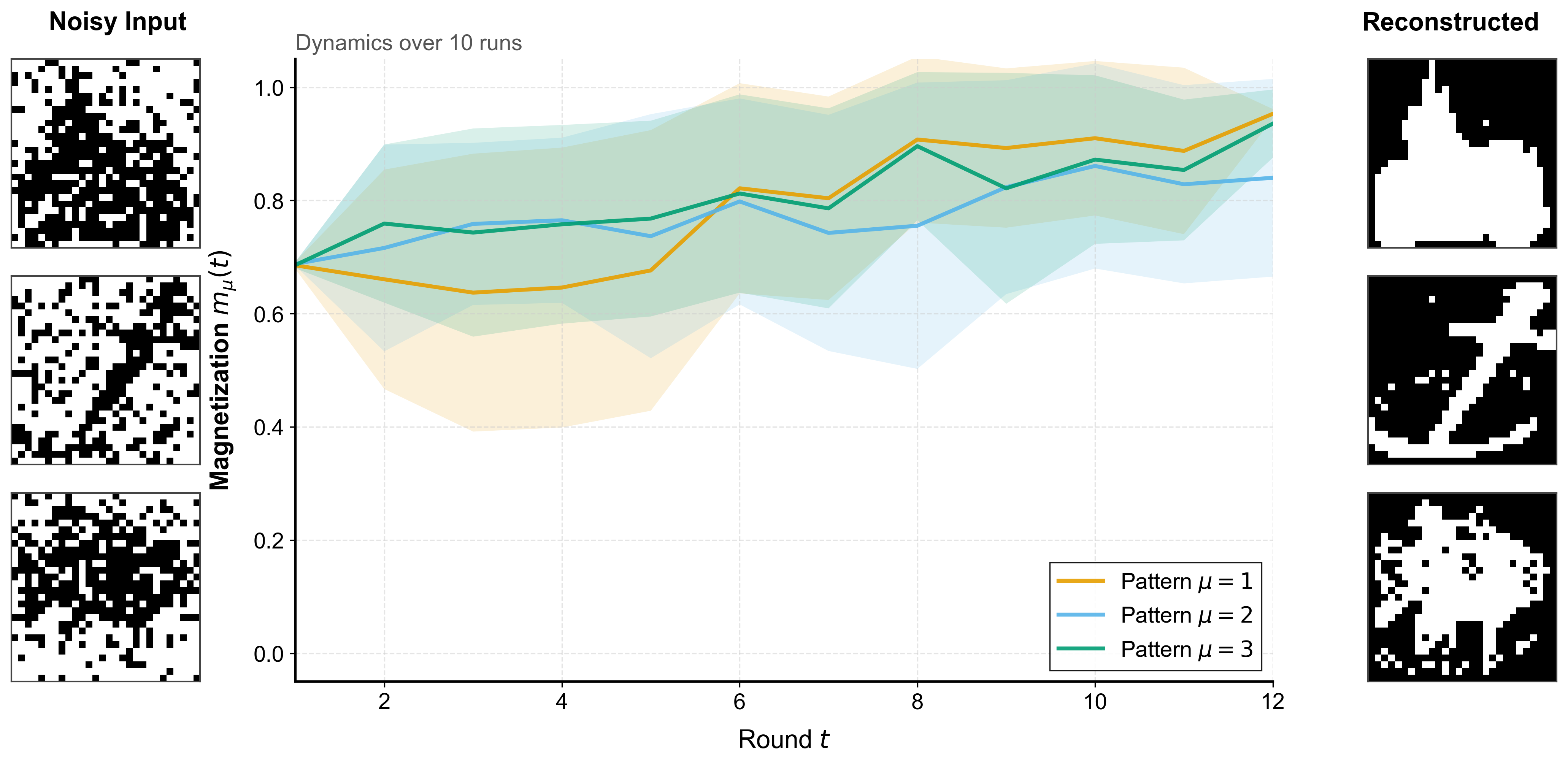}
    \end{minipage}
    \hfill
    \begin{minipage}[t]{0.49\textwidth}
        \centering
        \includegraphics[width=\textwidth]{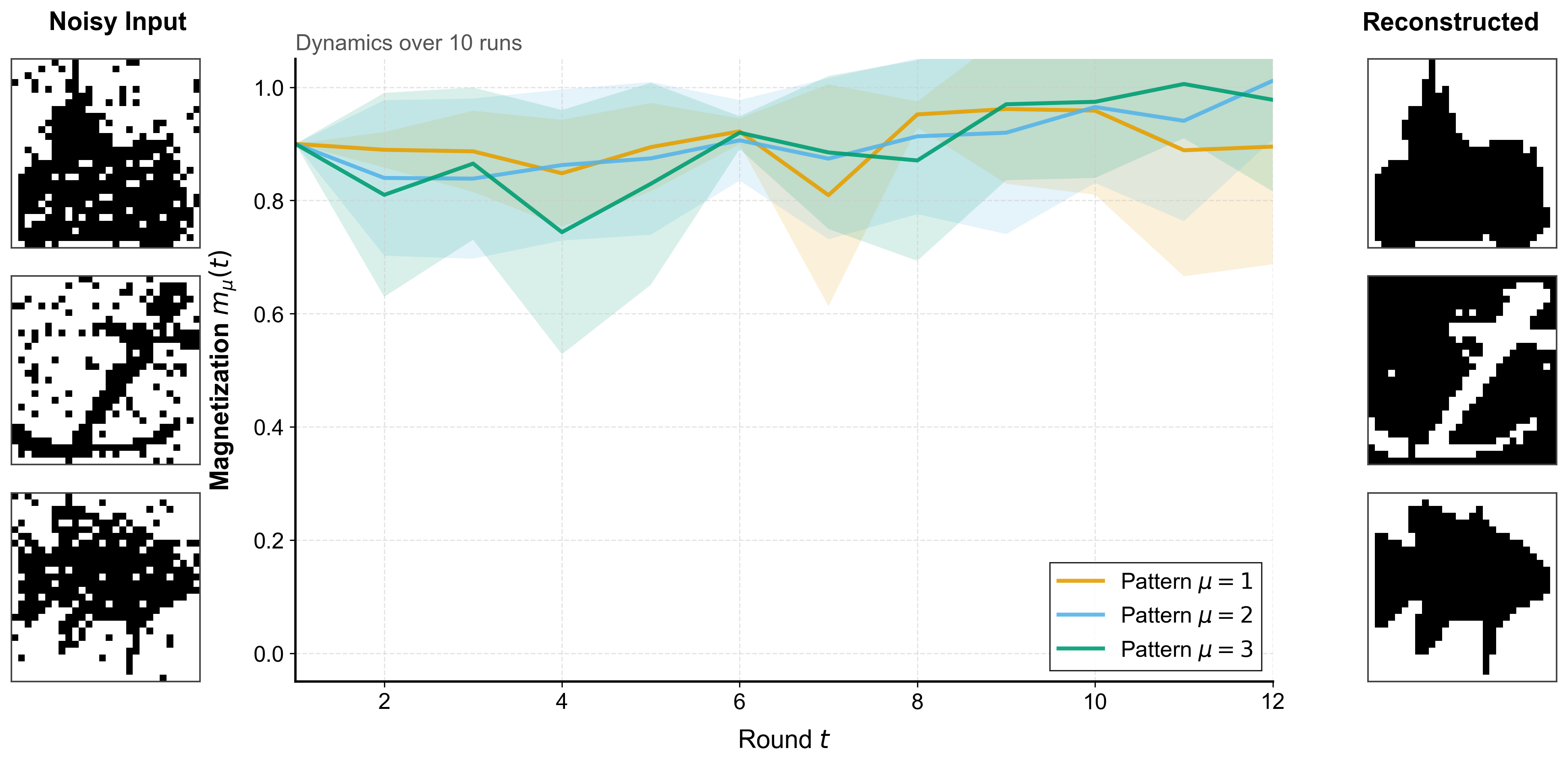}
    \end{minipage}
    \caption{Federated reconstruction of structured archetypes from Caltech-101 silhouettes.
    \textit{Left panel:} moderate-quality examples ($r=0.6$).
    \textit{Right panel:} high-quality examples ($r=0.8$).
    In each panel, the leftmost column shows representative noisy inputs, the central plot tracks the per-archetype magnetization $m_\mu(t)$ (mean $\pm$ std over $10$ seeds; cf.~\eqref{eq:mag}), and the rightmost column displays the final reconstructed archetypes. Despite the spatial correlations inherent in natural silhouettes, the federated pipeline recovers all three patterns; higher example quality leads to faster convergence and tighter confidence bands.
    Parameters: $N=784$, $L=3$, $K=3$, $T=12$, $M_{\mathrm{total}}=4000$, $w=0.6$.}
    \label{fig:structured_results}
\end{figure}

These results confirm that the federated TAM pipeline is not restricted to idealized, uncorrelated binary patterns.  The comparison between $r=0.6$ and $r=0.8$ further illustrates the role of example quality: as predicted by the theoretical analysis, higher $r$ strengthens the signal spikes in the aggregated operator, leading to faster and more reliable archetype recovery.

We note that this experiment primarily validates the pipeline on \emph{structured} data with recognizable spatial structure; it does not address the full complexity of high-dimensional image distributions. Extending the framework to richer, multi-scale image representations remains an interesting direction for future investigation.

\section{Conclusion}
\label{sec:conclusion}

We proposed a federated associative-memory framework for archetype learning under heterogeneous data, client drift, and streaming novelty, in regimes that are naturally modeled as \emph{independent but non-identically distributed} (i.n.i.d.) across clients.
By communicating compact Hebbian operators rather than raw samples, clients contribute privacy-compatible sufficient statistics that the server can aggregate and factorize to reconstruct global archetypes.
On the theoretical side, framing aggregation as a low-rank-plus-noise spectral problem enables random-matrix predictions of detectability and robustness: global archetypes emerge reliably once signal eigen-structure separates from heterogeneity-induced noise, providing a principled view of when federated consolidation is feasible under independence and heterogeneity.

Empirically, our method exhibits improved archetype reconstruction and retrieval stability across heterogeneous clients and drifting regimes, and the proposed entropy-based controller provides a practical mechanism to balance stability and plasticity without centralized replay.

Several directions remain open.
First, extending the reconstruction step beyond linear Hebbian operators to kernelized or deep-feature variants could enlarge the class of archetypes representable while preserving the operator-communication interface.
Second, integrating formal privacy guarantees (e.g.\ differential privacy noise calibrated at the operator level) and analyzing the induced spectral degradation would clarify the privacy--utility frontier.
Third, asynchronous and partial participation regimes may be studied through time-varying random-matrix models, connecting communication constraints to phase transitions in archetype detectability.
We hope that this work stimulates further connections between federated learning, associative memories, and spectral theory for continual, privacy-aware representation learning.

\section*{Acknowledgments}
A.A. acknowledges UniSalento for support via PhD-AI.
\\
A.A. acknowledges BULBUL “Brain-inspired ULtra-Fast \& Ultra-sharp neural networks” for support via post-Lauream research fellowships 
“Statistical mechanics of hetero-associative neural networks” (CUP F85F21006230001, D.D. n. 325/29-09-2025).
\\
A.L., A.L. and F.D. acknowledge Università degli Studi di Bari Aldo Moro for support via post-Lauream research fellowships under the ADAPTIVE‑AI project (CUP F83C24001470001, D.R. n. 123/16-01-2024) and via the Future Artificial Intelligence Research (FAIR) project (project code PE00000013, CUP H97G22000210007, Spoke 6 “Symbiotic AI”), funded by the European Union---NextGenerationEU. F.D. partially acknowledges Future Artificial Intelligence Research (FAIR) project (project code PE00000013, CUP H97G22000210007, Spoke 6 “Symbiotic AI”), funded by the European Union---NextGenerationEU.
\\
A.A., A.L. and A.L. authors are members of the  GNFM group within INdAM which is acknowledged too. F.D. is member of GNAMPA group within INdAM which is acknowledged too.


\begin{thebibliography}{10}

\bibitem{mcmahan2017fedavg}
H.~Brendan McMahan, Eider Moore, Daniel Ramage, Seth Hampson, and Blaise~Ag{\"u}era y~Arcas.
\newblock Communication-efficient learning of deep networks from decentralized data, 2017.

\bibitem{kairouz2021advances}
Peter Kairouz, H~Brendan McMahan, Brendan Avent, Aur{\'e}lien Bellet, Mehdi Bennis, Arjun~Nitin Bhagoji, Kallista Bonawitz, Zachary Charles, Graham Cormode, Rachel Cummings, et~al.
\newblock Advances and open problems in federated learning.
\newblock {\em Foundations and trends{\textregistered} in machine learning}, 14(1--2):1--210, 2021.

\bibitem{hsu2019measuring}
Tzu-Ming~Harry Hsu, Hang Qi, and Matthew Brown.
\newblock Measuring the effects of non-identical data distribution for federated visual classification.
\newblock {\em arXiv preprint arXiv:1909.06335}, 2019.

\bibitem{li2020fedprox}
Tian Li, Anit~Kumar Sahu, Ameet Talwalkar, and Virginia Smith.
\newblock Federated optimization in heterogeneous networks.
\newblock In {\em Proceedings of MLSys}, 2020.

\bibitem{o2014complementary}
Randall~C O’Reilly, Rajan Bhattacharyya, Michael~D Howard, and Nicholas Ketz.
\newblock Complementary learning systems.
\newblock {\em Cognitive science}, 38(6):1229--1248, 2014.

\bibitem{kumaran2016learning}
Dharshan Kumaran, Demis Hassabis, and James~L McClelland.
\newblock What learning systems do intelligent agents need? complementary learning systems theory updated.
\newblock {\em Trends in cognitive sciences}, 20(7):512--534, 2016.

\bibitem{hopfield1982Neural}
John~J Hopfield.
\newblock Neural networks and physical systems with emergent collective computational abilities.
\newblock {\em Proceedings of the National Academy of Sciences}, 79(8):2554--2558, 1982.

\bibitem{Amit1985SpinGlass}
Daniel~J Amit, Hanoch Gutfreund, and Haim Sompolinsky.
\newblock Spin-glass models of neural networks.
\newblock {\em Physical Review A}, 32(2):1007--1018, 1985.

\bibitem{personnaz1985information}
L.~Personnaz, I.~Guyon, and G.~Dreyfus.
\newblock Information storage and retrieval in spin-glass like neural networks.
\newblock {\em Journal de Physique Lettres}, 46:L359--L365, 1985.

\bibitem{ramsauer2020hopfield}
Hubert Ramsauer, Bernhard Sch{\"a}fl, Johannes Lehner, Philipp Seidl, Michael Widrich, Thomas Adler, Lukas Gruber, Markus Holzleitner, Milena Pavlovi{\'c}, Geir~Kjetil Sandve, et~al.
\newblock Hopfield networks is all you need.
\newblock In {\em International Conference on Learning Representations (ICLR)}, 2020.

\bibitem{baik2005phase}
Jinho Baik, G{\'e}rard Ben~Arous, and Sandrine P{\'e}ch{\'e}.
\newblock Phase transition of the largest eigenvalue for nonnull complex sample covariance matrices.
\newblock {\em The Annals of Probability}, 33(5):1643--1697, 2005.

\bibitem{benaych2011eigenvalues}
Florent Benaych-Georges and Raj~Rao Nadakuditi.
\newblock The eigenvalues and eigenvectors of finite, low rank perturbations of large random matrices.
\newblock {\em Advances in Mathematics}, 227(1):494--521, 2011.

\bibitem{agliari2025networks}
Elena Agliari, Andrea Alessandrelli, Adriano Barra, Martino~Salomone Centonze, and Federico Ricci-Tersenghi.
\newblock Networks of neural networks: more is different.
\newblock {\em arXiv preprint arXiv:2501.16789}, 2025.

\bibitem{Agliari2025MultiChannel}
Elena Agliari, Andrea Alessandrelli, Pedro~D Mour{\~a}o, and Alberto Fachechi.
\newblock Multi-channel pattern reconstruction through {L}-directional associative memories.
\newblock {\em arXiv preprint arXiv:2503.06274}, 2025.

\bibitem{kanter1987associative}
Ido Kanter and Haim Sompolinsky.
\newblock Associative recall of memory without errors.
\newblock {\em Physical Review A}, 35(1):380, 1987.

\bibitem{fachechi2019dreaming}
Alberto Fachechi, Elena Agliari, and Adriano Barra.
\newblock Dreaming neural networks: forgetting spurious memories and reinforcing pure ones.
\newblock {\em Neural Networks}, 112:24--40, 2019.

\bibitem{personnaz1986collective}
L~Personnaz, I~Guyon, and G~Dreyfus.
\newblock Collective computational properties of neural networks: New learning mechanisms.
\newblock {\em Physical Review A}, 34(5):4217, 1986.

\bibitem{tropp2012user}
Joel~A Tropp.
\newblock User-friendly tail bounds for sums of random matrices.
\newblock {\em Foundations of computational mathematics}, 12(4):389--434, 2012.

\bibitem{tropp2015introduction}
Joel~A Tropp et~al.
\newblock An introduction to matrix concentration inequalities.
\newblock {\em Foundations and Trends{\textregistered} in Machine Learning}, 8(1-2):1--230, 2015.

\bibitem{bai1993limit}
Zhi-Dong Bai, Yong-Qua Yin, et~al.
\newblock Limit of the smallest eigenvalue of a large dimensional sample covariance matrix.
\newblock {\em Ann. Probab}, 21(3):1275--1294, 1993.

\bibitem{KnowlesYin2017}
Antti Knowles and Jun Yin.
\newblock Anisotropic local laws for random matrices.
\newblock {\em Probability Theory and Related Fields}, 169(1):257--352, 2017.

\bibitem{AndersonGuionnetZeitouni2010}
Greg~W. Anderson, Alice Guionnet, and Ofer Zeitouni.
\newblock {\em An Introduction to Random Matrices}, volume 118 of {\em Cambridge Studies in Advanced Mathematics}.
\newblock Cambridge University Press, Cambridge, 2010.

\bibitem{PasturShcherbina2011}
Leonid Pastur and Mariya Shcherbina.
\newblock {\em Eigenvalue Distribution of Large Random Matrices}, volume 171 of {\em Mathematical Surveys and Monographs}.
\newblock American Mathematical Society, Providence, RI, 2011.

\bibitem{Erdos2017Dynamical}
L{\'a}szl{\'o} Erd{\H{o}}s and Horng-Tzer Yau.
\newblock {\em A dynamical approach to random matrix theory}, volume~28.
\newblock American Mathematical Soc., 2017.

\bibitem{BaiSilverstein2010}
Zhidong Bai and Jack~W. Silverstein.
\newblock {\em Spectral Analysis of Large Dimensional Random Matrices}.
\newblock Springer, New York, 2nd edition, 2010.

\bibitem{ErdosYauYin2012}
L{\'a}szl{\'o} Erd{\H{o}}s, Horng-Tzer Yau, and Jun Yin.
\newblock Rigidity of eigenvalues of generalized wigner matrices.
\newblock {\em Advances in Mathematics}, 229(3):1435--1515, 2012.

\bibitem{Johnstone2001}
Iain~M. Johnstone.
\newblock On the distribution of the largest eigenvalue in principal components analysis.
\newblock {\em Annals of Statistics}, 29(2):295--327, 2001.

\bibitem{BaikSilverstein2006}
Jinho Baik and Jack~W. Silverstein.
\newblock Eigenvalues of large sample covariance matrices of spiked population models.
\newblock {\em Journal of Multivariate Analysis}, 97(6):1382--1408, 2006.

\bibitem{Paul2007}
Debashis Paul.
\newblock Asymptotics of sample eigenstructure for a large dimensional spiked covariance model.
\newblock {\em Statistica Sinica}, 17(4):1617--1642, 2007.

\bibitem{BenaychGeorges2011Eigenvalues}
Florent Benaych-Georges and Raj~Rao Nadakuditi.
\newblock The eigenvalues and eigenvectors of finite, low rank perturbations of large random matrices.
\newblock {\em Advances in Mathematics}, 227(1):494--521, 2011.

\bibitem{BloemendalErdosKnowlesYauYin2014}
Alex Bloemendal, L{\'a}szl{\'o} Erd{\H{o}}s, Antti Knowles, Horng-Tzer Yau, and Jun Yin.
\newblock Isotropic local laws for sample covariance and generalized wigner matrices.
\newblock {\em Electronic Journal of Probability}, 19(33):1--53, 2014.

\bibitem{Davis1970The}
Chandler Davis and William~M Kahan.
\newblock The rotation of eigenvectors by a perturbation.
\newblock {\em SIAM Journal on Numerical Analysis}, 7(1):1--46, 1970.

\end{thebibliography}

\appendix
\label{sec:app}

\section{Algorithms and pseudocode}

\label{app:pseudocode}

\begin{algorithm}[H]
\caption{Pattern reconstruction in LAM: pre-processing and candidate generation\label{algo:reconstruction_gen}}
\SetKw{KwParam}{Parameters:}
\KwParam{Number of layers $L$, neurons per layer $N$; number of nonlinear mixtures $m$; number of parallel dynamic steps $N_p$; inverse thermal noise $\beta$; eigenvalue threshold $\tau_1$; KS tolerance $\Delta_{\min}=10^{-4}$; KS step $\epsilon$ as in \eqref{eq:iter_JKS}.}\\
\KwIn{Hebbian coupling matrix $\bm J^{\text{Hebb}}=\{J^{\text{Hebb}}_{ij}\}_{i,j=1,\dots,N}$.}
\KwOut{Candidate reconstructed patterns $\mathcal{V}=\{\bm\zeta\}$.}
\BlankLine

\textbf{1. Spectral pre-processing and initialization}\\
Initialize KS iteration:\;
$\bm J \leftarrow \bm J^{\text{Hebb}}$, $\Delta \leftarrow 10^{4}$, $k \leftarrow 0$\;
\While{$\Delta \geq \Delta_{\min}$}{
    Update KS-renormalized coupling:
    $\displaystyle 
    \bm J_{\text{new}} \leftarrow 
    \left(1+\frac{\epsilon}{1+\epsilon k}\right)\bm J
    -\left(\frac{\epsilon}{1+\epsilon k}\right)\bm J^{2}$\;
    Compute update magnitude: $\Delta \leftarrow \|\bm J_{\text{new}}-\bm J\|$\;
    Set $\bm J \leftarrow \bm J_{\text{new}}$ and increment $k \leftarrow k+1$\;
}
Set $\widehat{\bm J}^{KS}\leftarrow \bm J$\;
\BlankLine

Compute $\widehat{\bm J}^{KS}$ eigendecomposition\;
Retain the $\hat K$ eigenvectors $\{\tilde{\bm x}^{\delta}\}_{\delta=1}^{\hat K}$ associated with eigenvalues $\lambda_\delta > \tau_1$ (see~\eqref{eq:JKS_crit})\;
\BlankLine

Generate $m$ nonlinear initial configurations:\;
\For{$\alpha \leftarrow 1$ \KwTo $m$}{
    Sample coefficients $c_\delta^\alpha \sim \mathcal{N}(0,1)$ for $\delta=1,\dots,\hat K$\;
    Construct mixture and binarize:
    $\displaystyle 
    \bm x^\alpha \leftarrow \operatorname{sign}\!\Big(\sum_{\delta=1}^{\hat K} c_\delta^\alpha\,\tilde{\bm x}^{\delta}\Big)$\;
}

\BlankLine
\textbf{2. LAM dynamics and candidate collection}\\
Initialize candidate set $\mathcal{V}\leftarrow \emptyset$\;
\For{$\alpha \leftarrow 1$ \KwTo $m$}{
    Initialize all layers/modules of the LAM model with input $\bm x^\alpha$ and set $t \leftarrow 0$\;
    \While{$t < N_p$}{
        Update network state according to \eqref{eq:update_MC_iclr} with thermal noise $\beta^{-1}$\;
        $t \leftarrow t+1$\;
    }
    Extract final-layer configuration(s) as reconstructed candidate(s) $\bm\zeta$\;
    Append candidate(s): $\mathcal{V} \leftarrow \mathcal{V} \cup \{\bm\zeta\}$\;
}

\Return $\mathcal{V}$\;
\end{algorithm}

\begin{algorithm}[H]
\caption{Pattern reconstruction in LAM: pruning and acceptance filtering\label{algo:reconstruction_pruning}}
\SetKw{KwParam}{Parameters:}
\KwParam{Neurons per layer $N$; overlap threshold $\delta$ (default $0.5$); acceptance threshold $\tau_2$.}\\
\KwIn{Candidate set $\mathcal{V}=\{\bm\zeta\}$; KS-renormalized coupling matrix $\widehat{\bm J}^{KS}$.}
\KwOut{Accepted reconstructed patterns $\mathcal{Z}=\{\bm\hat{\xi}\}$.}
\BlankLine

\textbf{1. Duplicate removal via overlap threshold}\\
Initialize pruned set $\widetilde{\mathcal{V}} \leftarrow \emptyset$\;
\ForEach{$\bm\zeta \in \mathcal{V}$}{
    $\texttt{is\_duplicate} \leftarrow \texttt{false}$\;
    \ForEach{$\bm\zeta' \in \widetilde{\mathcal{V}}$}{
        Compute overlap:
        $$\displaystyle q(\bm\zeta,\bm\zeta') \leftarrow \frac{1}{N}\sum_{i=1}^N \zeta_i\,\zeta'_i$$
        \If{$q(\bm\zeta,\bm\zeta') \ge \delta$}{
            $\texttt{is\_duplicate} \leftarrow \texttt{true}$\;
            \textbf{break}\;
        }
    }
    \If{$\neg\,\texttt{is\_duplicate}$}{
        $\widetilde{\mathcal{V}} \leftarrow \widetilde{\mathcal{V}} \cup \{\bm\zeta\}$\;
    }
}
\BlankLine

\textbf{2. Acceptance test via quadratic score}\\
Initialize accepted set $\mathcal{Z} \leftarrow \emptyset$\;
\ForEach{$\bm\zeta \in \widetilde{\mathcal{V}}$}{
    Compute pattern score under $\widehat{\bm J}^{KS}$:
    $$\displaystyle p(\bm\zeta) \leftarrow \frac{1}{N}\,\bm\zeta^{T}\,\widehat{\bm J}^{KS}\,\bm\zeta$$
    \If{$p(\bm\zeta) \ge \tau_2$}{
        $\hat{\bm\xi}\leftarrow \bm\zeta$\;
        $\mathcal{Z} \leftarrow \mathcal{Z} \cup \{\hat{\bm\xi}\}$\;
    }
}

\Return $\mathcal{Z}$\;
\end{algorithm}

\begin{algorithm}[H]
\caption{Federated M2O pipeline\label{algo:Fede_algo}}
\SetKw{KwParam}{Parameters:}
\KwParam{Number of clients $L$; neurons $N$; communication rounds $T$; batch sizes $\{M_c^{t}\}$; blending weights $\{w_c(t)\in[0,1]\}$; server factorization routine $\mathrm{LAM}(\cdot)$ (Sec.~\ref{subsec:factorization_procedure}).}\\
\KwIn{For each client $c$, a stream of \emph{local} batches $\{\bm\eta_c^{(t)}\}_{t=0}^{T}$, with $\bm\eta_c^{(t)}=\{\bm\eta_{c}^{(t),a}\}_{a=1}^{M_c^{t}}$ and $\bm\eta_{c}^{(t),a}\in\{-1,+1\}^{N}$ (never shared).}
\KwOut{Server reconstructed patterns $\{\hat{\bm\xi}_{s,\mu}^{(T)}\}_{\mu=1}^{\hat{K}}$.}
\BlankLine

\textbf{Initialization (client-side, round $t=0$)}\\
\ForEach{$c \leftarrow 1$ \KwTo $L$ \textbf{in parallel}}{
    Compute local synaptic matrix from the first batch:\;
    $$\displaystyle \big(J_c^{(0)}\big)_{ij} \leftarrow \frac{1}{N\,M_c^{0}}
    \sum_{a=1}^{M_c^{0}} \big(\eta_c^{(0),a}\big)_i\,\big(\eta_c^{(0),a}\big)_j$$
    Send $\bm J_c^{(0)}$ to server\;
}
\BlankLine

\For{$t \leftarrow 0$ \KwTo $T-1$}{
    \textbf{Server aggregation}\\
    Receive $\{\bm J_c^{(t)}\}_{c=1}^{L}$ and compute:
    $$\displaystyle \big(J_s^{(t)}\big)_{ij} \leftarrow \frac{1}{L}\sum_{c=1}^{L}\big(J_c^{(t)}\big)_{ij}$$
    \BlankLine

    \textbf{Server factorization and re-encoding}\\
    Using LAM model to reconstruct patterns underlied the $\bm J_s^{(t)}$:$$\displaystyle \{\hat{\bm\xi}^{(t)}_{s,\mu}\}_{\mu=1}^{\hat K} \leftarrow \mathrm{LAM}\!\left(\bm J_s^{(t)}\right)$$
    Compute server synaptic matrix
    $\displaystyle \big(\hat J_s^{(t)}\big)_{ij} \leftarrow \frac{1}{N}\sum_{\mu=1}^{\hat K}
    \big(\hat\xi^{(t)}_{s,\mu}\big)_i\,\big(\hat\xi^{(t)}_{s,\mu}\big)_j$\;
    Broadcast $\hat{\bm J}_s^{(t)}$ to all clients\;
    \BlankLine

    \textbf{Client fusion (new batch at round $t{+}1$) and upload}\\
    \ForEach{$c \leftarrow 1$ \KwTo $L$ \textbf{in parallel}}{
        Compute operator from the \emph{new} local batch:
        $$\displaystyle \big(J^{\mathrm{loc}}_{c}\big)_{ij} \leftarrow \frac{1}{N\,M_c^{t+1}}
        \sum_{a=1}^{M_c^{t+1}} \big(\eta_c^{(t+1),a}\big)_i\,\big(\eta_c^{(t+1),a}\big)_j$$
        Fuse local evidence with server reconstruction:
        $$\displaystyle \bm J_c^{(t+1)} \leftarrow w_c(t)\,\bm J_c^{\mathrm{loc}} + \big(1-w_c(t)\big)\,\hat{\bm J}_s^{(t)}$$
        Upload $\bm J_c^{(t+1)}$ to server\;
    }
}
\Return $\{\hat{\bm\xi}_{s,\mu}^{(T)}\}_{\mu=1}^{\hat{K}}$ \;
\end{algorithm}

\section{Random Matrix theory digression}
\label{app:threshold}
The blended correlator $J^{(t)}_{\mathrm{rec}}(w_t)$ is informative but still entangles signal with bulk noise. We recall that the true archetypes $\{\bm{\xi}^{\mu}\}_{\mu=1}^K$ are eigenvectors (with eigenvalue equal to $1$) of the pseudo–inverse coupling matrix
\[
(\tilde J_{\mathrm{KS}})_{ij}\;=\;\frac{1}{N}\sum_{\mu,\nu=1}^K \xi_i^\mu\,(C^{-1})_{\mu\nu}\,\xi_j^\nu,
\]
where $C_{\mu\nu}=\frac{1}{N}\sum_{i=1}^N \xi_i^\mu\,\xi_i^\nu$ is the pattern correlation matrix.

\noindent We can obtain the latter coupling matrix as fixed point of the following algorithm \cite{fachechi2019dreaming}
\begin{equation}\label{eq:phi}
\Phi_{\varepsilon}(J)\;:=\;J\;+\;\varepsilon\big(J-J^2\big),\qquad \varepsilon\in(0,1],
\end{equation}
and we compose a small number of micro–steps with a decaying step–size,
\begin{equation}\label{eq:ks-stack}
J^{(t)}_{\mathrm{KS}}
\;:=\;
\Phi_{\varepsilon_{n_{\mathrm{prop}}-1}}\circ\cdots\circ \Phi_{\varepsilon_1}\circ \Phi_{\varepsilon_0}\big(J^{(t)}_{\mathrm{rec}}(w_t)\big),
\qquad
\varepsilon_k=\frac{\varepsilon_0}{1+k\varepsilon_0},\ \ k=0,\dots,n_{\mathrm{prop}}-1.
\end{equation}

In the eigenbasis of $J$, the map \eqref{eq:phi} updates each eigenvalue by a
forward--Euler step of a logistic flow,
\begin{equation}\label{eq:eig-logistic}
\lambda\ \longmapsto\ \lambda+\varepsilon\,\lambda(1-\lambda).
\end{equation}
Eigenvalues in $(0,1)$ are pushed upward toward $1$, eigenvalues larger than $1$
are pulled downward, and values very close to $0$ are almost unchanged. After a
small number of iterations (in practice $n_{\mathrm{prop}}\in[3,10]$ is
sufficient), the gap between the leading spikes and the bulk becomes more
pronounced, which makes the subsequent thresholding step more robust. Once per
round we also enforce symmetry and remove diagonal terms,
\begin{equation}\label{eq:sym-off}
J^{(t)}_{\mathrm{KS}}\ \leftarrow\ \tfrac12\left(J^{(t)}_{\mathrm{KS}}
+J^{(t)\top}_{\mathrm{KS}}\right),\qquad
J^{(t)}_{\mathrm{KS}}\ \leftarrow\ \mathrm{Off}\big(J^{(t)}_{\mathrm{KS}}\big),
\end{equation}
so that the operator remains symmetric and focused on cross--feature
interactions. The map \eqref{eq:phi} is a matrix polynomial and is therefore
cheap to evaluate; when an eigendecomposition is already available (for instance
for thresholding), one can equivalently apply \eqref{eq:eig-logistic} directly
to the eigenvalues.

\vspace{3mm}
\noindent
We then choose a data--dependent threshold $\tau$ according to one of the
following rules and define\footnote{For quick checks in simulations, one may use
a fixed threshold $\tau\in(0,1)$ on the \emph{sharpened} spectrum
$\lambda_i\big(J^{(t)}_{\mathrm{KS}}\big)$ (for example $\tau=0.5$), but all
reported effective ranks and principled counts are based on the data--driven
thresholds below, applied to $J^{(t)}_\mathrm{KS}$ in a way that is consistent
with MP/TW assumptions.}
\begin{equation}\label{eq:keff}
\hat{K}(t)\;:=\;\#\{\,i:\ \lambda_i(J^{(t)}_\mathrm{KS})>\tau\,\}.
\end{equation}

\noindent
\emph{Shuffle threshold.} We construct a null model by randomly reshuffling only
the off--diagonal entries of $J^{(t)}_\mathrm{KS}$ (keeping both symmetry and
the diagonal fixed), compute the empirical distribution of $\lambda_{\max}$ over
several reshuffles, and set $\tau$ to be the $(1-\alpha)$ quantile. This choice
is robust when the form of the noise is uncertain.

\noindent
\emph{Marchenko--Pastur (MP) edge.} We estimate the bulk variance $\sigma_n^2$
from the lower part of the spectrum of $J^{(t)}_\mathrm{KS}$ and the aspect
ratio $q:=N/M_{\mathrm{round}}$. We then use the upper MP edge
\begin{equation}\label{eq:mp-edge}
\lambda_+\;=\;\sigma_n^2\big(1+\sqrt{q}\big)^2,
\end{equation}
optionally with a small finite--sample correction $\delta_N>0$, and set
$\tau=\lambda_++\delta_N$. With this choice, the MP--based effective rank is
\begin{equation}\label{eq:keff-mp}
\hat{K}^{\mathrm{MP}}(t)\;=\;\#\Big\{\,i:\,
\lambda_i\!\big(J^{(t)}_\mathrm{KS}\big)>\lambda_+(q)+\delta_N\Big\},
\end{equation}
which, by Theorem~\ref{thm:detection-thresholds} (D1)--(D2), coincides (up to
finite--sample fluctuations) with the number of supercritical spikes with high
probability. For each detected outlier $\lambda>\tau$, we also estimate its
spike strength via the closed--form inversion of the outlier--location map
(see~\eqref{eq:outlier-loc}):
\begin{equation}\label{eq:kappa-inv}
\widehat\kappa(\lambda)\;=\;\frac{\big(\lambda/\sigma_n^2-1-q\big)
+\sqrt{\big(\lambda/\sigma_n^2-1-q\big)^2-4q}}{2},
\end{equation}
which provides a continuous diagnostic of exposure intensity consistent with
(D5).

\vspace{3mm}
\noindent
Initialization does \emph{not} run once per seed. Instead, we schedule a number of blocks according to a coverage heuristic. Specifically, for a given number of layers $P$, we set
\begin{equation}\label{eq:s-schedule}
s(t)\;=\;10\cdot\Big\lfloor \frac{\hat{K}(t)}{P}\,\log\!\Big(\frac{\hat{K}(t)}{0.01}\Big)\Big\rfloor,
\end{equation}
so that a total of $s(t)\times P$ candidate states are explored before pruning; cf. the balls–into–bins estimate $m_{\min}\!\approx\!(K/L)\log(K/\varepsilon)$ with $\varepsilon=0.01$.

\noindent
Each block $b=1,\dots,s(t)$ selects a (multi)set $\mathcal S_b$ of $P$ seeds from $\{\hat{\bm\xi}^{\mu}_{\mathrm{seed}}\}_{\mu=1}^{\hat{K}(t)}$ (e.g., uniformly or weighted by spectral gaps) and runs the $L$–layer refinement. The resulting raw candidate pool is then
\begin{equation}\label{eq:Xi-raw}
\Xi^{(t)}_{\mathrm{raw}}
\;:=\;
\bigcup_{b=1}^{s(t)} \mathrm{Refine}\!\left(\mathcal S_b;P\right),
\end{equation}
to be subsequently deduplicated and filtered by the acceptance tests.
\medskip

\noindent
The eigenvector–to–seed step is justified by BBP alignment: once exposure crosses the spectral threshold, an outlier splits from the MP bulk and its eigenvector aligns with the archetypal direction; sign–quantized leaders therefore provide consistent seeds for refinement, while the coverage schedule ensures sufficient exploration across classes.

\section{Detailed Proofs of Theorems}\label{appendix}

\subsection{Proof of Theorem~\ref{thm:concentration}}
\begin{proof}
Fix a communication round $t$. Throughout the proof, $t$ denotes this fixed round index, while we use $u>0$ (or $s>0$) for generic deviation levels in the matrix Bernstein bounds.
Recall that the round-wise Hebbian estimator is
\[
\bm{J}_s^{(t)}
:= \frac{1}{N M_{\mathrm{round}}}
   \sum_{\ell=1}^L \sum_{m=1}^{M_c}
   \eta_{\ell,t,m}\,\eta_{\ell,t,m}^\top,
\qquad
M_{\mathrm{round}} := L M_c,
\]
where $\eta_{\ell,t,m}\in\mathbb{R}^N$ denotes the $m$-th local sample at client $\ell$ and round $t$.
For each pair $(\ell,m)$ set
\[
B_{\ell,m} := \eta_{\ell,t,m}\,\eta_{\ell,t,m}^\top
\in\mathbb{R}^{N\times N},
\qquad
\bar B_{\ell} := \mathbb{E}[B_{\ell,m}],
\]
where the expectation is over the sampling of $\eta_{\ell,t,m}$ under the generative model of Sec.~\ref{sec:setting}.
By construction, $\bar B_{\ell}$ does not depend on $m$ but may depend on the client $\ell$, reflecting statistical non-i.i.d.\ behavior across the federation.

\medskip\noindent
Define the centered and rescaled increments
\[
Z_{\ell,m}
:= \frac{B_{\ell,m} - \bar B_{\ell}}{N M_{\mathrm{round}}},
\qquad
S := \sum_{\ell=1}^L \sum_{m=1}^{M_c} Z_{\ell,m}.
\]
Each $Z_{\ell,m}$ is self-adjoint, mean-zero, and the family $\{Z_{\ell,m}\}$ is independent over $(\ell,m)$.
Moreover,
\[
S
= \frac{1}{N M_{\mathrm{round}}} \sum_{\ell,m} B_{\ell,m}
  - \frac{1}{N M_{\mathrm{round}}} \sum_{\ell,m} \bar B_{\ell}
= \bm{J}_s^{(t)} - \mathbb{E}[\bm{J}_s^{(t)}],
\]
so that $\|S\|_{\mathrm{op}} = \|\bm{J}_s^{(t)} - \mathbb{E}[\bm{J}_s^{(t)}]\|_{\mathrm{op}}$.

\medskip\noindent
For each client $\ell$, the population covariance at round $t$ can be written as
\[
\Theta_{t,\ell}
:= \frac{1}{N}\,\bar B_{\ell}
= \sigma^2 I + r^2 \sum_{\mu=1}^K \pi_{t,\ell}(\mu)\,u_\mu u_\mu^\top,
\qquad u_\mu := \xi^\mu/\sqrt{N},
\]
where $\pi_{t,\ell}$ is the client-specific class mixture and $(u_\mu)$ are the normalized archetypes.
The round-level population operator is the average
\[
\Theta_t := \mathbb{E}[\bm{J}_s^{(t)}]
= \frac{1}{N M_{\mathrm{round}}} \sum_{\ell,m} \bar B_{\ell}
= \frac{1}{L} \sum_{\ell=1}^L \Theta_{t,\ell}
= \sigma^2 I + r^2 \sum_{\mu=1}^K \pi_t(\mu)\,u_\mu u_\mu^\top,
\]
with $\pi_t := L^{-1} \sum_{\ell} \pi_{t,\ell}$.
In particular, each $\Theta_{t,\ell}$ and $\Theta_t$ is positive semidefinite, with spectrum contained in some compact interval $[0,\Lambda_\star]$, and
\[
\Lambda_t := \max_{\ell} \lambda_{\max}(\Theta_{t,\ell}) < \infty.
\]

\medskip\noindent
We first control $\|Z_{\ell,m}\|_{\mathrm{op}}$.
Using $\|B_{\ell,m}\|_{\mathrm{op}} = \|\eta_{\ell,t,m}\|_2^2$ and
$\|\bar B_{\ell}\|_{\mathrm{op}} = N\,\lambda_{\max}(\Theta_{t,\ell})$, we obtain
\begin{align*}
\|Z_{\ell,m}\|_{\mathrm{op}}
&= \frac{\|B_{\ell,m} - \bar B_{\ell}\|_{\mathrm{op}}}
         {N M_{\mathrm{round}}}
 \;\le\;
   \frac{\|B_{\ell,m}\|_{\mathrm{op}} + \|\bar B_{\ell}\|_{\mathrm{op}}}
        {N M_{\mathrm{round}}} \\
&\le
   \frac{\|\eta_{\ell,t,m}\|_2^2/N + \lambda_{\max}(\Theta_{t,\ell})}
        {M_{\mathrm{round}}}.
\end{align*}
In the bounded/Rademacher channel (items (ii)–(iii) in the theorem), we have
$\|\eta_{\ell,t,m}\|_2^2 \equiv N$ almost surely, hence
\[
\|Z_{\ell,m}\|_{\mathrm{op}}
\le \frac{1 + \Lambda_t}{M_{\mathrm{round}}}
=: R_t.
\]
In the sub-Gaussian channel (item (i)), the same scaling holds up to a universal constant:
standard sub-Gaussian moment bounds imply
$\|\eta_{\ell,t,m}\|_2^2 \le C N$ with high probability, and a truncation argument (or a sub-exponential matrix Bernstein/Freedman inequality, see e.g.~\cite{tropp2012user, tropp2015introduction}) allows one to absorb $C$ into $R_t$.
This only affects absolute numerical constants and leaves the functional dependence on $M_{\mathrm{round}}$ unchanged.

\medskip\noindent
We now compute the variance parameter
\[
v_t := \Big\|\sum_{\ell,m} \mathbb{E}[Z_{\ell,m}^2]\Big\|_{\mathrm{op}}.
\]
\noindent
We first treat the bounded/Rademacher channel.
Since $B_{\ell,m} = \eta_{\ell,t,m}\eta_{\ell,t,m}^\top$ and
\[
B_{\ell,m}^2
= (\eta_{\ell,t,m}\eta_{\ell,t,m}^\top)^2
= (\eta_{\ell,t,m}^\top\eta_{\ell,t,m})\,\eta_{\ell,t,m}\eta_{\ell,t,m}^\top,
\]
the assumption $\|\eta_{\ell,t,m}\|_2^2 \equiv N$ implies
\[
B_{\ell,m}^2 = N B_{\ell,m}
\quad\Longrightarrow\quad
\mathbb{E}[B_{\ell,m}^2] = N\,\bar B_{\ell}.
\]
Therefore
\begin{align*}
\mathbb{E}\big[(B_{\ell,m} - \bar B_{\ell})^2\big]
&= \mathbb{E}[B_{\ell,m}^2] - \bar B_{\ell}^2
 = N\,\bar B_{\ell} - \bar B_{\ell}^2 \\
&= N^2\big(\Theta_{t,\ell} - \Theta_{t,\ell}^2\big)
 = N^2\,\Theta_{t,\ell}(I-\Theta_{t,\ell}).
\end{align*}
Using $Z_{\ell,m} = (B_{\ell,m} - \bar B_{\ell})/(N M_{\mathrm{round}})$, we get
\[
\mathbb{E}[Z_{\ell,m}^2]
= \frac{1}{N^2 M_{\mathrm{round}}^2}\,
  \mathbb{E}\big[(B_{\ell,m} - \bar B_{\ell})^2\big]
= \frac{1}{M_{\mathrm{round}}^2}\,
  \Theta_{t,\ell}(I-\Theta_{t,\ell}),
\]
and hence
\[
\sum_{\ell,m} \mathbb{E}[Z_{\ell,m}^2]
= \frac{1}{M_{\mathrm{round}}^2}
  \sum_{\ell,m} \Theta_{t,\ell}(I-\Theta_{t,\ell}).
\]
Taking operator norm and using that the spectrum of each $\Theta_{t,\ell}$ is contained in $[0,\Lambda_\star]$, we obtain
\begin{align*}
v_t
&:= \Big\|\sum_{\ell,m} \mathbb{E}[Z_{\ell,m}^2]\Big\|_{\mathrm{op}} \\
&\le \frac{1}{M_{\mathrm{round}}^2}
      \sum_{\ell,m}
      \|\Theta_{t,\ell}(I-\Theta_{t,\ell})\|_{\mathrm{op}} \\
&= \frac{1}{M_{\mathrm{round}}^2}
   \sum_{\ell,m}
   \max_{\lambda\in\operatorname{spec}(\Theta_{t,\ell})} \lambda(1-\lambda) \\
&\le \frac{M_{\mathrm{round}}}{M_{\mathrm{round}}^2}
     \sup_{\lambda\ge 0} \lambda(1-\lambda)
 = \frac{1}{4M_{\mathrm{round}}},
\end{align*}
since the function $\lambda \mapsto \lambda(1-\lambda)$ attains its maximum $1/4$ on $[0,1]$ at $\lambda=1/2$.

\medskip\noindent
In the homogeneous case $\Theta_{t,\ell}\equiv \Theta_t$ for all $\ell$, the computation simplifies and yields the exact identity
\[
\sum_{\ell,m} \mathbb{E}[Z_{\ell,m}^2]
= \frac{1}{M_{\mathrm{round}}}\,\Theta_t(I-\Theta_t),
\qquad
v_t = \frac{1}{M_{\mathrm{round}}}
      \big\|\Theta_t(I-\Theta_t)\big\|_{\mathrm{op}},
\]
so that the bound $v_t\le 1/(4M_{\mathrm{round}})$ follows as a special case.

In the sub-Gaussian channel, an analogous argument applies.
The identity $B_{\ell,m}^2 = (\eta_{\ell,t,m}^\top\eta_{\ell,t,m})\,B_{\ell,m}$ still holds, and standard sub-Gaussian moment bounds give
$\mathbb{E}[\|\eta_{\ell,t,m}\|_2^4] \le C N^2$ for some universal $C$.
This implies a uniform bound
\[
\|\mathbb{E}[B_{\ell,m}^2]\|_{\mathrm{op}}
\le C N^2,
\]
and hence
\[
\big\|\mathbb{E}[(B_{\ell,m}-\bar B_{\ell})^2]\big\|_{\mathrm{op}}
\le C N^2.
\]
It follows that $v_t \le C/M_{\mathrm{round}}$, which again has the same
$1/M_{\mathrm{round}}$ dependence and can be absorbed into a universal constant in the final bound.

\medskip\noindent
We now apply a standard matrix Bernstein inequality for sums of independent self-adjoint matrices (see, e.g., \cite{tropp2012user, tropp2015introduction}).
Let $d:=N$ be the dimension.
For all $u>0$,
\[
\mathbb{P}\big(\|S\|_{\mathrm{op}} > u\big)
\le 2d \exp\!\left(
  -\,\frac{u^2/2}{\,v_t + (R_t u)/3}
\right).
\]
Substituting the bounds on $v_t$ and $R_t$, and using the elementary inequality
\[
\frac{x^2}{a+bx}
\;\ge\; \min\Big\{\frac{x^2}{2a}, \frac{x}{2b}\Big\}
\qquad\text{for all }x,a,b>0,
\]
we obtain, for some universal numerical constant $c_1>0$,
\begin{equation}\label{eq:concentration-last}
    \mathbb{P}\big(\|\bm{J}_s^{(t)} - \mathbb{E}[\bm{J}_s^{(t)}]\|_{\mathrm{op}} > u\big)
\le 2N\,
\exp\Bigg(
  -\,c_1\,M_{\mathrm{round}}\cdot
  \min\Bigg\{
    \frac{u^2}{\displaystyle\sup_{\lambda\in[0,1]} \lambda(1-\lambda)},
    \ \frac{u}{1+\Lambda_t}
  \Bigg\}
\Bigg).
\end{equation}
Since the prefactor $2N$ only enters through a logarithm when the failure probability is inverted (as in the statement of Theorem~\ref{thm:concentration}), it can be absorbed into the numerical constants.
This yields the concentration inequality claimed in the main text, up to universal constant factors.
\medskip
\noindent\textbf{(i) Sub-Gaussian coordinates.}
In the sub-Gaussian channel, we first normalize to unit covariance scale.
Let $\tilde\eta_{\ell,t,m}:=\eta_{\ell,t,m}/\sigma$ and define
\[
\tilde{\boldsymbol{J}}_s^{(t)}
:= \frac{1}{N M_{\mathrm{round}}}
   \sum_{\ell,m}
   \tilde\eta_{\ell,t,m}
   \bigl(\tilde\eta_{\ell,t,m}\bigr)^{\top}
= \sigma^{-2} \boldsymbol{J}_s^{(t)},
\qquad
\tilde\Theta_t := \mathbb{E}\!\left[\tilde{\boldsymbol{J}}_s^{(t)}\right]
= \sigma^{-2}\Theta_t.
\]
By construction, the eigenvalues of $\tilde\Theta_t$ are contained in a compact interval $[0,\Lambda_\star]$ with $\Lambda_\star=\mathcal{O}(1)$, and
\[
\sup_{\lambda\in\operatorname{spec}(\tilde\Theta_t)} \lambda(1-\lambda)\le \tfrac14,
\qquad
1+\lambda_{\max}(\tilde\Theta_t)\le C
\]
for some universal constant $C$ (e.g., $C=2$ if $\lambda_{\max}(\tilde\Theta_t)\le 1$).
Applying~\eqref{eq:concentration-last} to $\tilde{\boldsymbol{J}}_s^{(t)}$ in place of $\boldsymbol{J}_s^{(t)}$ yields
\[
\mathbb{P}\!\left(\,\bigl\|\tilde{\boldsymbol{J}}_s^{(t)}-\mathbb{E}\!\left[\tilde{\boldsymbol{J}}_s^{(t)}\right]\bigr\|_{\mathrm{op}}>s\,\right)
\;\le\;2N\exp\Big(-\,c_2\,M_{\mathrm{round}}\cdot \min\{s^2,\,s\}\Big),
\]
for some constant $c_2>0$.
Since $\boldsymbol{J}_s^{(t)}-\mathbb{E}\!\left[\boldsymbol{J}_s^{(t)}\right]
=\sigma^2\bigl(\tilde{\boldsymbol{J}}_s^{(t)}-\mathbb{E}\!\left[\tilde{\boldsymbol{J}}_s^{(t)}\right]\bigr)$, we have
\[
\bigl\|\boldsymbol{J}_s^{(t)}-\mathbb{E}\!\left[\boldsymbol{J}_s^{(t)}\right]\bigr\|_{\mathrm{op}}>u
\quad\Longleftrightarrow\quad
\bigl\|\tilde{\boldsymbol{J}}_s^{(t)}-\mathbb{E}\!\left[\tilde{\boldsymbol{J}}_s^{(t)}\right]\bigr\|_{\mathrm{op}}>u/\sigma^2.
\]
Substituting $s=u/\sigma^2$ in the previous inequality gives
\[
\mathbb P\left(\,\|\bm{J}_s^{(t)}-\mathbb E[\bm{J}_s^{(t)}]\|_{\mathrm{op}}>u\,\right)
\;\le\;2N\exp\left(-\,c_2\,M_{\mathrm{round}}\cdot
\min\Big\{\tfrac{u^2}{\sigma^{4}},\ \tfrac{u}{\sigma^{2}}\Big\}\right),
\]
which is the claimed form in item~(i) (up to a readjustment of the numerical constant $c_1$).

\medskip
\noindent\textbf{(ii) Bounded / Rademacher channel.}
In the Rademacher channel, $\chi_j\in\{\pm1\}$ with $\mathbb{E}[\chi_j]=r$, and hence
$\|\eta_{\ell,t,m}\|_2^2\equiv N$ almost surely.
As computed above, this implies $B_{\ell,m}^2 = N B_{\ell,m}$ and
\[
\mathbb{E}\big[(B_{\ell,m}-\bar B_\ell)^2\big]
= N^2\,\Theta_{t,\ell}(I-\Theta_{t,\ell}),
\]
so that the variance proxy satisfies
\[
v_t
:= \Big\|\sum_{\ell,m}\mathbb{E}[Z_{\ell,m}^2]\Big\|_{\mathrm{op}}
\le \frac{1}{4M_{\mathrm{round}}},
\]
because $\lambda(1-\lambda)\le 1/4$ for all $\lambda\ge 0$.
Moreover, in this channel we have the uniform bound
$\|Z_{\ell,m}\|_{\mathrm{op}}\le R_t\le c_0/M_{\mathrm{round}}$ for a universal $c_0>0$.
A Hoeffding-type matrix Bernstein inequality for bounded self-adjoint increments then yields a universal constant $c_\ast>0$ such that, for all $u>0$,
\[
\mathbb P\left(\,\|\bm{J}_s^{(t)}-\mathbb E[\bm{J}_s^{(t)}]\|_{\mathrm{op}}>u\,\right)
\;\le\;2N\exp\left(-\,c_\ast\,M_{\mathrm{round}}\cdot \min\Big\{\tfrac{u^2}{\displaystyle\sup_{\lambda\in[0,1]}\lambda(1-\lambda)},\ u\Big\}\right).
\]
Using again $\sup_{\lambda\in[0,1]}\lambda(1-\lambda)=\tfrac14$, we obtain
\[
\mathbb P\left(\,\|\bm{J}_s^{(t)}-\mathbb E[\bm{J}_s^{(t)}]\|_{\mathrm{op}}>u\,\right)
\;\le\;2N\exp\left(-\,\tilde c\,M_{\mathrm{round}}\cdot \min\{u^2,\,u\}\right),
\]
for another universal constant $\tilde c>0$, which is the bound stated in item~(ii) up to harmless changes in numerical constants.

\medskip
\noindent\textbf{(iii) Finite-$N$ refined bound.}
Set $m:=M_{\mathrm{round}}$ and normalize $Y_k:=\eta_k/\sqrt N$, so that
\[
J:=\frac1m\sum_{k=1}^m Y_kY_k^{\top},\qquad
\Sigma:=\mathbb E[J].
\]
Define $X_k:=Y_kY_k^{\top}-\Sigma$, which are self-adjoint, mean-zero, and independent, and note that
\[
J-\Sigma=\frac1m\sum_{k=1}^m X_k.
\]
In the bounded/Rademacher model, $\|Y_k\|_2=1$ almost surely, hence
\[
\|X_k\|_{\mathrm{op}}
\le \|Y_kY_k^{\top}\|_{\mathrm{op}}+\|\Sigma\|_{\mathrm{op}}
\le 1+1=2=:R_\star.
\]
Writing $B:=Y_kY_k^{\top}$, we have $B^2=B$, so
\[
\mathbb {E}[X_k^2]
= \mathbb {E}[(B-\Sigma)^2]
= \mathbb {E}[B^2]-\Sigma^2
= \mathbb {E}[B]-\Sigma^2
= \Sigma-\Sigma^2.
\]
Thus the variance parameter is
\[
V_0:=\Big\|\sum_{k=1}^m \mathbb {E}[X_k^2]\Big\|_{\mathrm{op}}
= m\,v_\star,\qquad
v_\star:=\|\Sigma-\Sigma^2\|_{\mathrm{op}}
= \max_{\lambda\in\operatorname{spec}(\Sigma)}\lambda(1-\lambda)
\le \tfrac14.
\]
A dimension-explicit matrix Bernstein inequality for sums of self-adjoint matrices gives, for any sum deviation $U>0$,
\[
\mathbb P\left(\,\Big\|\sum_{k=1}^m X_k\Big\|_{\mathrm{op}}\ge U\,\right)
\ \le\ 2N\,\exp\left(-\,\frac{U^2}{2V_0+\tfrac{2}{3}R_\star U}\right).
\]
Setting $U=mu$ (where $u$ is the average deviation) and using the elementary inequality
$\tfrac{x^2}{a+bx}\ge \tfrac12\min\{\tfrac{x^2}{a},\,\tfrac{x}{b}\}$,
we obtain universal constants $c,C>0$ such that
\[
\mathbb P\left(\,\|J-\Sigma\|_{\mathrm{op}}\ge u\,\right)
\ \le\ 2\exp\left(-\,c\,m\,\min\Big\{\tfrac{u^2}{v_\star},\,\tfrac{u}{R_\star}\Big\}+C\log N\right),
\]
where $C\log N$ accounts for the dimensional prefactor $2N$.
In our model, $\Sigma=\sigma^2 I + r^2\sum_\mu \pi_t(\mu) u_\mu u_\mu^{\top}$, so $\lambda_{\max}(\Sigma)\le 1$ and
$v_\star\le \min\{\tfrac14,\,1-\sigma^2\}=\min\{\tfrac14,\,r^2\}$,
which yields the finite-$N$ bound stated in item~(iii).
\end{proof}

\begin{remark}\label{rem:concentration-discussion}
The derived bounds reveal that the concentration speed is governed by the spectral mapping $\lambda \mapsto \lambda(1-\lambda)$. Since the population covariance satisfies $\Sigma \succeq \sigma^2 I$, the variance proxy is bounded by $v_\star \le r^2 = 1-\sigma^2$. This highlights a favorable trade-off: in low-noise regimes (large $r$), the eigenvalues push towards the boundaries of the interval $[0,1]$, thereby reducing the effective variance and tightening the concentration.
From a scaling perspective, if the total sample size grows linearly with dimension ($M_{\mathrm{round}} = \Theta(N)$), the deviation probability decays exponentially in $N$, matching the accuracy predicted by local laws in RMT.
Finally, we note that the dimensional term $C \log N$ in the exponent accounts for the ambient dimension; for matrices with rapidly decaying spectra, standard arguments allow replacing $N$ with the intrinsic dimension (effective rank), yielding further refinement.
\end{remark}

\subsection{BBP threshold lemmata and proofs}\label{subapp:bbp-details}

\begin{theorem}[MP universality + isotropic global law]\label{thm:mp-global-law}

Let $Z \in \mathbb{R}^{N \times m}$ be a random matrix with independent entries.
We view $Z$ as a collection of $m$ independent sample vectors (columns) $x_1,\dots,x_m\in\mathbb{R}^N$.
Assume the entries are standardized and sub-Gaussian:
\begin{equation}\label{eq:moments}
\mathbb{E}[Z_{ij}] = 0,\qquad
\mathbb{E}[Z_{ij}^2] = 1,\qquad
\sup_{i,j}\|Z_{ij}\|_{\psi_2} \le B < \infty.
\end{equation}
To facilitate the application of standard universality results, we further assume that
the first four moments of $Z_{ij}$ match those of a standard Gaussian (or that the mismatch is negligible for the precision we target).
Let
\[
S_N := \frac{1}{m} Z Z^\top = \frac{1}{m}\sum_{\mu=1}^m x_\mu x_\mu^\top
\]
be the sample covariance matrix, and let $G(z) := (S_N - z I_N)^{-1}$ be its resolvent.
Set $q_N := N/m$ and assume $q_N \to q \in (0,\infty)$ as $N\to\infty$.
Denote by $\mu_{\mathrm{MP}}$ the Marchenko--Pastur law with aspect ratio $q$, and by $m_{\mathrm{MP}}(z)$ its Stieltjes transform.

Fix $\delta>0$ and a large constant $R > (1+\sqrt{q})^2 + \delta$. Define the bounded spectral domain strictly away from the asymptotic support of $\mu_{\mathrm{MP}}$:
\begin{equation}\label{eq:Ddelta}
\mathcal{D}_\delta
:=
\Big\{ z \in \mathbb{C}:\ \operatorname{dist}\big(z,\operatorname{supp}(\mu_{\mathrm{MP}})\big) \ge \delta, \ |z| \le R \Big\}.
\end{equation}

Then, for any finite collection of unit vectors $\mathcal{V} \subset \mathbb{S}^{N-1}$ with
$|\mathcal{V}| \le N^{C_0}$, there exists $C = C(\delta,q,B,C_0,R)$ such that, for every $D>0$ and all $N$ large enough,
\begin{equation}\label{eq:isotropic-global}
\mathbb{P}\left(
\sup_{z \in \mathcal{D}_\delta}
\max_{a,b \in \mathcal{V}}
\Big| a^\top G(z)\,b - m_{\mathrm{MP}}(z)\,\langle a, b \rangle \Big|
\le C \sqrt{\frac{\log N}{N}}
\right)
\ge 1 - N^{-D}.
\end{equation}
\end{theorem}

\begin{proof}
The argument combines three standard ingredients of modern random matrix theory:
(I) concentration of bilinear forms of the resolvent;
(II) identification of the mean in the Gaussian case via the Dyson equation;
(III) universality of the mean via the Green Function Comparison Theorem.
A final step (IV) upgrades pointwise bounds to uniform ones in $z$ and over $\mathcal{V}$ using rigidity and an $\varepsilon$--net.

Throughout, constants $C,c>0$ may change from line to line but depend only on $(\delta,q,B,C_0,R)$.

\medskip
\noindent
Fix deterministic unit vectors $a,b\in\mathbb{R}^N$ and a spectral parameter $z \in \mathcal{D}_\delta$.
We first control the random fluctuations of the bilinear form
\[
F(z) := a^\top G(z)\,b
\]
around its mean $\mathbb{E}F(z)$.

A prerequisite for concentration is the boundedness of the resolvent norm. By the Bai--Yin theorems \cite[Thms. 1 \& 2]{bai1993limit} (or the stronger rigidity estimates in Knowles--Yin~\cite[Lemma 10.1]{KnowlesYin2017}), with high probability the spectrum of $S_N$ is contained in a small neighborhood of the asymptotic support. Consequently, for $z \in \mathcal{D}_\delta$, the map $Z \mapsto a^\top G(z) b$ is Lipschitz continuous with constant depending on $\delta^{-1}$.

Conditional on this high-probability event, we apply standard concentration inequalities for functions of independent sub-Gaussian variables (see Anderson--Guionnet--Zeitouni~\cite[Sec.~2.3 \& 4.4]{AndersonGuionnetZeitouni2010} or the general Herbst's argument). This yields a sub-Gaussian bound of the form:
\begin{equation}\label{eq:conc-step}
\mathbb{P}\Big( \big| a^\top G(z) b - \mathbb{E}[a^\top G(z) b] \big| > t \Big)
\;\le\; 2 \exp\!\big( - c\,N\,\min(t,t^2)\big)
\end{equation}
for all $t>0$, with $c>0$ independent of $N$.
Choosing
\[
t = C \sqrt{\frac{\log N}{N}}
\]
and enlarging $C$ if necessary yields
\[
\big| a^\top G(z) b - \mathbb{E}[a^\top G(z) b] \big|
\;\le\; C \sqrt{\frac{\log N}{N}}
\]
with probability at least $1-N^{-D}$, for any prescribed $D>0$, provided $N$ is large enough.
This step relies on the independence and tail properties of the entries, without requiring gaussianity.

\medskip
\noindent
We now specialize to the case where $Z$ has i.i.d.\ $\mathcal{N}(0,1)$ entries, and denote by $G^{\mathrm{G}}(z)$ the corresponding resolvent.
By rotational invariance of the Gaussian measure,
\[
\mathbb{E}\,G^{\mathrm{G}}(z) = m_N(z)\,I_N,
\qquad
m_N(z):=\frac{1}{N}\,\mathbb{E}\operatorname{Tr}G^{\mathrm{G}}(z).
\]

It is a standard consequence of the resolvent method for sample covariance matrices (see, e.g., Pastur--Shcherbina~\cite[Sec.~7.2 and 7.6]{PasturShcherbina2011}) that $m_N(z)$ satisfies the Marchenko--Pastur fixed-point equation up to an error of order $O(1/N)$, uniformly on domains at positive distance from the MP support.
More precisely, for $z\in\mathcal{D}_\delta$,
\[
1 + z\,m_N(z) - \frac{q}{1+m_N(z)} = O\!\left(\frac{1}{N}\right),
\]
where the implicit constant depends on $(\delta,q)$ but not on $N$.
Since the exact solution $m_{\mathrm{MP}}(z)$ of
\(
1+z m_{\mathrm{MP}}(z) = \frac{q}{1+m_{\mathrm{MP}}(z)}
\)
is stable under small perturbations of the equation on $\mathcal{D}_\delta$ (the derivative of the defining RHS does not vanish away from the spectral edges), we obtain
\[
|m_N(z) - m_{\mathrm{MP}}(z)| \le \frac{C}{N}
\qquad \text{for all } z \in \mathcal{D}_\delta.
\]
Hence, for any unit vectors $a,b$,
\begin{equation}\label{eq:gauss-mean}
\big|\mathbb{E}[a^\top G^{\mathrm{G}}(z)b] - m_{\mathrm{MP}}(z)\,\langle a,b\rangle\big|
= |\langle a,b\rangle|\cdot|m_N(z)-m_{\mathrm{MP}}(z)|
\le \frac{C}{N}.
\end{equation}
This deterministic error is asymptotically negligible compared to the fluctuation scale
$N^{-1/2}\sqrt{\log N}$ from Step~I.

\medskip
\noindent
We now remove the Gaussian assumption and return to a general matrix $Z$ satisfying
\eqref{eq:moments} and the moment-matching condition up to order~4.
Let $G(z)$ and $G^{\mathrm{G}}(z)$ denote the resolvents of $S_N$ in the non-Gaussian and Gaussian cases respectively.

We appeal to the Green Function Comparison Theorem (GFT). While originally formulated for Wigner matrices, the result extends to sample covariance matrices, as detailed in Erd\H{o}s--Yau~\cite[Thm~16.1 and Remark~16.2]{Erdos2017Dynamical}. 
(See also Knowles--Yin~\cite{KnowlesYin2017} for an alternative approach via continuous interpolation and self-consistent comparison).
The GFT asserts that if two ensembles have independent entries with matching first four moments and suitable tail bounds,
then expectations of smooth functionals of their resolvents differ asymptotically by a small power of $N^{-1}$.

In particular, applying this to the bilinear form functional (which is smooth away from the real axis), for each fixed $z$ with $\Im z>0$ and unit vectors $a,b$, we have:
\begin{equation}\label{eq:gft}
\big| \mathbb{E}[a^\top G(z) b] - \mathbb{E}[a^\top G^{\mathrm{G}}(z) b] \big|
\;\le\; C N^{-c}
\end{equation}
for some $c=c(q,B)>0$.
The precise value of $c$ is immaterial for our purposes; the key point is that the error vanishes as $N\to\infty$, uniformly on bounded $z$–domains away from the spectrum.

Combining the Gaussian mean computation \eqref{eq:gauss-mean} with the comparison estimate \eqref{eq:gft}, we deduce that, for any fixed $z\in\mathcal{D}_\delta$ and unit $a,b$,
\begin{equation}\label{eq:mean-total}
\big|\mathbb{E}[a^\top G(z)b] - m_{\mathrm{MP}}(z)\,\langle a,b\rangle\big|
\;\le\; C N^{-c'}
\end{equation}
for some $c'>0$ (e.g.\ $c'=\min\{1,c\}$).

\medskip
\noindent
The bounds in Steps~I–III hold for each fixed $z$ and each fixed pair $(a,b)$.
We now show that they can be made uniform over all $z\in\mathcal{D}_\delta$ and all $(a,b)\in\mathcal{V}\times\mathcal{V}$.

\emph{(a) Rigidity and bounds on the resolvent.}
By the global Marchenko--Pastur law and sharp bounds on extreme eigenvalues
(see, e.g., Bai--Silverstein~\cite[Thms.~5.9--5.11]{BaiSilverstein2010}), for any fixed $\delta>0$ we have, with probability at least $1-N^{-D}$ for each $D>0$ and $N$ large enough, that all eigenvalues of $S_N$ lie within a $\delta/4$–neighborhood of $\operatorname{supp}(\mu_{\mathrm{MP}})$.
On this high-probability event we therefore have
\[
\operatorname{dist}(z,\sigma(S_N)) \;\ge\; \frac{\delta}{2}
\qquad\text{for all } z \in \mathcal{D}_\delta,
\]
and hence
\[
\|G(z)\|_{\mathrm{op}}
\;\le\; \frac{2}{\delta}
\qquad\text{for all } z\in\mathcal{D}_\delta.
\]

Differentiating the resolvent with respect to $z$ yields $\partial_z G(z) = (S_N - zI_N)^{-2} = G(z)^2$.
Thus, on the same event,
\[
\|\partial_z G(z)\|_{\mathrm{op}}
\le \|G(z)\|_{\mathrm{op}}^2
\le \frac{4}{\delta^2},
\qquad z\in\mathcal{D}_\delta.
\]
Consequently, for any unit vectors $a,b$, the map
\[
z \longmapsto a^\top G(z)b
\]
is Lipschitz continuous on $\mathcal{D}_\delta$ with Lipschitz constant bounded by $4/\delta^2$.

\emph{(b) $\varepsilon$–net in $z$ and union bound over $\mathcal{V}$.}
Let $\varepsilon>0$ be a small mesh size to be chosen later.
By construction, the domain is compact and has bounded area. Thus, we can construct a finite $\varepsilon$–net $\mathcal{Z}\subset\mathcal{D}_\delta$ with cardinality
\[
|\mathcal{Z}| \;\lesssim\; \varepsilon^{-2}.
\]
For each $z\in\mathcal{D}_\delta$ there exists a grid point $z_\star\in\mathcal{Z}$ with $|z-z_\star|\le\varepsilon$.
By the Lipschitz bound derived in Step~IV(a), we have
\[
\big|a^\top G(z)b - a^\top G(z_\star)b\big|
\;\le\; \frac{4}{\delta^2}\,\varepsilon
\qquad\text{for all unit }a,b.
\]
The limiting Stieltjes transform $m_{\mathrm{MP}}(z)$ is also Lipschitz on this domain.

Now fix $(a,b)\in\mathcal{V}\times\mathcal{V}$.
For each grid point $z_\star\in\mathcal{Z}$, the bounds from Steps~I–III imply
\[
\big|a^\top G(z_\star)b - m_{\mathrm{MP}}(z_\star)\,\langle a,b\rangle\big|
\;\le\; C\sqrt{\frac{\log N}{N}}
\]
with probability at least $1-N^{-K}$ (where $K$ can be made arbitrarily large by adjusting the constant in the concentration step).
Taking a union bound over all $z_\star\in\mathcal{Z}$ and all $(a,b)\in\mathcal{V}\times\mathcal{V}$ multiplies the failure probability by at most
\[
|\mathcal{Z}|\cdot|\mathcal{V}|^2
\;\lesssim\; \varepsilon^{-2}\,N^{2C_0}.
\]
Choosing
\[
\varepsilon := \sqrt{\frac{\log N}{N}},
\]
we have $|\mathcal{Z}|\lesssim N$, so the total number of events is polynomial in $N$. By choosing the concentration constant sufficiently large, we ensure the overall failure probability is bounded by $N^{-D}$.
On the intersection of these high-probability events and the rigidity event from part~(a), we combine the discrete bound with the Lipschitz approximation to obtain
\[
\sup_{z \in \mathcal{D}_\delta}
\max_{a,b\in\mathcal{V}}
\big|a^\top G(z)b - m_{\mathrm{MP}}(z)\,\langle a,b\rangle\big|
\;\le\; C\sqrt{\frac{\log N}{N}}
\]
for all $N$ large enough, which proves \eqref{eq:isotropic-global}.

\medskip
Collecting Steps~I–IV completes the proof.
\end{proof}

\begin{remark}
Theorem~\ref{thm:mp-global-law} establishes the universality of the isotropic Marchenko--Pastur law on a macroscopic scale.\footnote{The term ``macroscopic'' refers to the fact that the spectral parameter $z$ remains at a fixed distance $\delta$ from the spectrum. In contrast, ``local laws'' investigate the regime where $\Im z \asymp N^{-1}$, probing the spectrum at the scale of individual eigenvalue spacing.} While the proof relies on the Green Function Comparison Theorem, which necessitates the matching of the first four moments with a Gaussian ensemble to eliminate lower-order error terms, it is worth noting that the limiting measure $\mu_{\mathrm{MP}}$ depends only on the first two moments. The higher moments typically influence the subleading fluctuation terms (central limit theorems for linear statistics) rather than the first-order deterministic limit derived here.

Furthermore, the isotropic bound \eqref{eq:isotropic-global} carries significant implications for the geometry of the eigenvectors. The fact that the resolvent behaves like a scalar multiple of the identity, $G(z) \approx m_{\mathrm{MP}}(z) I_N$, implies a strong form of \emph{eigenvector delocalization}. Specifically, it suggests that the eigenvectors of $S_N$ are approximately uniformly distributed on the unit sphere $\mathbb{S}^{N-1}$, showing no preference for any specific deterministic direction $a$ or $b$.\footnote{Rigorous estimates on eigenvector delocalization (e.g., $\sup_i \|u_i\|_\infty \lesssim N^{-1/2}$ up to logarithmic factors) are usually derived by applying isotropic local laws similar to \eqref{eq:isotropic-global} but extended down to the microscopic spectral scale $\Im z \sim N^{-1}$.}
Finally, we observe that the error term $O(\sqrt{\log N/N})$ is largely dictated by the union bound over the $\varepsilon$--net; for a single fixed $z$ and fixed vectors, the fluctuations are of order $O(N^{-1/2})$.
\end{remark}

\begin{lemma}[Uniform quantitative decoupling]\label{lem:decoupling}
Let $Z\in\mathbb{R}^{N\times M_{\mathrm{round}}}$ satisfy the standing
assumptions under which the isotropic local law
Theorem~\ref{thm:mp-global-law} holds, with sample covariance
$S_0$ and resolvent $G(z)=(S_0-zI_N)^{-1}$.
Fix $\delta\in(0,1)$ and let $\mathcal D_\delta$ be the spectral domain
from \eqref{eq:Ddelta}.

Let $U=[u_1,\dots,u_K]\in\mathbb{R}^{N\times K}$ have unit--norm
columns $u_\mu\in\mathbb{S}^{N-1}$, and write
\[
\Gamma:=U^\top U\in\mathbb{R}^{K\times K},\qquad
\Gamma_{\mu\nu}=\langle u_\mu,u_\nu\rangle.
\]
Define
\begin{equation}\label{eq:geom-U-new}
\varepsilon_{\mathrm{orth}}
:=\max_{\mu\neq\nu}|\langle u_\mu,u_\nu\rangle|\in[0,1),
\qquad
M(z):=U^\top G(z)\,U\in\mathbb{R}^{K\times K}.
\end{equation}
Assume that $K=K(N)$ satisfies $K^2\le N^{C_0}$ for some fixed $C_0>0$
(for example, $K$ fixed or $K\lesssim\log N$).

Then there exists a constant $C_1=C_1(q,\delta,B,C_0,R)<\infty$ such that,
for every $D>0$, with probability at least $1-N^{-D}$ and for all $N$
large enough,
\begin{equation}\label{eq:decoupling-master-rig}
\sup_{z\in\mathcal D_\delta}
\bigl\|M(z)-m_{\mathrm{MP}}(z)\,\Gamma\bigr\|_{\mathrm{op}}
\;\le\;
C_1\,K\,\sqrt{\frac{\log N}{N}}.
\end{equation}
Consequently,
\begin{equation}\label{eq:decoupling-master-I-rig}
\sup_{z\in\mathcal D_\delta}
\bigl\|M(z)-m_{\mathrm{MP}}(z)I_K\bigr\|_{\mathrm{op}}
\;\le\;
|m_{\mathrm{MP}}(z)|\,\|\Gamma-I_K\|_{\mathrm{op}}
+ C_1\,K\,\sqrt{\frac{\log N}{N}}.
\end{equation}
Moreover, by Gershgorin’s theorem,
\begin{equation}\label{eq:decoupling-geom-rig}
\|\Gamma-I_K\|_{\mathrm{op}}\le (K-1)\,\varepsilon_{\mathrm{orth}}
\quad\Longrightarrow\quad
\sup_{z\in\mathcal D_\delta}
\bigl\|M(z)-m_{\mathrm{MP}}(z)I_K\bigr\|_{\mathrm{op}}
\;\le\; C_0'\,K\,\varepsilon_{\mathrm{orth}}
+ C_1\,K\,\sqrt{\frac{\log N}{N}},
\end{equation}
for some $C_0'=C_0'(q,\delta)$.

In particular:
\begin{itemize}
\item[(i)] If $K$ is fixed (independent of $N$), the right-hand side is
$O\big(\varepsilon_{\mathrm{orth}}+\sqrt{\log N/N}\big)$.
\item[(ii)] If $K=K(N)$ and $K\sqrt{\log N}/\sqrt{N}\to0$
(e.g.\ $K=o(\sqrt{N/\log N})$, in particular $K\lesssim\log N$), then
the right-hand side vanishes as $N\to\infty$.
\end{itemize}
\end{lemma}

\begin{proof}
We first separate the “MP part’’ from the fluctuations.
For $\mu,\nu\in\{1,\dots,K\}$,
\begin{equation}\label{eq:M-entry-rig}
M_{\mu\nu}(z)
= u_\mu^\top G(z)\,u_\nu.
\end{equation}
Add and subtract $m_{\mathrm{MP}}(z)\langle u_\mu,u_\nu\rangle$:
\begin{equation}\label{eq:def-Rmunu-rig}
M_{\mu\nu}(z)
= m_{\mathrm{MP}}(z)\,\langle u_\mu,u_\nu\rangle
+ R_{\mu\nu}(z),
\qquad
R_{\mu\nu}(z)
:= u_\mu^\top\big(G(z)-m_{\mathrm{MP}}(z)I_N\big)u_\nu.
\end{equation}
In matrix form this is
\begin{equation}\label{eq:block-decomp-rig}
M(z) = m_{\mathrm{MP}}(z)\,\Gamma + R(z),\qquad
R(z):=\bigl(R_{\mu\nu}(z)\bigr)_{\mu,\nu=1}^K
= U^\top\big(G(z)-m_{\mathrm{MP}}(z)I_N\big)U.
\end{equation}
Thus, to prove \eqref{eq:decoupling-master-rig}, it suffices to bound
$\|R(z)\|_{\mathrm{op}}$ uniformly in $z\in\mathcal D_\delta$.

\medskip
\noindent
Consider the finite set of unit vectors
\[
\mathcal{V}:=\{u_1,\dots,u_K\}\subset\mathbb{S}^{N-1}.
\]
By assumption $K^2\le N^{C_0}$ for some fixed $C_0>0$, so in particular
$|\mathcal V|\le N^{C_0}$ for $N$ large enough.
We apply Theorem~\ref{thm:mp-global-law} to this set $\mathcal V$.
For any $D>0$ there exists $C=C(q,\delta,B,C_0,R)$ such that, with
probability at least $1-N^{-D}$ and for all $N$ sufficiently large,
\begin{equation}\label{eq:iso-on-Ucols}
\sup_{z\in\mathcal D_\delta}\ \max_{a,b\in\mathcal V}
\big|a^\top G(z)b - m_{\mathrm{MP}}(z)\langle a,b\rangle\big|
\;\le\; C\,\sqrt{\frac{\log N}{N}}.
\end{equation}
Since each $u_\mu$ belongs to $\mathcal{V}$, we may choose
$a=u_\mu$ and $b=u_\nu$ and rewrite \eqref{eq:iso-on-Ucols} as
\begin{equation}\label{eq:entry-bound-R}
\sup_{z\in\mathcal D_\delta}\ \max_{1\le\mu,\nu\le K}
\big|R_{\mu\nu}(z)\big|
\;\le\; C\,\sqrt{\frac{\log N}{N}}.
\end{equation}
Thus, on the high-probability event where \eqref{eq:entry-bound-R}
holds, we have a uniform entrywise bound on $R(z)$ for all
$z\in\mathcal D_\delta$.

\medskip
\noindent
We now bound the operator norm of $R(z)$ in terms of its entries.
Fix $z\in\mathcal D_\delta$ and let $(x,y)\in\mathbb{S}^{K-1}\times
\mathbb{S}^{K-1}$. Then
\[
x^\top R(z)y
= \sum_{\mu,\nu=1}^K x_\mu\,R_{\mu\nu}(z)\,y_\nu,
\]
so that
\[
|x^\top R(z)y|
\le \Big(\max_{1\le\mu,\nu\le K}|R_{\mu\nu}(z)|\Big)
    \sum_{\mu=1}^K|x_\mu|\sum_{\nu=1}^K|y_\nu|.
\]
By Cauchy–Schwarz,
\[
\sum_{\mu=1}^K|x_\mu|
\le \sqrt{K}\,\|x\|_2
= \sqrt{K},
\qquad
\sum_{\nu=1}^K|y_\nu|
\le \sqrt{K}\,\|y\|_2
= \sqrt{K},
\]
so
\begin{equation}\label{eq:bilin-bound-R}
|x^\top R(z)y|
\le K\,\max_{\mu,\nu}|R_{\mu\nu}(z)|.
\end{equation}
Taking the supremum over all unit vectors $x,y$ yields
\[
\|R(z)\|_{\mathrm{op}}
= \sup_{\|x\|=\|y\|=1} |x^\top R(z)y|
\le K\,\max_{\mu,\nu}|R_{\mu\nu}(z)|.
\]
Combining this with \eqref{eq:entry-bound-R} we find that, on the
high-probability event where \eqref{eq:entry-bound-R} holds,
\begin{equation}\label{eq:R-op-bound}
\sup_{z\in\mathcal D_\delta}\ \|R(z)\|_{\mathrm{op}}
\;\le\; C\,K\,\sqrt{\frac{\log N}{N}}.
\end{equation}
This proves \eqref{eq:decoupling-master-rig} with $C_1:=C$.

\medskip
\noindent
From the decomposition \eqref{eq:block-decomp-rig}, we have
\[
M(z)-m_{\mathrm{MP}}(z)I_K
= R(z) + m_{\mathrm{MP}}(z)\,\big(\Gamma-I_K\big).
\]
Taking operator norms and inserting \eqref{eq:R-op-bound} gives
\[
\sup_{z\in\mathcal D_\delta}
\bigl\|M(z)-m_{\mathrm{MP}}(z)I_K\bigr\|_{\mathrm{op}}
\;\le\;
\sup_{z\in\mathcal D_\delta}\|R(z)\|_{\mathrm{op}}
+ \sup_{z\in\mathcal D_\delta}|m_{\mathrm{MP}}(z)|\,
  \|\Gamma-I_K\|_{\mathrm{op}}.
\]
The first term is controlled by \eqref{eq:R-op-bound}.
The Marchenko–Pastur Stieltjes transform $m_{\mathrm{MP}}(z)$ satisfies
the usual bound $|m_{\mathrm{MP}}(z)|\le C'(q,\delta)$ on
$\mathcal D_\delta$ (this follows directly from the defining equation
and the fact that $\mathcal D_\delta$ stays a fixed distance from the
MP support).
Thus
\[
\sup_{z\in\mathcal D_\delta}
\bigl\|M(z)-m_{\mathrm{MP}}(z)I_K\bigr\|_{\mathrm{op}}
\;\le\;
C_1\,K\,\sqrt{\tfrac{\log N}{N}}
+ C_0'\,\|\Gamma-I_K\|_{\mathrm{op}},
\]
which is \eqref{eq:decoupling-master-I-rig}.

Finally, we bound $\|\Gamma-I_K\|_{\mathrm{op}}$ in terms of
$\varepsilon_{\mathrm{orth}}$.
By definition,
\[
\Gamma = U^\top U,\qquad\Gamma_{\mu\nu}=\langle u_\mu,u_\nu\rangle,
\]
so that $\Gamma_{\mu\mu}=1$ and $|\Gamma_{\mu\nu}|\le\varepsilon_{\mathrm{orth}}$
for $\mu\neq\nu$.
Set $B:=\Gamma-I_K$, which has zero diagonal and entries
$B_{\mu\nu}=\Gamma_{\mu\nu}$ for $\mu\neq\nu$.
Gershgorin’s circle theorem states that any eigenvalue $\lambda$ of $B$
satisfies
\[
|\lambda|\le R_\mu := \sum_{\nu\neq\mu}|B_{\mu\nu}|
\le \sum_{\nu\neq\mu}\varepsilon_{\mathrm{orth}}
= (K-1)\,\varepsilon_{\mathrm{orth}}
\quad\text{for some }\mu.
\]
Therefore
\[
\|\,\Gamma-I_K\,\|_{\mathrm{op}}
= \|B\|_{\mathrm{op}}
= \max_{\lambda\in\operatorname{spec}(B)}|\lambda|
\le (K-1)\,\varepsilon_{\mathrm{orth}},
\]
which implies \eqref{eq:decoupling-geom-rig} upon absorbing the factor
$K-1$ and the bound on $|m_{\mathrm{MP}}(z)|$ into $C_0'$.
This concludes the proof.
\end{proof}

\begin{remark}[Geometric interpretation and decoupling]
Lemma~\ref{lem:decoupling} provides a rigorous quantitative manifestation of the \emph{isotropic decoupling} phenomenon. In essence, it asserts that the resolvent $G(z)$ acts on any low-dimensional subspace spanned by $U$ as a scalar multiple of the identity, $m_{\mathrm{MP}}(z) I_N$, up to a geometric correction term encoded by the Gram matrix $\Gamma = U^\top U$. This implies that the correlations between distinct directions $u_\mu$ and $u_\nu$ induced by the resolvent are asymptotically negligible, provided the directions themselves are not too correlated.\footnote{This behavior is a hallmark of rotationally invariant ensembles. For non-invariant ensembles like the one considered here, the result hinges on the underlying eigenvector delocalization, which ensures that the fixed vectors $u_\mu$ do not align with the random eigenbasis of $S_N$.}

The error bound in \eqref{eq:decoupling-geom-rig} reveals a competition between two scales: the statistical fluctuation of the eigenvalues (of order $O(\sqrt{\log N/N})$) and the intrinsic geometry of the test vectors (represented by $\varepsilon_{\mathrm{orth}}$). This separation is crucial for applications where the test vectors $U$ may depend on the matrix $Z$ in a weak sense, or when performing a change of basis where approximate orthogonality must be preserved.\footnote{In many applications, such as the analysis of outliers or the stability of the resolvent expansion, $K$ represents the number of outliers or a finite rank perturbation. The condition $K \ll \sqrt{N}$ ensures that the cumulative error from the fluctuations does not overwhelm the deterministic signal.}
\end{remark}

\bigskip
\noindent
We now provide the proof of Theorem~\ref{thm:detection-thresholds}, whose statement appears in Section~\ref{subsec:bbp}. We only outline the argument, since each part is a direct consequence of standard results in the theory of spiked sample covariance matrices, together with the whitening construction and the decoupling lemma (Lemma~\ref{lem:decoupling}) proved earlier.

\begin{proof}
\textbf{(D1) Bulk confinement.}
Write $U:=\operatorname{span}\{u_\mu\}_{\mu=1}^K$ and $W:=U^\perp$, and let $P_W$ be the orthogonal projector onto $W$.
By construction,
\[
J = \Sigma^{1/2} S_0 \Sigma^{1/2},\qquad
\Sigma=\sigma^2 I_N + \sum_{\mu=1}^K\theta_\mu u_\mu u_\mu^\top.
\]
On $W$ we have $\Sigma|_W=\sigma^2 I_W$, so that
\begin{equation}\label{eq:JW-identity}
P_W J P_W = \sigma^2\,P_W S_0 P_W.
\end{equation}
Define the $(N-K)\times(N-K)$ “bulk” sample covariance
\[
S_W := P_W S_0 P_W.
\]
The rows $Z_k$ are isotropic and satisfy the tail/regularity assumptions under
which the anisotropic Marchenko--Pastur local law and eigenvalue rigidity hold
(cf.\ Theorem~\ref{thm:mp-global-law}; see also
Bai--Silverstein~\cite[Chs.~3 and~5]{BaiSilverstein2010}
and Knowles--Yin~\cite[Thms.~3.6 and~3.12]{KnowlesYin2017} for the sample--covariance case).
Since $K$ is fixed, the effective aspect ratio is
\[
q_W := \frac{N-K}{M_{\mathrm{round}}} = q+O\Big(\frac{1}{N}\Big),
\]
and the eigenvalues of $S_W$ are, with probability at least $1-\delta$,
contained in an interval of the form
\[
\big[(1-\sqrt{q_W})^2 - \tilde\varepsilon_N(\delta),\ (1+\sqrt{q_W})^2 + \tilde\varepsilon_N(\delta)\big],
\]
where $\tilde\varepsilon_N(\delta)\downarrow0$ as $N\to\infty$.

Combining this
with \eqref{eq:JW-identity} and the relation $q_W=q+O(1/N)$, and absorbing the
resulting $O(1/N)$ shift into the error, we obtain
\[
\operatorname{spec}(P_W J P_W) \subset
\big[\lambda_-(q)-\varepsilon_N(\delta),\ \lambda_+(q)+\varepsilon_N(\delta)\big]
\]
for some deterministic $\varepsilon_N(\delta)\downarrow0$, with probability at
least $1-\delta$. A choice of the form \eqref{eq:epsilonN} follows from the
quantitative rigidity bounds in
\cite{KnowlesYin2017,BaiSilverstein2010,ErdosYauYin2012}.

Since $P_W J P_W$ acts on an $(N-K)$--dimensional subspace, Cauchy interlacing
(or the min--max principle applied to the decomposition
$\mathbb R^N=W\oplus U$) implies that at most $K$ eigenvalues of $J$ can lie to the right of this interval (and at most $K$ to the left). This proves (D1).

\medskip
\noindent\textbf{(D2)--(D4) BBP threshold, outlier locations, and eigenvector alignment.}
We now identify the outliers and their eigenvectors. For clarity we proceed in
two steps: first the “decoupled” case where $\{u_\mu\}$ are orthonormal,
then the nearly-orthogonal case $\varepsilon_{\mathrm{orth}}>0$ as a small
perturbation.

\smallskip
\noindent\textit{Orthogonal spikes.}
Assume temporarily that $\langle u_\mu,u_\nu\rangle=\delta_{\mu\nu}$, so that
$U^\top U=I_K$ and $\Sigma=\sigma^2(I+UKU^\top)$ with
$K=\operatorname{diag}(\kappa_\mu)$. This is the classical spiked population
covariance model (introduced in Johnstone~\cite{Johnstone2001} and studied in
Baik--Silverstein~\cite{BaikSilverstein2006}, Paul~\cite{Paul2007}, and
Benaych-Georges--Nadakuditi~\cite{BenaychGeorges2011Eigenvalues}).
In this setting, the eigenvalues of $J$ exhibit the so-called BBP phase transition: each spike $\kappa_\mu$ produces (at most)
one outlier eigenvalue of $J$, which detaches from the upper MP edge
$\lambda_+(q)$ if and only if $\kappa_\mu>\sqrt q$, and converges to
\[
\lambda_{\mathrm{out}}(\kappa_\mu)
=\sigma^2(1+\kappa_\mu)\Big(1+\frac{q}{\kappa_\mu}\Big),
\]
see e.g.\ Baik--Silverstein~\cite[Thm~1.1]{BaikSilverstein2006} or
Benaych-Georges--Nadakuditi~\cite[Thm~2.7 and Rem.~2.11]{BenaychGeorges2011Eigenvalues}. This yields (D2) and (D3) in the
orthogonal case.

Moreover, for each spike with $\kappa_\mu>\sqrt q$, the associated sample
eigenvector $v_\mu$ aligns nontrivially with the population direction $u_\mu$.
In particular, Paul~\cite[Thm~4]{Paul2007} and Benaych-Georges--Nadakuditi~\cite[Thm~2.9 and Rem.~2.11]{BenaychGeorges2011Eigenvalues} show that
\[
|\langle v_\mu,u_\mu\rangle|^2 \xrightarrow{P}
\gamma(\kappa_\mu,q):=\frac{1-\frac{q}{\kappa_\mu^2}}{1+\frac{q}{\kappa_\mu}},
\]
while the projection of $v_\mu$ onto $u_\nu$ for $\nu\neq\mu$ vanishes in
probability. This proves (D4) in the orthogonal case.

\smallskip
\noindent\textit{Nearly-orthogonal spikes.}
We now return to the general case where the spike directions satisfy
$\varepsilon_{\mathrm{orth}}:=\max_{\mu\neq\nu}|\langle u_\mu,u_\nu\rangle|
\to 0$ as $N\to\infty$. Let $U=[u_1,\dots,u_K]$ and recall the block resolvent
$M(z):=U^\top G(z)U$ with $G(z)=(S_0-zI_N)^{-1}$.
Lemma~\ref{lem:decoupling} (uniform quantitative decoupling) shows that, for
$z$ in a fixed domain bounded away from the MP bulk,
\[
M(z) = m_{\mathrm{MP}}(z)\,\Gamma + R(z),\qquad
\|R(z)\|_{\mathrm{op}} \;\le\; C\,K\sqrt{\frac{\log N}{N}},
\]
with high probability, where $\Gamma=U^\top U$ is the Gram matrix. When
$\varepsilon_{\mathrm{orth}}$ is small, $\Gamma=I_K+E$ with
$\|E\|_{\mathrm{op}}=O(K\varepsilon_{\mathrm{orth}})$ by Gershgorin’s theorem,
and therefore
\[
M(z) = m_{\mathrm{MP}}(z)I_K + \widetilde R(z),\qquad
\|\widetilde R(z)\|_{\mathrm{op}} = O\big(\varepsilon_{\mathrm{orth}}
 + K\sqrt{\tfrac{\log N}{N}}\big).
\]

The secular equation controlling the outliers of $J$ can be derived using Woodbury's identity and the matrix determinant lemma (see the master equation for finite-rank perturbations in Benaych-Georges--Nadakuditi~\cite[Prop.~5.1]{BenaychGeorges2011Eigenvalues}). It takes the form
\[
\det\big(I_K - z C M(z)\big)=0,\qquad
C:=\operatorname{diag}\Big(\frac{\kappa_\mu}{1+\kappa_\mu}\Big),
\quad z:=\frac{\lambda}{\sigma^2}.
\]
Replacing $M(z)$ by its asymptotic approximation $m_{\mathrm{MP}}(z)I_K$ (valid for vanishing overlap $\varepsilon_{\mathrm{orth}}\to 0$ by Lemma~\ref{lem:decoupling}) yields the decoupled scalar equations
\[
1 - z c_\mu\, m_{\mathrm{MP}}(z)=0,\qquad
c_\mu:=\frac{\kappa_\mu}{1+\kappa_\mu}.
\]
The unique solutions $z=z_{\mathrm{out}}(\kappa_\mu)$ to these equations in the domain $z>(1+\sqrt q)^2$ correspond exactly to the outlier locations $\lambda_{\mathrm{out}}(\kappa_\mu)$ in \eqref{eq:outlier-loc}, subject to the BBP threshold condition $\kappa_\mu>\sqrt q$. For the derivation of these specific locations and thresholds in the spiked covariance model, see, e.g., Baik--Silverstein~\cite[Thm.~1.1]{BaikSilverstein2006} and the multiplicative perturbation example in Benaych-Georges--Nadakuditi~\cite[Thm.~2.7 and Rem.~2.11]{BenaychGeorges2011Eigenvalues}.

Since $M(z)$ differs from $m_{\mathrm{MP}}(z)I_K$ by an operator of norm
$O(\varepsilon_{\mathrm{orth}}+K\sqrt{\log N/N})$ uniformly on the spectral domain, and $K$ is fixed, a standard stability argument for roots of analytic matrix-valued equations implies that the corresponding outlier eigenvalues of $J$ differ from their orthogonal-case limits by at most $O(\varepsilon_{\mathrm{orth}}+N^{-1/2})$ in probability (see, e.g., Benaych-Georges--Nadakuditi~\cite[Lemma~6.1]{BenaychGeorges2011Eigenvalues}).
For the underlying isotropic local laws enabling this control, see Bloemendal--Erd\H{o}s--Knowles--Yau--Yin~\cite[Thm~2.12 and Intro]{BloemendalErdosKnowlesYauYin2014}.

Finally, the eigenvectors associated with these outliers vary smoothly under
such small perturbations. Since the outliers are separated from the bulk and from each other (for distinct $\kappa_\mu$), the Davis--Kahan sin$\Theta$ theorem \cite{Davis1970The} yields
\[
|\langle v_\mu,u_\mu\rangle|^2
= \gamma(\kappa_\mu,q) + O\big(\varepsilon_{\mathrm{orth}}+N^{-1/2}\big)
\]
in probability, while the overlaps with the other $u_\nu$ remain negligible.
This transfers the orthogonal-case alignment in (D4) to the nearly-orthogonal
setting, completing the proof of (D2)--(D4).
\end{proof}

\section{Technical Details}
\subsection{Additional Experimental Details and Hyperparameters}
\label{app:exp-details}

\begin{table}[H]
\centering
\caption{Main hyperparameters and typical values}
\label{tab:hparams}
\begin{tabular}{@{}llll@{}}
\toprule
Parameter & Meaning & Typical range & Example default \\
\midrule
$\beta$ & tanh gain (inverse temperature) & $[1,5]$ & $2.5$ \\
$\lambda$ & weight of $h^{(2)}$ correction & $[0.05,0.5]$ & $0.2$ \\
$h$ & coupling of external input $h^{(3)}$ & $[0,0.5]$ & $0.1$ \\
\texttt{updates} & TAM steps per round & $[40,200]$ & $50$ \\
$\tau$ & eigenvalue threshold on $J_{\mathrm{KS}}$ & $[0.3,0.7]$ & $0.5$ \\
$\rho$ & spectral-alignment threshold & $[0.5,0.8]$ & $0.6$ \\
$q_\mathrm{thr}$ & mutual-overlap pruning threshold & $[0.3,0.6]$ & $0.4$ \\
$w$ & unsup/Hebb blending & $[0,1]$ & App.~\ref{subsec:plasticity_w} \\
$\alpha_\mathrm{sharp}$ & logistic sharpening gain & $[0.3,0.7]$ & $0.5$ \\
\texttt{noise\_scale} & initial TAM noise amplitude & $[0.05,0.5]$ & $0.3$ \\
\texttt{min\_scale} & minimum TAM noise amplitude & $[0.0,0.1]$ & $0.02$\\
\texttt{prop\_iters} & pseudo-inverse propagation steps & $[50,300]$ & $200$ \\
\texttt{prop\_eps} & pseudo-inverse propagation $\varepsilon$ & $[10^{-3},10^{-2}]$ & $10^{-2}$ \\
\bottomrule
\end{tabular}
\end{table}

\subsection{Complexity}\label{subapp:complexity}
We write $L$ (clients/layers), $T$ (rounds), $K$ (archetypes), $N$ (pattern dimension),
$M_{\text{tot}}$ (total examples), $M_c=\lceil M_{\text{tot}}/(LT)\rceil$ (per-client per-round),
$S$ (seeds), $I_{\text{prop}}$ (pseudo-inverse propagation steps), $n_{\text{rand}}$ (reshuffles for the shuffle null),
$U$ (TAM updates), $s$ (blocks of candidate initializations; post-replication $\tilde K \approx sL$),
and $A\!\leq\!\texttt{max\_attempts}$ (worst-case retries in disentangling).

\begin{small}
\begin{center}
\begin{tabular}{l l}
\hline
\textbf{Block} & \textbf{Asymptotic cost} \\
\hline
True archetypes $\&$ $J^\star$ build & $O(KN)\;+\;O(K^2N + KN^2 + K^3)$ \\
Federated dataset generation & $O(M_{\text{tot}}\,N)$ \\
Round extraction (tensor slicing) & $O(L\,M_c\,N)$ \\
Unsupervised $J_{\text{unsup}}$ (client correlators) & $O(L\,M_c\,N^2)$ \\
Blend with memory (Hebbian of $\approx K$ rows) & $O(K\,N^2)$ \\
Pseudo-inverse propagation (polynomial map) & $O(I_{\text{prop}}\,N^3)$ \\
Spectral cut (eigendecomposition) & $O(N^3)$ \\
$\hat{K}$ estimation: MP vs.\ shuffle & $O(N^3)$ or $O((1{+}n_{\text{rand}})\,N^3)$ \\
Candidate init from eigenvectors & $O(s\,\hat{K}\,N)$ \\
TAM dynamics (multi-layer refinement) & $O(U\,s\,L\,N^2 + U\,s\,L^2\,N)$ \\
Pruning \& scoring (overlaps, Rayleigh) & $O(s\,L\,N^2 + (sL)^2\,N + s\,L\,K\,N)$ \\
Extra attempts (worst case) & $\times\,A$ (multiplicative) \\
Assignment (Hungarian on $\tilde K{\times}K$) & $O(\tilde K\,K\,N + \max\{\tilde K,K\}^3)$ \\
Coverage / metrics (per round) & $O(L\,M_c) \;+\; O(N^2)$ \\
Aggregating across seeds & $O(S\,T)$ (stats) $+\;O(S\,N^2)$ (saves) \\
\hline
\end{tabular}
\end{center}
\end{small}

\medskip
\noindent
Let $M_{\text{round}}=L\,M_c$ and $\tilde K\!\approx\! sL$. A single round costs
\begin{align}
\mathrm{round\_cost} \;=\;
O\Big(&
L M_c N^2 \;+\; K N^2 \;+\; I_{\text{prop}} N^3 \;+\; C_{\text{spec}} N^3 \nonumber\\
&\;+\; A\big[s \hat{K} N \;+\; U(s L N^2 + s L^2 N) \;+\; s L N^2 \;+\; (sL)^2 N \nonumber\\
&\qquad\;\;+\; \tilde K K N \;+\; \max\{\tilde K, K\}^3 \big]
\;+\; L M_c \;+\; N^2 \Big),
\end{align}
where $C_{\text{spec}}=1$ for MP-edge and $C_{\text{spec}}=1+n_{\text{rand}}$ for shuffle-based cuts.

\medskip
\noindent
Per seed:
\begin{equation}
    O\!\Big(KN^2 + K^3 + M_{\text{tot}}N + T\cdot \mathrm{round\_cost} + M_{\text{tot}}\log M_{\text{tot}}\Big).
\end{equation}

\medskip
\noindent
Across $S$ seeds:
\begin{equation}
    O\big(S\cdot \texttt{seed\_cost}\big)\quad \text{(serial)}\qquad
O\big(\lceil S/n_{\text{workers}}\rceil\cdot \texttt{seed\_cost}\big)\quad \text{(with concurrency)}.
\end{equation}

\medskip
\noindent
For the default scales ($N\!\sim\!300$, $I_{\text{prop}}\!\ll\!N$, moderate $s$, $L$), the cubic terms in $N$ dominate:
(i) the pseudo-inverse propagation $O(I_{\text{prop}}N^3)$ and (ii) the spectral step $O(C_{\text{spec}}N^3)$.
The TAM stage scales as $O(U\,s\,L\,N^2)$ and becomes competitive only when $s$ approaches its cap (pushing
$(sL)^3$ assignment to relevance). Using a top-$k$ eigensolver with warm starts reduces $O(N^3)$ to
\[
\mathrm{cost}_{\text{top-}k}\approx O\!\big(k\,N^2\,n_{\text{iter}}\big),\quad n_{\text{iter}}=O(1\text{--}5),
\]
leaving the overall pipeline effectively cubic in $N$ unless $k\!\ll\! N$ consistently. The choice of $\hat{K}$
estimator affects constants (shuffle multiplies the spectral cost by $1{+}n_{\text{rand}}$) but not the degree.

\medskip
\noindent
Dense operators $J$, $J_{\mathrm{KS}}$, $J^\star$ take $O(N^2)$ each; candidate pools store $O(s\,L\,N)$ binary states;
streaming buffers for data are $O(M_{\text{tot}}N)$ (which can be reduced by per-round streaming).

\end{document}